\documentclass[10pt]{article}

\usepackage[authoryear,round]{natbib}					
\usepackage{amsmath}    													
\usepackage{graphicx}   													
\usepackage{subfigure}														
\usepackage{amssymb}    													
\usepackage[mathscr]{eucal}	
\usepackage{colortbl}
\usepackage{color}
\usepackage[colorlinks,citecolor=blue,urlcolor=blue]{hyperref}
\usepackage{cancel}																
\usepackage[normalem]{ulem} 																
\usepackage{pstricks}
\usepackage{enumitem}
\usepackage{rotating}
\usepackage{lscape}
\usepackage[paperwidth=8.5in,paperheight=11in,top=1.25in, bottom=1.25in, left=1.00in, right=1.00in]{geometry}
\usepackage{mathtools}														
\mathtoolsset{showonlyrefs=true}									
\newcommand{\be}{\begin{equation}}
\newcommand{\ee}{\end{equation}}
\newcommand{\nee}{\nonumber\end{equation}}
\newcommand{\eel}[1]{\label{#1}\end{equation}}																									
\usepackage{amsthm}																
\theoremstyle{plain}
\newtheorem{theorem}{Theorem}
\newtheorem{proposition}[theorem]{Proposition}	
\newtheorem{corollary}[theorem]{Corollary}	
\theoremstyle{definition}

\newtheorem{remark}[theorem]{Remark}
\theoremstyle{remark}
\numberwithin{equation}{section}	
\numberwithin{theorem}{section}

\definecolor{OrangeRMA}{RGB}{232,81,0}

\definecolor{GreenRMA}{RGB}{51,125,0}

\definecolor{BlueRMA}{RGB}{51,51,153}

\renewcommand{\(}{\left(}				
\renewcommand{\)}{\right)}
\renewcommand{\[}{\left[}
\renewcommand{\]}{\right]}			

\def\E{\mathbb{E}}			
\def\P{\mathbb{P}}	
							
\def\R{\mathbb{R}}										
\def\I{\mathbb{I}}

\def\A{\mathscr{A}}
\def\B{\mathscr{B}}

\def\Ec{\mathscr{E}}
\def\F{\mathscr{F}}

\def\H{\mathscr{H}}
\def\Ic{\mathscr{I}}

\def\Pc{\mathscr{P}}

\def\m{\mathfrak{m}}
\def\s{\mathfrak{s}}

\def\eps{\varepsilon}

\def\om{\omega}
\def\Om{\Omega}
\def\sig{\sigma}

\def\Gam{\Gamma}
\def\lam{\lambda}
\def\del{\delta}

\def\psib{\overline{\psi}}

\def\omh{\widehat{\om}}
\def\Fh{\widehat{F}}

\def\d{\partial}

\begin{document}

\title{
Variance Swaps on Defaultable Assets and Market Implied Time-Changes
}

\author{
Matthew Lorig
\thanks{ORFE Department, Princeton University.  Work partially supported by NSF grant DMS-0739195.}
\and
Oriol Lozano-Carbass\'e
\thanks{
Bendheim Center for Finance, 26 Prospect Avenue, Princeton University, Princeton, NJ  08540, United States.
}
\and
Rafael Mendoza-Arriaga
\thanks{McCombs School of Business, University of Texas at Austin.}
}

\maketitle

\begin{abstract}
We compute the value of a variance swap when the underlying is modeled as a Markov process time changed by a L\'{e}vy subordinator.  In this framework, the underlying may exhibit jumps with a state-dependent L\'{e}vy measure, local stochastic volatility and have a local stochastic default intensity.  Moreover, the L\'{e}vy subordinator that drives the underlying can be obtained directly by observing European call/put prices.  To illustrate our general framework, we provide an explicit formula for the value of a variance swap when the underlying is modeled as (i) a L\'evy subordinated geometric Brownian motion with default and (ii) a L\'evy subordinated Jump-to-default CEV process (see \citet{carr-linetsky-1}).
{In the latter example, we extend} the results of \cite{mendoza-carr-linetsky-1}, by allowing for joint valuation of credit and equity derivatives as well as variance swaps.
\end{abstract}

%
%

\section{Introduction}
\label{sec:intro}
A \emph{variance swap} (VS) is a forward contract written on the realized variance of an underlying $S=\{S_t, t \geq 0\}$.  As is typical in derivatives literature, we define the \emph{realized variance} over the interval $[0,t]$ as $[\log S]_t$, the continuously sampled quadratic variation of the $\log S$.  Thus, at maturity the VS has a payoff (to the long side) of
\begin{align}
[\log S]_t - K_{var} . \label{eq:VS}
\end{align}
The \emph{variance swap rate} $K_{var}=\E \( [\log S]_t \)$ is determined at inception so that initial value of the VS is zero.  Here, $\E$ is the expectation under the risk-neutral pricing measure $\P$.
\par
There are a number of reasons for which one may wish to enter into a VS agreement.  First, a trader who delta-hedges a short position in a European option can limit his exposure to stochastic volatility risk by trading a VS (see, e.g., \cite{carr-schoutens-1}).  Second, a drop in the level of an underlying $S$ is often accompanied by an increase in the volatility of $S$ (the leverage effect).  Thus, a long position in a VS serves as protection against a market crash.  Such is the demand for VSs that, according to \cite{jung-1}, the daily trading volume in equity index VSs reached 45 million USD vega notional in 2006 (vega measures the change in an option's price caused by changes in volatility).  On an annual basis, this corresponds to payments of more than 1 billion USD per percentage point of volatility (see \cite{carr-lee-2}). 
\par
At the level of individual stocks there is yet another source for the leverage effect that is related to the default risk associated with the underlying firm. The \emph{credit-related leverage effect} explains the interaction between market risk (return variance) and credit risk (default arrival). For instance, \cite{cremers-driessen-maenhout-weinbaum-1} show that CDS rates are correlated with both stock option implied volatility levels and the at-the-money slope of implied volatility.
Similar results can be found in \cite{Consigli-1}. \cite{carr-wu-2} study the interaction between the pricing of stock options and the pricing of credit default swaps. Specifically, they regress CDS spreads on stock option implied volatilities for four companies and find $R^2$'s ranging from 36-82\%. More recently, \cite{Chang-Hung-Tsai-1} verify not only that, on average, the returns in the CDS and stock markets are negatively correlated, but also that this correlation is higher (i.e., more negative) for high yield firms than it is for investment grade firms. This suggests that the credit-leverage effect is even stronger for companies with lower credit ratings.  In summary, when the probability of default of a firm increases, its stock price tends to lose value while the option-implied volatility increases. In such a situation, it is common for investors to look for credit protection by taking long positions in deep out-of-the-money puts (see, e.g., \cite{Angelos-1} and \cite{carr-wu-2,carr-wu-3}). However, there is a side-effect in this hedging strategy. Though the investor's 
 objective is to minimize credit exposure, this strategy will increase his exposure to volatility due to the credit-equity leverage effect. This fact could also be seen from the turmoil experienced in the variance swap market when the stock market acutely dropped during the credit crisis of 2008-9 (see, e.g., \cite{Filipovic-Gourier-Mancini-1}).  Hence, to summarize, in the pricing of variance swaps written on an individual stock, it is especially important to take into account default risk (and the credit-related leverage effect).
\par
In his seminal paper, \cite{neuberger-1} showed that, when the underlying $S$ is modeled by a process with \emph{continuous} sample paths, $dS_t = \sig_t S_t dW_t$, (for simplicity, assume the risk-free rate of interest is zero) the fair value of $K_{var}$ is given by a European-style $\log$ contract $K_{var} =	\E[ -2 \log (S_t/S_0)]$.  Later, \cite{carr-madan-6} showed that the $\log$ contract (or in fact, any European-style derivative with a twice differentiable payoff) can be synthesized from a linear combination of calls and puts.\footnote{In fact, the replication result of \cite{carr-madan-6} is independent of the continuity assumption of the price process.}
Thus, when $t$-expiry calls and/or puts are available at every $K \in (0,\infty)$, the fair value of $K_{var}$ is uniquely determined by the implied volatility smile.  Remarkably, this result is independent of any assumption about the volatility process $\sig = \{ \sig_t , t \geq 0 \}$.  However, this result does rely on the continuity of sample paths of $S$.  The events of the recent (and ongoing) financial crisis underscore the need to include jumps in the underlying price process $S$.  The question naturally arises then: what is the fair value of $K_{var}$ when the price process $S$ is allowed to have discontinuous sample paths?
\par
One possible answer to this question is given by \cite{carr-lee-wu-1}, who show that, when $\log S$ is modeled as a L\'{e}vy process time-changed by an absolutely continuous time-change, the fair value of $K_{var}$ is equal to a multiplier $Q_X$ times the value of a log contract $K_{var}= Q_X \E \[ - \log \( S_t/S_0 \) \]$.  Interestingly, the multiplier $Q_X$ does \emph{not} depend on the time-change process.  While the value of $K_{var}$ \emph{does} depend on the time-change through the $\log$ contract, since the $\log$ contract can be synthesized by a linear combination of calls, when one has full knowledge of the volatility smile, the value of $K_{var}$ can be determined in the framework of \cite{carr-lee-wu-1} without any knowledge of the time-change process.
\par
This last point cannot be emphasized enough.  A parametric model for the time-change process would leave one open to model risk,  as any misspecification of the time-change parameters would (in general) result in erroneous values of $K_{var}$.  By using knowledge of the volatility smile to construct the value of the $\log$ contract, \cite{carr-lee-wu-1} circumvent the need to parametrically model the time-change process.  Thus, the risk of model misspecification is greatly reduced (though, some model risk still exists, since the multiplier $Q_X$ depends on the choice of a specific L\'{e}vy process).
An alternative and quite interesting approach to the joint pricing of volatility derivatives and index options in the presence of jumps is provided in \cite{cont-kokholm-1}.
\par
In this paper, we consider the class of L\'{e}vy subordinated diffusion processes described in \cite{mendoza-carr-linetsky-1}.  This class of models, like the class considered by \cite{carr-lee-wu-1}, allows for the underlying $S$ to experience both jumps and stochastic volatility.  However, there are a few important differences between these two frameworks.  In \cite{carr-lee-wu-1}, the background process is modeled as a L\'{e}vy process, which naturally includes the possibility of jumps, but does not include stochastic volatility.  Stochastic volatility is added by time-changing the background process with a continuous increasing stochastic clock.  In contrast, in \cite{mendoza-carr-linetsky-1}, the background process is modeled as a diffusion, which may include local stochastic volatility, but does not include jumps.  Jumps are added by time-changing the background process with a L\'{e}vy subordinator.  Additionally, the framework of \cite{mendoza-carr-linetsky-1} allows for the possibility of a default event, whereas the framework of \cite{carr-lee-wu-1} does not.  While default may not be realistic for an index, it is certainly an important consideration for individual stocks as explained above.\footnote{A default event would cause $[ \log S ]_t$ to blow up.  As such, we must amend our definition of a VS to account for this possibility.  We will do this in Section \ref{sec:VS}.}  Finally, in \cite{carr-lee-wu-1}, the ratio of $K_{var}$ to the value of the $\log$ contract $\E \[ - \log \( S_t/S_0 \) \]$ is a constant $Q_X$, which is \emph{independent} of the initial value of the underlying $S_0$.  However, empirical results from \cite{carr-lee-wu-1} indicate that this ratio is not constant.  In the L\'evy subordinated diffusion setting of our paper, the ratio $\frac{K_{var}}{\E \[ - \log \( S_t/S_0 \) \]}$ can, in general, depend on the initial value $S_0$ of the underlying.  The reason for this difference is that L\'evy processes are spatially homogeneous, whereas diffusion processes may have locally-dependent drift and diffusion coefficients.
\par
Despite the differences between these frameworks, they share one desirable feature: when the background process is fixed and one has full knowledge of the volatility smile, the fair value of $K_{var}$ is robust to misspecification of the time-change process.  In \cite{carr-lee-wu-1}, the effect that the time-change has on the value of $K_{var}$ is felt through the $\log$ contract, which is constructed directly from European calls.  In the framework of \cite{mendoza-carr-linetsky-1}, we will show that the L\'{e}vy subordinator can be inferred directly from the volatility smile.  Once the subordinator is obtained, it can be used to compute the fair value of $K_{var}$ (among other things).
\par
The rest of this paper proceeds as follows.  In Section \ref{sec:model} we describe the class of L\'{e}vy subordinated diffusion processes in detail.  In Section \ref{sem.gens.ito} we introduce some mathematical tools, which we shall need to compute the price of a VS.  In Section \ref{sec:VS} we modify the payoff of a VS to account for the possibility that the underlying defaults (i.e., jumps to zero).  We then derive a general expression for the value of the modified VS in the L\'{e}vy subordinated diffusion setting.  In Section \ref{sec:spectral0} we present some important results concerning generalized eigenfunction expansions.  These results will be needed for the diffusion-specific VS computations provided in Section \ref{sec:examples}.  In Section \ref{sec:impliedtime} we show that, by observing European call and put prices written on the underlying, one can uniquely determine (in a non-parametric way) the drift and L\'{e}vy measure of the subordinator that drives the underlying price process.  
In section \ref{sec:calibration} we implement the numerical methods developed in Sections \ref{sec:model}-\ref{sec:impliedtime}.  First, we extract the subordinator driving four different stocks by observing call and put prices.  Then, we then use the extracted subordinator to compute VS rates.
Finally, in Section \ref{sec:conclusion}, we provide some concluding remarks and discuss directions for future research.

%
%

\section{Model}
\label{sec:model}
We assume frictionless markets, no arbitrage, and take an equivalent martingale measure (EMM) $\P$ chosen by the market on a complete filtered probability space $(\Om,\F,\P)$ as given.  All stochastic processes defined below live on this probability space, and all expectations are with respect to $\P$ unless stated otherwise. {As all of the processes considered below are Feller, the natural filtrations generated by these processes are right-continuous.  We shall assume that all the (natural) filtrations defined throughout this Section are augmented to include the $\P$-null sets.}  In particular, this assumption holds for the \emph{subordinate filtration} ${\mathbb F}=\{\F_t,t\geq0\}$, which we describe at the end of this Section.  For simplicity, we assume zero interest rates.  All of our results can be easily extended to include deterministic interest rates.
\par
As in \cite{mendoza-carr-linetsky-1}, we model the stock price dynamics under the risk-neutral pricing measure $\P$  as a stochastic process $S=\{S_t, t\geq0\}$ defined by
\begin{align}
S_t
		&= 	(1-D_{T_t})e^{\rho t}X_{T_t} , &
D_t
		&=	\I_{\{\zeta \leq t\}} , &
X_0
		&=	x , &
T_0
		&=	0 , &
\zeta
		&>	0 .  \label{eq:S}
\end{align}
Here, the background process $X=\{X_t, t \geq 0\}$ is a scalar Feller diffusion, $T=\{T_t, t \geq 0\}$ is a L\'{e}vy subordinator, $\rho$ is a scaling factor, which is needed to ensure that the asset price $S$ is a martingale, and $\zeta$ is a positive random variable, which will be used to model a possible default event of the underlying $S$.  Below, we describe each of the above-mentioned elements in detail.


\textbf{Background Feller process $X$.}
We let $X=\{X_t,t\geq 0\}$ be a time-homogeneous Markov diffusion process, starting from a positive value $X_0=x>0$, which solves a stochastic differential equation (SDE) of the form
\begin{align}
dX_{t}
	&=b(X_t)\, dt+a(X_t)\, dB_{t},
\label{eq:dX}
\end{align}
where
\begin{align}
b(x)&:= [\mu+k(x)]x &
		& \text{and}& 
a(x)&:=\sigma(x) x . \label{eq:drift,vol}
\end{align}
Here, $\sigma(x)$ and $[\mu+k(x)]$ are the state-dependent instantaneous volatility and drift rate, $\mu\in {\mathbb R}$ is a constant, and $B=\{B_t,t\geq 0\}$ is a standard Brownian motion.  We assume that $\sigma(x)>0$ and $k(x)>0$ are Lipschitz continuous on the interval $[\eps,\infty)$ for each $\eps>0$ (i.e., locally Lipschitz), and that $\sigma(x)$ and $k(x)$ remain bounded as $x\rightarrow \infty$. We do not assume that $\sigma(x)$ and $k(x)$ remain bounded as $x\rightarrow 0$. Under these assumptions the process $X$ does not explode to infinity (i.e., infinity is a {\em natural boundary} for the diffusion process; see \cite{borodin-salminen} p.14 for boundary classification of diffusion processes). We also assume that zero is either an (unattainable) natural boundary or an entrance boundary. If zero is a natural boundary the state space is given by $E=(0,\infty)$. If zero is an entrance boundary, i.e., the process $X$ can be started at $x=0$  then it quickly moves to interior of $[0,\infty)$ to never hit zero again. Throughout this document we assume that the process always starts from a positive value $X_0=x>0$, and hence the state space is also defined as $E=(0,\infty)$.
Under all our previous assumptions $X$ is the unique strong solution to the SDE~\eqref{eq:dX}. The transition function $P^0_t(x,dy)={\mathbb P}(X_t\in dy|X_0=x)$ of the diffusion process $X$ started at $x$ defines a Feller semigroup $\Pc^0=\{\Pc_t^0,t\geq 0\}$ acting  
on the space $C([0,\infty])$ 
of functions continuous on $(0,\infty)$ and such that the limits $\lim_{x\rightarrow 0}f(x)$ and $\lim_{x\rightarrow \infty}f(x)$ exist and are finite (Ethier and Kurtz (1986) p.366)  by 
\begin{align}\label{smgn.1}\Pc_t^0f(x)&=\E_x\big[f(X_t)\big]=\int_E f(y) P^0_t(x,dy). \end{align}
The infinitesimal generator of $\Pc^0$ is a second-order differential operator of the form
\begin{align}\label{infg.0} 
\A^0 f(x)&= \frac{1}{2}a^2(x)\frac{d^2 f}{dx^2}(x)+b(x)\frac{df}{dx}(x)
\end{align}
with the domain ${\rm Dom}(\A^0)=\{f\in C([0,\infty])\cap C^2((0,\infty)),\A^0f\in C([0,\infty])\}$ if zero is an inaccessible boundary (natural or entrance).  We also note that the semigroup leaves the space $C_0((0,\infty))\subset C([0,\infty])$ of functions continuous on $(0,\infty)$ and having zero limits  $\lim_{x\rightarrow 0}f(x)=0$ and $\lim_{x\rightarrow \infty}f(x)=0$ invariant and is a Feller semigroup on it.  Lastly, we denote by ${\mathbb F}^X=\{\F_t^X,t\geq0\}$ the completed natural filtration of the process $X$.


\textbf{The trigger event $\zeta$ and the indicator process $D$.} Let $\Ec \sim \text{Exp}(1)$ be an exponential random variable, independent of $X$, and define the \emph{trigger event time} $\zeta$ as 
\begin{align}
\zeta
		&=		\inf \Big\{ t \geq 0 : \int_0^{t} k(X_s) ds \geq \mathscr{E} \Big\}. \label{zeta.1}
\end{align}
That is,  $\zeta$ is the first jump time of a doubly stochastic Poisson process with jump intensity given by the killing rate $k(x)$.  Observe that the killing rate is added in the drift~\eqref{eq:drift,vol} to compensate for the killing (jump-to-default).  This compensation will be needed in order to ensure that the stock price $S$ is a martingale.  Moreover, since the process $X$ cannot reach zero from the interior of the state space, then the underlying stock price process $S$ cannot go to zero continuously, rather, it may only jump to zero from a strictly positive value. We denote 
by $\Pc^1=\{\Pc^1_t,t\geq0\}$ the Feynman-Kac semigroup associated with the killing rate $k(X)$:
\begin{align} 
\Pc^1_tf(x)&=\E_x\big[e^{-\int_0^t k(X_u)du}f(X_t)\big]=\int_E f(y) P^1_t(x,dy),\quad t\geq0.\label{fk.1}
\end{align}
$\Pc^1$ is a (sub-Markovian) Feller semigroup on $C([0,\infty])$
with the generator
\begin{align}
\A^1 f(x)&=\A^0f(x)-k(x)f(x)
\label{infg.1}\end{align}
with the domain ${\rm Dom}(\A^1)\subseteq {\rm Dom}(\A^0)$. More precisely,  ${\rm Dom}(\A^1)=\{f\in C([0,\infty])\cap C^2((0,\infty)),\A^1 f\in C([0,\infty])\}$. Since zero is an inaccessible boundary (natural or entrance) of the {\em diffusion with killing at the rate} $k(x)$ it suffices to restrict the domain to  ${\rm Dom}(\A^1)=\{f\in C_0([0,\infty])\cap C^2((0,\infty)),\A^1 f\in C_0([0,\infty])\}$ whenever we start the process from the interior of $E$, i.e., $x>0$ (cf. \cite{borodin-salminen} pp.15ff). 

Since the random variable $\zeta$ is not $\F^X_\infty$-measurable, we  introduce an \emph{indicator process} $D=\{D_t, t \geq 0\}$ in order to keep track of the event $\{\zeta \leq t\}$. The indicator process $D$ is defined by
\begin{align}
D_t
		&=		\I_{\{\zeta \leq t\}}  . \label{eq:D}
\end{align}
Lastly, we denote by ${\mathbb F}^D=\{\F^D_t,t\geq0\}$ the (completed) natural filtration of the process $D$. 


\textbf{The auxiliary bivariate process  $(X,D)$.}  As in \cite{linetsky-mendoza-arriaga-3} we define the auxiliary bivariate process $(X,D)=\{(X_t, D_t),t\geq0\}$ with state space $\tilde{E}=E\times\{0,1\}$. The process $(X,D)$ is a Feller semimartingale in the enlarged filtration ${\mathbb F}^{X,D}=\{\F_t^{X,D},t\geq0\}$ with $\F_t^{X,D}=\F_t^X\vee\F_t^D$ (${\mathbb F}^{X,D}$ is the smallest filtration that contains ${\mathbb F}^X$ and in which $\zeta$ is a stopping time). Moreover, any function $f(x,d)\in C([0,\infty]\times\{0,1\})$ is represented as,
\begin{align}\label{func.f.1}
f(x,d)
	&=  u(x)(1-d)+v(x)d 
	=  (u-v)(x)(1-d)+v(x) , &u,v&\in C([0,\infty]).
\end{align}
As we shall see below, the function $u(x)=f(x,0)$ can be interpreted as a promised payoff function if the triggering event $\zeta$ {\em does not} occur by time $t>0$, while $v(x)=f(x,1)$ can be understood as a recovery payoff function if the triggering event occurs prior time $t>0$. 

The following theorem gives the Markovian characterization of the bivariate process $(X,D)$.
\begin{theorem}\label{th1.1} 
(i) The bi-variate process $(X,D)$ is a Feller process with the Feller semigroup $\Pc=\{\Pc_t,t\geq 0\}$ acting on $f\in C([0,\infty]\times \{0,1\})$ according to:
\be
\Pc_t f(x,d)=\Pc^0_tv(x)+(1-d)\Pc^1_t(u-v)(x),\eel{semi.bi.1}
where $u,v\in C([0,\infty])$, $\Pc^0$ is the Feller semigroup of the process $X$~\eqref{smgn.1} and $\Pc^1$ is the corresponding Feynman-Kac semigroup~\eqref{fk.1}.

(ii) The infinitesimal generator of the Feller semigroup $\Pc$ is given by
\be
\A f(x,d)=\A^0 f(x,d) +k(x)(f(x,1)-f(x,d))
\eel{semigdelt.1.0}
\be=\A^0 v(x)+(1-d)\A^1(u-v)(x),\eel{semigdelt.1}
where
$\A^0$ and $\A^1$ are the generators of $\Pc^0$  and  $\Pc^1$, respectively.

(iii) If $f(x,d)\in {\rm Dom}(\A)$ (i.e., $f$ is of the form \eqref{func.f.1} with $u,v\in  {\rm Dom}(\A^1)$) and $(X,D)$ is the bi-variate process with $X_0=x>0$ and $D_0=d\in \{0,1\}$, then the process 
\be
M^f_t:=f(X_t,D_t)-f(x,d)-\int_0^t \A f(X_s,D_s)ds
\eel{gen.uno.b}
is an ${\mathbb F}^{X,D}$-martingale.
\end{theorem} 
\begin{proof} The proof can be found in \cite{linetsky-mendoza-arriaga-3}.
\end{proof}
Similarly, the Doob-Meyer decomposition of $D$ is given in the following corollary.
\begin{corollary}
 The increasing process $D$ has the Doob-Meyer decomposition
\be 
D_t=M_t+A_t
\eel{doobs1.1}
with the predictable ${\mathbb F}^{X,D}$-compensator  $A_t$ and  ${\mathbb F}^{X,D}$-martingale $M_t$
\be 
A_t=\int_0^t k(X_u)(1-D_u)du,\quad M_t=D_t-\int_0^t k(X_u)(1-D_u)du.
\eel{doobs1.2}
\end{corollary}
\begin{proof} See Section 5.1.4 in \cite{bielecki-rutkowski-1} (also, Lemma 7.3.4.3 and Comment 7.3.4.4 in \cite{jeanblanc-yor-chesney-1}).\end{proof}


\textbf{The L\'{e}vy subordinator $T$.}
A \emph{L\'{e}vy subordinator} $T=\{T_t,t\geq0\}$ is a L\'{e}vy process with positive jumps and non-negative drift.  For a standard reference on subordinators we refer the reader to \cite{bertoin-4}.  We require that $T$ be independent of $X$ and $\Ec$ and satisfy $T_0=0$.  Every L\'{e}vy subordinator has the following \emph{It\^{o}-L\'{e}vy decomposition}
\begin{align}
dT_t
		&=		\gamma \, dt + \int_{(0,\infty)} s \, dN_t(ds) ,
\end{align}
where $\gamma \geq 0$ is the \emph{drift} of the subordinator and $N$ is a \emph{Poisson random measure} with the property that, for any Borel set $A \in \B((0,\infty))$ we have $\E \[ dN_t(A) \]=\nu(A) dt$ for some $\sig$-finite measure $\nu$ on $\B((0,\infty))$.  The measure $\nu$, which must satisfy $\int_0^\infty \( s \wedge 1 \) \nu(ds) < \infty$, is referred to as the \emph{L\'{e}vy measure}.  The Laplace transform of a L\'{e}vy subordinator $T$ is given by
\begin{align}
\E\[ e^{\lam T_t} \]
		&=	\int_0^\infty e^{\lam s}\pi_t(ds)=		e^{- t \phi(-\lam)} , &
\lam
		&\in	\left\{ \lam \in \R : \int_{[1,\infty)} e^{\lam s}\nu(ds)<\infty \right\} =: \Ic ,		\label{eq:set}
\end{align}
where $\pi_t(ds)$ is the transition function of the subordinator $T$ and $\phi(\lam)$ is the \emph{Laplace exponent} of $T$, which can be computed explicitly from the \emph{L\'{e}vy-Kintchine formula}
\begin{align}
\phi(\lam)
		&=		\gamma \, \lam + \int_{(0,\infty)} \(1 - e^{-\lam s} \) \nu(ds) . \label{eq:LevyKintchine}
\end{align}
Note that the Laplace exponent $\phi(\lam)$ is concave and increasing and satisfies $\phi(0)=0$ (see \cite{bertoin-1}, page 73).  For any Borel set $A \in \B((0,\infty))$ the process $\{ N_t(A), t \geq 0\}$ is a Poisson process with intensity $\nu(A) = \int_A \nu(ds)$.  If the arrival rate of all jumps is finite $\alpha:=\nu((0,\infty))<\infty$ then $\int_0^T\int_0^\infty s \, dN_t(ds)$  is a compound Poisson process.  In this case, the distribution of jumps is $\alpha^{-1}\nu$.  If $\nu((0,\infty))=\infty$ then the subordinator is said to be an \emph{infinite activity} subordinator.
\par
We define $L=\{L_t,t\geq 0\}$ the {\em first passage process} or {\em right inverse process} of $T$ as $L_t:=\inf\{s:  T_s > t\}$.  Recall that $T$ is assumed to be {\em independent} of $\F^{X,D}_\infty$.   Therefore, $L$ is independent of $\F^{X,D}_\infty$ as well. We let ${\mathbb F}^L= \{\F^L_t,t\geq0\}$ be the natural filtration of the process $L$.


{\bf The subordinate filtration ${\mathbb F}$, the auxiliary subordinate bivariate process $(X^\phi,D^\phi)$, and the default  time $\zeta^\phi$.}  Recall from our above discussion that ${\mathbb F}^X$, ${\mathbb F}^D$ and ${\mathbb F}^L$ correspond to the filtrations generated by the processes $X$, $D$, and $L$, respectively, and that the filtration ${\mathbb F}^{X,D}$ is the smallest filtration that contains ${\mathbb F}^X$ and in which $\zeta$ is an ${\mathbb F}^{X,D}$-stopping time. Similarly, we can define the filtration ${\mathbb F}^G=\{\F^G_t,t\geq0\}$, with $\F^G_t =\F^{X,D}_t\vee\F^L_t$, to be the smallest filtration that contains ${\mathbb F}^{X,D}$ and in which $T$ is an increasing family of ${\mathbb F}^G$-stopping times. Then the {\em subordinate filtration} ${\mathbb F}=\{\F_t,t\geq0\}$ is constructed by time-changing the filtration ${\mathbb F}^G$ with the L\'{e}vy subordinator $T$, i.e., $\F_t=\F^G_{T_t}$. Observe that since $T$ is an ${\mathbb F}^G$-stopping time, then ${\mathbb F}$ is the filtration containing all of the information of the bivariate process $(X,D)$ prior to the stopping time $T$. Consequently, one can define the {\em subordinate bivariate process} $(X^\phi,D^\phi):=\{(X_{T_t},D_{T_t}),t \geq 0 \}$ by time-changing the bivariate process $(X,D)$ with the subordinator $T$. This transformation is called {\em Bochner's subordination} due to work on subordination of semigroups and their generators by \cite{bochner-1}.
The subordinate process $(X^\phi,D^\phi)$  is a Feller ${\mathbb F}$-semimartingale (see \cite{linetsky-mendoza-arriaga-3}).
\par
From \eqref{eq:D} we observe that, before subordination, the indicator process $D$ satisfies $D_t=1$ for all $t\geq\zeta$. On the other hand, after subordination, the subordinate indicator process $D^\phi$ satisfies $D^\phi_t=1$ for all $t \geq 0$ for which $T_t \geq \zeta$ (i.e., $D^\phi_t=\I_{\{T_t\geq\zeta\}}$). Therefore, in the credit-equity context one can define the {\em default time}  $\zeta^\phi$ by
\begin{align}
\zeta^\phi&= L_{\zeta-},&L_{a-}&=\inf\{t: T_t \geq a\}. \label{deftime1}
\end{align}
Certainly, $\zeta^\phi$ is the first passage time process of $T$ across the level $\zeta$ and the identity $\{\zeta^\phi\leq t\}=\{T_t\geq\zeta\}$ holds (see Section 5.1.1 in \cite{jeanblanc-yor-chesney-1}). Hence, $D^\phi_t=\I_{\{T_t\geq\zeta\}}=\I_{\{\zeta^\phi\leq t\}}$ is the {\em default indicator process}. The characterization of the subordinate process $(X^\phi,D^\phi)$ is provided below in Section~\ref{sem.gens.ito}.


\textbf{The stock price $S$ and the scaling constant $\rho$.} From Eq.~\eqref{eq:S} we observe that the dynamics of the stock price $S$ can be described by means of the subordinate bivariate process $(X^\phi,D^\phi)$. Indeed, the stock price $S$ can be seen as a function $f(t,X^\phi_t,D^\phi_t)\in C([0,\infty)\times\tilde{E})$ which is decomposed according to~\eqref{func.f.1} with the payoff function $u(t,x)=e^{\rho t}x$ if no default occurs by time $t\geq0$, and zero recovery $v(t,x)=0$ if the firm defaults prior to time $t\geq0$.  That is,  

\begin{align}
S_t &= e^{\rho t} X^\phi_t(1-D_t^\phi). \label{eq:S2}
\end{align}
The scaling constant $\rho$ is introduced to ensure that the asset $S$ is an ${\mathbb F}$-martingale.  As shown in \cite{mendoza-carr-linetsky-1}, $S$ will be a martingale if and only if $\rho = \phi(-\mu)$ and $\mu \in \Ic$, where the set $\Ic$ is defined in Eq.~\eqref{eq:set}. That is, assuming zero interest rates, $\E[S_t]<\infty$ for every $t\geq0$, and $\E[S_{t_2}|\F_{t_1}]=S_{t_1}$ for every $t_1<t_2$. From the previous condition in $\rho$ we are free to choose any value of $\mu$ as long as $\mu\in\Ic$. Hence, from this point onward we assume that $\mu\in\Ic$. Observe that the underlying assumption $v(x)=0$ is equivalent to modeling the stock price $S$ under {\em absolute priority}, which means that the stock holder has zero recovery in the event of default.

%
%


\section{Markovian and Semimartingale Characterization of the Subordinate Process $(X^\phi,D^\phi)$}
\label{sem.gens.ito}
Before proceeding with the calculation of the quadratic variation of the $\log$ price process $\log S$ it is essential to describe the characteristics of the underlying stock process $S$ in terms of the subordinate bivariate process $(X^\phi,D^\phi)$. For convenience, we summarize some of the key results of Section 3 in \cite{linetsky-mendoza-arriaga-3} who give the Markovian and semimartingale characterization of the process $(X^\phi,D^\phi)$. We refer the reader to \cite{linetsky-mendoza-arriaga-3} for the corresponding proofs.

We begin by recalling some key results about subordination (in the sense of \cite{bochner-1}) of semigroups of operators in Banach spaces. The expression for the generator is due to \cite{phillips-1}.
\begin{theorem}\label{phillips.theo}{\bf (\cite{phillips-1})} Let $T $ be a subordinator with L\'{e}vy measure $\nu$, drift $\gamma\geq0$, Laplace exponent $\phi(\lambda)$, and transition function $\pi_t(ds)$.  Let $\Pc$ be a strongly continuous contraction semigroup of linear operators on a Banach space ${\mathfrak B}$ with infinitesimal generator $\A$. 

$(i)$ Define 
\begin{align}\Pc^\phi_tf(x)&=\int_{[0,\infty)} \Pc_s f(x)\pi_t(ds),\quad t\geq0,\,\, f\in{\mathfrak B}.\label{subsem.1}\end{align}
Then $\Pc^\phi=\{\Pc^\phi_t,t\geq0\}$ is a strongly continuous contraction semigroup of linear operators on ${\mathfrak B}$ called {subordinate semigroup of $\Pc$ with respect to the subordinator} $T $.

$(ii)$ Denote the infinitesimal generator of $\Pc^\phi$  by $\A^\phi$. Then the domain of $\A$ is a core of $\A^\phi$ and
\begin{align} \A^\phi f& = \gamma \A f + \int_{(0,\infty)} (\Pc_s f-f)\nu(ds), \quad f\in {\rm Dom}(\A).
\label{subgen.1}\end{align}

$(iii)$ Moreover, if $\Pc$ is a Feller semigroup on $C([0,\infty])$, then the subordinate semigroup $\Pc^\phi$ is also a Feller semigroup on $C([0,\infty])$.
\end{theorem}

Next, recall that if $0$ is not an absorbing boundary for the process $X$ with diffusion coefficient $a(x)$, drift $b(x)$ and killing rate $k(x)$, then the transition kernels $P^\beta(x,dy)$, $\beta=0,1$, of the semigroups $\Pc^\beta$ have densities with respect to the Lebesgue measure, $P^\beta_t(x,dy)=p^\beta(t,x,y)dy,$
where $p^\beta(t,x,y)$ are jointly continuous in $t,x,y$. This follows from the fact that any one-dimensional diffusion has a density with respect to the speed measure that is jointly continuous in $t,x,y$ (cf. \cite{mckean-1} or \cite{borodin-salminen} p.13). Under our assumptions, the speed measure is absolutely continuous with respect to the Lebesgue measure (cf. \cite{borodin-salminen}, p.17) and, hence, the semigroups have densities with respect to the Lebesgue measure. For $\beta=0$, the transition kernel is conservative, i.e., $P_t^0(x,E)=\int_{E}p^0(t,x,y)dy=1$. For $\beta=1$ the kernel is generally defective, i.e., $P_t^1(x,E)=\int_{E}p^1(t,x,y)dy\leq 1$. While our diffusion is non-negative, for future convenience we extend the transition densities from $E$ to ${\mathbb R}$ by setting $p^\beta(t,x,y)\equiv0$ for all $y\notin E$ and for all $x\in E$ and $t\geq 0$. Then, the Markovian characterization of the subordinate bivariate process $(X^\phi,D^\phi)$ can be obtained from Theorem~\ref{phillips.theo} as follows.

\begin{theorem}\label{mkvian.sub.biv} {\bf (Markovian characterization of $(X^\phi,D^\phi)$)} 
$(i)$ The bi-variate process $(X^\phi,D^\phi)$ is a Feller process with the Feller semigroup $\{\Pc^\phi_t,t\geq 0\}$ acting on $f\in C([0,\infty]\times \{0,1\})$ by:
\begin{align}
\Pc_t^\phi f(x,d)&=\Pc^{\phi,0}_tv(x)+(1-d)\Pc^{\phi,1}_t(u-v)(x),
\label{semi.sub.bi.1}\end{align}
where $u(x)=f(x,0)\in C([0,\infty])$, $v(x)=f(x,1)\in C([0,\infty])$, and $\{\Pc_t^{\phi,0},t\geq 0\}$ and 
$\{\Pc_t^{\phi,1},t\geq 0\}$ are Feller semigroups obtained by subordination in the sense of Bochner from Feller semigroups 
$\{\Pc_t^{0},t\geq 0\}$ and 
$\{\Pc_t^{1},t\geq 0\}$.

$(ii)$ The infinitesimal generator $\A^\phi$ of the Feller semigroup $\{\Pc_t^\phi,t\geq 0\}$ has the following representation:
\begin{align}\A^\phi f(x,d)&=\A^{\phi,0} v(x)+(1-d)\A^{\phi,1} (u-v)(x),\quad u,v \in{\rm Dom}(\A^1),\label{subdeltsemi.3}\end{align}
where $\A^{\phi,\beta}$, $\beta\in\{0,1\}$, are generators of $\{\Pc_t^{\phi,\beta},t\geq 0\}$.

$(iii)$ 
The generator $\A^{\phi,\beta}$ has the following L\'{e}vy-Khintchine-type representations with state-dependent coefficients:
\begin{align}
\A^{\phi,\beta} f(x)&=\frac{1}{2}\gamma a^2(x)f^{\prime\prime}(x)+b^{\phi,\beta}(x)f^\prime(x)-\beta k^{\phi}(x)f(x)\label{eq13}\\
&+\int_{{\mathbb R}}\Big(f(x+y)-f(x)-\I_{\{|y|\leq1\}}y f^\prime(x)\Big)\pi^{\phi,\beta}(x,y)dy,\quad f\in {\rm Dom}(\A^\beta),\label{subgen.a.1}\end{align}
where 
the state-dependent L\'{e}vy density $\pi^{\phi,\beta}(x,y)$ is defined for all $y\not= x$ by
\begin{align}
\pi^{\phi,\beta}(x,y)&=\int_{(0,\infty)}p^\beta(s,x,x+y)\nu(ds),
\label{sublevmeas.1}\end{align}
and satisfies the integrability condition $\int_{{\mathbb R}}(|y|^2\wedge1)\pi^{\phi,\beta}(x,y)dy<\infty$ for each $x\in E$ (recall that we extended $p^\beta(t,x,y)$ to ${\mathbb R}$ by setting $p^\beta(t,x,y)\equiv0$ for $y\notin E$),
the drift with respect to the truncation function $h^{X^\phi}(x)=x\I_{\{|x|\leq 1\}}$ is given by, 
\begin{align}b^{\phi,\beta}(x)&=\gamma b(x) + \int_{(0,\infty)}\Big(\int_{\{|y|\leq1\}} y p^\beta(s,x,x+y)dy\Big)\nu(ds),\label{subdrift.1}\end{align}
and the killing rate is given by 
\begin{align}k^{\phi}(x)&=\gamma k(x) + \int_{(0,\infty)}\Big(1-P^1_s(x,E)\Big)\nu(ds),\label{subkill.1} \end{align}
where $P^1_s(x,E)=\int_{E} p^1(s,x,y)dy$.
\par
$(iv)$ If $f(x,d)\in {\rm Dom}(\A)$ (i.e., $f$ is of the form \eqref{func.f.1} with $u,v\in {\rm Dom}(\A^1)$) and $(X^\phi,D^\phi)$ starts with $X^\phi_0=x>0$ and $D^\phi_0=d\in \{0,1\}$, then the process 
\be
M^f_t:=f(X^\phi_t,D^\phi_t)-f(x,d)-\int_0^t \A^\phi f(X^\phi_s,D^\phi_s)ds
\label{eqr.9}\end{equation}
is an ${\mathbb F}$-martingale.
\end{theorem}

Now we turn our attention to the semimartingale characterization of the process $(X^\phi,D^\phi)$ (see \cite{jacod-shiryaev-1}, p.76, for the definition of predictable characteristics of a semimartingale).
\begin{theorem}\label{semchar.t.1} {\bf (Semimartingale characterization of $(X^\phi,D^\phi)$)} 
$(i)$ The bi-variate ${\mathbb F}$-semimartingale $(X^\phi,D^\phi)$ has the following predictable characteristics.
The predictable quadratic variation of the continuous local martingale component $X^{\phi,c}$ is:
\begin{align}
C^{X^\phi X^\phi}_t&=\gamma\int_0^t  a^2(X_s^\phi)ds
\label{eqr.16}\end{align}
($C^{D^\phi D^\phi}_t=0$ and $C^{X^\phi D^\phi}_t=0$ since $D^\phi$ is purely discontinuous).
The predictable process of finite variation associated with the truncation function $(h^{X^\phi}(x)=x\I_{\{|x|\leq 1\}},h^{D^\phi}(x,d)=d)$ is:
\begin{align}
B_t^{X^\phi}&=\int_0^t b^{\phi,0}(X_s^\phi)ds,& B_t^{D^\phi}=\int_0^t (1-D_s^\phi)k^{\phi}(X_s^\phi)ds,
\label{eqr.17}\end{align}
where the function $b^{\phi,0}(x)$ is defined in Eq.~\eqref{subdrift.1}  and $k^{\phi}(x)$ is defined in Eq.~\eqref{subkill.1}.
The compensator of the random measure $\mu^{X^\phi,D^\phi}(\omega;dt, dy\, dz)=\sum_u \I_{\{\Delta (X_u^\phi,D_u^\phi)(\omega)\not=0\}}\delta_{(u,\Delta X_u^\phi(\omega),1)}(ds,dy\,dz)$ associated to the jumps of $(X^\phi,D^\phi)$ is a predictable random measure on ${\mathbb R}_+ \times ({\mathbb R}^2\backslash \{(0,0)\})$:
\begin{align}
\nu^{X^\phi,D^\phi}(\omega;dt, dy\, dz) &=\left[\pi^{\phi,0}(x,y)-(1-d)(\pi^{\phi,0}(x,y)-\pi^{\phi,1}(x,y))\right]dy\,\delta_0(dz)
\label{teq.4}\\
&+(1-d)\gamma k(x)\delta_0(dy)\delta_1(dz)+(1-d)(\pi^{\phi,0}(x,y)-\pi^{\phi,1}(x,y))dy\,\delta_1(dz),
\label{eqr.18}\end{align}
where $\pi^{\phi,\beta}(x,y)$ are the L\'{e}vy densities defined in Eq.~\eqref{sublevmeas.1} with $\beta=0,1$, and $\delta_a$ is the Dirac measure charging $a$.

$(ii)$ The L\'{e}vy-It\^{o} canonical representation of $X^\phi$ with respect to the truncation function $h^{X^\phi}(x)=x \I_{\{|x|\leq 1\}}$ is:
\begin{align}
X^\phi_t&=x+B_t^{X^\phi}+X_t^{\phi,c}+\int_0^t \int_{{\mathbb R}} y \I_{\{|y|\leq 1\}} \left(\mu^{X^\phi}(ds,dy)-\nu^{X^\phi}(ds,dy)\right)
+\int_0^t \int_{{\mathbb R}} y \I_{\{|y|>1 \}} \mu^{X^\phi}(ds,dy),
\label{eqr.21}\end{align}
where the compensator of the random measure $\mu^{X^\phi}(\omega;dt, dy)=\sum_u \I_{\{\Delta X_u^\phi(\omega)\not=0\}}\delta_{(u,\Delta X_u^\phi(\omega))}(ds,dy)$ associated to the jumps of $X^\phi$ is a predictable random measure on ${\mathbb R}_+ \times ({\mathbb R}\backslash \{0\})$:
\begin{align}
\nu^{X^\phi}(\omega;dt,dy)&=\pi^{\phi,0}(X^\phi_{t-},y)dy\, dt.
\label{eqr.22}\end{align}

$(iii)$ The Doob-Meyer decomposition of $D^\phi_t$ is:
\begin{align}
D^\phi_t&=B_t^{D^\phi} + M^\phi_t
\label{eqr.23}\end{align}
with the martingale $M^\phi_t=D_t^\phi-B_t^{D^\phi}$ and the predictable compensator $B_t^{D^\phi}$ given in Eq.~\eqref{eqr.17}, so that the ${\mathbb F}$-intensity is $\lambda_t^{\phi}=(1-D_t^\phi)k^{\phi}(X_t^\phi)$.
\end{theorem}

Lastly,  we formulate the It\^{o} formula for functions of the bi-variate process in a form convenient for our application. {Observe also that the continuous local martingale part can be represented in terms of the Brownian motion $W$ as $X_t^{\phi,c}= \sqrt{\gamma}\int_0^t a(X_s^\phi)dW_s$, which will be useful for our representation of the stock price process $S$.} Since for  $\rho=\phi(-\mu)$ the stock price process $S$ is a martingale (hence, a special semimartingale) then it suffices to present the It\^o formula for the case in which the process $f(t,X_t^\phi,D_t^\phi)$ is a special semimartingale (see.~\cite{jacod-shiryaev-1}, Definition 4.21, p.43). 

\begin{theorem}\label{cor.xD.1} {\bf (It\^{o} Formula for $(X^\phi,D^\phi)$)}
Suppose $(X^\phi,D^\phi)$ starts from $X^\phi_0=x>0$ and $D_0^\phi=d\in \{0,1\}$. 
For any function $f(t,x,d)=v(t,x)+(1-d)(u(t,x)-v(t,x))$ with $u(t,x)$ and $v(t,x)\in C^{1,2}({\mathbb R}_+ \times (0,\infty))$  (recall that zero is an unattainable boundary for the diffusion process $X$ starting at $x>0$), if $f(t,X_t^\phi,D_t^\phi)$ is a special semimartingale, the
It\^{o} formula can be written in the following form:
\begin{align}
f(t,X^\phi_t,D^\phi_t)&=f(0,x,d)+\int_0^t \big(\partial_s+\A^{\phi}\big)f(s,X^\phi_s,D^\phi_t)ds
+\sqrt{\gamma}\int_0^t\partial_x f(s,X^\phi_s,D^\phi_t)a(X_s^\phi)dW_s
\label{eqm.18.f1}\\
&+\int_0^t \int_{{\mathbb R}}\left(v(s,X^\phi_{s-}+y)-v(s,X^\phi_{s-})\right)(\mu^{X^\phi}(ds,dy)-\nu^{X^\phi}(ds,dy)).
\label{eqr.19.f1}\\
&+\int_0^t \int_{{\mathbb R}}\left((u-v)(s,X^\phi_{s-}+y)-(u-v)(s,X^\phi_{s-})\right)(1-D^\phi_{s-})(\hat{\mu}(ds,dy)-\hat{\nu}(ds,dy)).
\label{eqr.19.f2}\\
&-\int_0^t (1-D^\phi_{s-})(u-v)(X^\phi_{s-})dM^\phi_s.
\label{eqr.19.f2.b}\end{align}
where we introduced a random measure associated  to those jumps of $X^\phi$ that do {\em not} coincide with jump of $D^\phi$,
\begin{align}
\hat{\mu}(\omega;ds,dy)&=\sum_u \I_{\{\Delta X_u^\phi(\omega)\not=0\}} \I_{\{\Delta D_u^\phi(\omega)=0\}}\delta_{(u,\Delta X_u^\phi(\omega))}(ds,dy),
\label{teq.14}\end{align}
and its compensator measure
\begin{align}
\hat{\nu}(\omega;ds,dy)&=\left[\pi^{0,\phi}(X_{s-}^\phi,y)-(1-D^\phi_{s-})(\pi^{0,\phi}(X_{s-}^\phi,y)-\pi^{1,\phi}(X_{s-}^\phi,y))\right]dyds.
\label{teq.15}\end{align}
Here, $\mu^{X^\phi}$ is the random measure associated to the jumps of $X^\phi$ and $\nu^{X^\phi}$ is its compensator measure~\eqref{eqr.22}. The generator $\A^\phi$ is given by Eq.~\eqref{subdeltsemi.3}.
\end{theorem}

\section{Variance Swap Computation}
\label{sec:VS}
Due to the semimartingale characterization of the process $(X^\phi,D^\phi)$ of Section~\ref{sem.gens.ito} we are now in position to provide the characterization of the stock price process $S$ and to explicitly compute the value of variance swap rate $K_{var}=\E([\log S]_t)$. First, let us recall from  Eq.~\eqref{eq:S2} that the stock price $S$ can be seen as a function $f(t,X^\phi_t,D^\phi_t)\in C([0,\infty)\times\tilde{E})$ which is decomposed according to~\eqref{func.f.1} with the payoff function $u(t,x)=e^{\rho t}x$ if no default occurs by time $t\geq0$, and zero recovery $v(t,x)=0$ if the firm defaults prior to time $t\geq0$. The following theorem formally characterizes the stock price process.

\begin{theorem}\label{itostockcor} {\bf (Stock Price Process $S$)} Let the stock price process $S$ be specified in terms of the bivariate process $(X^\phi,D^\phi)$ by the prescription $S_t=e^{\rho t}(1-D^\phi_t)X^\phi_t$, where $X^\phi_0=x>0$ and $D_0^\phi=d\in \{0,1\}$. Moreover, assume that the scaling factor $\rho$ satisfies $\rho=\phi(-\mu)$ where $\phi(u)$ is the Laplace exponent~\eqref{eq:LevyKintchine} of the subordinator $T$, and where $\mu\in\Ic$ is the constant drift of the background process $X$ \eqref{eq:dX}--\eqref{eq:drift,vol} (the set $\Ic$ is defined in Eq.~\eqref{eq:set}). Then, the stock price process $S$ is a martingale with canonical representation
\begin{align}
S_t&={S_0+\sqrt{\gamma}\int_0^t\sigma(X_u^\phi)S_{u}dW_u+\int_0^t\int_{{\mathbb R}}e^{\rho s}(1-D^\phi_{s-})y\left(\hat{\mu}(ds,dy)-\pi^{\phi,1}(X^\phi_{s-},y)dyds\right)-\int_0^tS_{u-}dM_u^\phi,}
\label{stock.xd.1}
\end{align}
where $\gamma\geq0$ is the drift of the L\'evy subordinator $T$. The random measure $\hat\mu$ corresponds to those jumps of $X^\phi$ that do {\em not} coincide with jump of the default indicator $D^\phi$ (see Eq.~\eqref{teq.14}). The L\'evy density $\pi^{\phi,1}(x,y)$ is defined in Eq.~\eqref{sublevmeas.1}.
$W$ is a Brownian motion and $M^\phi$ is the martingale~\eqref{eqr.23} associated to $D^\phi$.
\end{theorem}
\begin{proof} From the restrictions on $\mu$ and $\rho$, the stock price $S_t$ is a discounted martingale (see, \cite{mendoza-carr-linetsky-1}, Section 4). In the presence of a constant interest rate $r\geq0$ and dividend yield $q\geq0$, the stock price can be decomposed as $S_t=S_0+A_t+\int_0^tS_{u-}dM_u$, where $M_t$ is a martingale and where $A_t=\int_0^t(r-q)S_{u-}du$ is predictable, and hence, it is a special semimartingale. Then, the canonical representation of $S$ follows from the It\^o formula of Theorem~\ref{cor.xD.1} applied to the function $f(t,x,d)$ defined by $u(t,x)=e^{\rho t}x$ and $v(t,x)=0$. We further observe that the drift vanishes since
\begin{align}  \big(\partial_s+\A^{\phi,1}\big)u(s,X^\phi_s,D^\phi_s)&=e^{\rho s}X^\phi_s(1-D^\phi_s) \label{drift.van.1}\\
&\times\Big[\rho+\gamma [\mu +k(X_s^\phi)]+(1/X^\phi_s)\int_{(0,\infty)}\Big(\int_{E} (z-X_s^\phi) p^1(w,X_s^\phi,z)dz\Big)\nu(dw)\label{drift.van.2}\\
&-\gamma k(X_s^\phi) - \int_{(0,\infty)}\Big(1-\int_{E}p^1(w,X_s^\phi,z)dz\Big)\nu(dw)\Big]\label{drift.van.3}\\
&=e^{\rho s}X^\phi_s(1-D^\phi_s)\Big[\rho+\gamma \mu+\int_{(0,\infty)}\Big((1/X^\phi_s)\int_{E} z p^1(w,X_s^\phi,z)dz-1\Big)\nu(dw)\Big]\label{drift.van.4}\\
&=e^{\rho s}X^\phi_s(1-D^\phi_s)\Big[\rho+\gamma \mu+\int_{(0,\infty)}\big(e^{\mu w}-1\big)\nu(dw)\Big]\label{drift.van.5}\\
&=e^{\rho s}X^\phi_s(1-D^\phi_s)\big(\rho-\phi(-\mu)\big)=0.\label{drift.van.6}\end{align}
In the third equality we used the fact that $\int_{E} z p^1(w,x,z)dz=\E[e^{-\int_0^t k(X_u)du}X_t]=x e^{\mu t}$ (cf.~\cite{linetsky-1}, Proposition 2.1). The fourth equality follows from the definition~\eqref{eq:LevyKintchine} of the Laplace exponent, which cancels from the condition $\rho=\phi(-\mu)$.
\end{proof}
This canonical representation decomposes the stock price process $S$ into a purely discontinuous  martingale of jumps prior to default with the compensator measure $(1-D_{u-}^\phi)\pi^{1,\phi}(X_{u-}^\phi,y)dydu$ (observe from Eq.~\eqref{teq.15} that $(1-D_{u-}^\phi)\hat{\nu}(du,dy)=(1-D_{u-}^\phi)\pi^{1,\phi}(X_{u-}^\phi,y)dydu$),
a continuous martingale component represented in terms of a Brownian motion,  and a final jump to zero (the default term $-\int_0^t S_{u-} dM_u^\phi$). Clearly, the process $S$ is a jump-diffusion process whenever $\gamma>0$, and a  purely discontinuous process for $\gamma=0$.

Next, we note that if the firm underlying $S$ were to default at some time $\zeta^\phi$ in the interval $[0,t]$, then the payoff of a traditional variance swap contract \eqref{eq:VS} would be infinite.  To account for this possibility, we modify the floating leg of the VS so that it only accumulates quadratic variation \emph{prior} to the default time $\zeta^\phi$.  That is, the long side of a VS, under our modified definition, has a payoff of
\begin{align}
[\log S]_{t \wedge \zeta^\phi-} - K_{var}
		&=		\int_0^t (1-D_u^\phi) d[\log S]_u - K_{var}	. \label{eq:VS1}
\end{align}
Notice that, for an asset that cannot default, our modified
definition of a VS coincides with the traditional definition of a VS.
Meanwhile, for an asset that \emph{can} default, our modified
definition of a VS is guaranteed to have a finite payoff, since the
floating leg of the modified VS only accumulates quadratic variation up the
time just prior to default.
\par
Using definition \eqref{eq:VS1}, the fair value of $K_{var}$ is the risk-neutral expectation of the floating leg  
\begin{align}
K_{var}
		&=		\E_x \[ \int_0^t (1-D_u^\phi) d[\log S]_u \] . \label{eq:Kvar}
\end{align}

An explicit expression for the right-hand side of \eqref{eq:Kvar} is given in the following theorem.
\begin{theorem}
\label{thm.QV}
Let $S$ be given by $S_t=e^{\rho t}X_t^\phi(1-D_t^\phi)$.  Then the right-hand side of \eqref{eq:Kvar} is given by
\begin{align}
K_{var}
		&=		\int_0^t \E_x \[ \I_{\{\zeta^\phi>u\}} \gamma \sig^2(X_u^\phi) \] du 
					+ \int_0^t  \E_x \[ \I_{\{\zeta^\phi>u\}}\int_{\mathbb R}\log^2\Big(1+\frac{y}{X^\phi_{u-}}\Big)\pi^{\phi,1}(X_{u-}^\phi,y)dy \] du .
					\label{eq:Kvar2}
\end{align}
\end{theorem}
\begin{proof}
In view of Eq.~\eqref{eq:Kvar}, it suffices to first calculate the L\'evy-It\^o canonical representation of the function $f(t,x,d)=(1-d)\log (e^{\rho t}x)$, which corresponds  to a zero--recovery  $\log$--contract on the stock price. That is, a contract that pays $u(t,X^\phi_t)=\log (e^{\rho t}X^\phi_t)=\log S_t$ if no default occurs by time $t\geq0$, and zero otherwise (i.e., we set $v(t,X^\phi_t)=0$). Hence, the canonical representation of the pre-default $\log$--contract of $S$ can be obtained by means of an application of the It\^o formula of Theorem~\ref{cor.xD.1} to the function $f(t,x,d)=(1-d)\log (e^{\rho t}x)$,
\begin{align}
d(\log S_t)=&\big[(1-D_{t-}^\phi)\big(\partial_t+\A^{\phi,1}\big)\log (e^{\rho t}X^\phi_{t-})\big] \, dt+ (1-D_{t-}^\phi)\, \sqrt{\gamma}\, \sig(X_{t-}^\phi) dW_t\label{eq:QV.temp.0} \\
&+(1-D_{t-}^\phi)\int_{\mathbb R}\log\Big(1+\frac{y}{X^\phi_{t-}}\Big)\left(\hat{\mu}(dt,dy)-\pi^{\phi,1}(X_{t-}^\phi,y)dydt\right)\\ & -(1-D_{t-}^\phi)\log (e^{\rho t}X^\phi_{t-})dM^\phi_{t}.
\label{eq:QV.temp} \end{align}
Observe that due to the default term $(1-D_{t-}^\phi)\log (e^{\rho t}X^\phi_{t-})dM^\phi_{t}=(1-D_{t-}^\phi)\log (S_{t-})dM^\phi_{t}$ the process jumps to zero at default, which is consistent with our selection of the function $f(t,x,d)=(1-d)\log (e^{\rho t}x)$ that has zero-recovery in case of default. Consequently, it describes the pre-default dynamics of $\log S_t$ and prevents it from exploding at default time. From \eqref{eq:QV.temp} it is straightforward to compute the differential $d[\log S]_t$ of the pre-default dynamics of $\log S_t$,
\begin{align}
d[\log S ]_t
		&=		(1-D_{t-}^\phi)\( \gamma\sig^2(X_t^\phi) +\int_{\mathbb R}\log^2\Big(1+\frac{y}{X^\phi_{t-}}\Big)\pi^{\phi,1}(X_{t-}^\phi,y)dy\) dt	\\ & \qquad
					+ (1-D_{t-}^\phi)\int_{\mathbb R}\log^2\Big(1+\frac{y}{X^\phi_{t-}}\Big)\left(\hat{\mu}(dt,dy)-\pi^{\phi,1}(X_{t-}^\phi,y)dydt\right) \\ &\qquad
					+ (1-D_{t-}^\phi)\log^2 (e^{\rho t}X^\phi_{t-})dD^\phi_{t} . \label{eq:QV.temp2}
\end{align}
Finally, multiplying \eqref{eq:QV.temp2} by $1-D_t^\phi = \I_{\{t<\zeta^\phi\}}$, observing that $(1-D^\phi_{t-})(1-D^\phi_{t})=(1-D^\phi_{t})$ and $(1-D^\phi_{t})dD^\phi_t=0$ a.s., integrating over the interval $[0,t]$, taking an expectation, and using the fact that the random measure $(1-D^\phi_{t-})\big(\hat{\mu}(dt,dy)-\pi^{\phi,1}(X_{t-}^\phi,y)dydt\big)$ is a  martingale measure, one arrives at \eqref{eq:Kvar2}. 
\end{proof}
Next, we give an alternative formulation of the value of $K_{var}$ in terms of Feynman-Kac semigroups.
\begin{proposition}\label{prop.semi.g.1} Let $S$ be given by $S_t=e^{\rho t}X_t^\phi(1-D_t^\phi)$ with $X_0=x>0$ and $D^\phi_0=d\in\{0,1\}$. Also, let $\Pc_t^{1}$ (resp., $\Pc_t^{\phi,1}$) be the (resp., the subordinate) Feynman-Kac semigroup defined in Eq.~\eqref{fk.1} (resp., Theorem~\ref{mkvian.sub.biv}). Then,  $K_{var}$  can be represented as follows, 
\begin{align}
K_{var}&=\I_{\{\zeta^\phi>0\}}\,\gamma\,\int_0^t (\Pc^{\phi,1}_u\sigma^2(x))du+\I_{\{\zeta^\phi>0\}}\int_0^t\Big(\int_{(0,\infty)}\Pc^{\phi,1}_uf(s,\cdot)(x)\nu(ds)\Big)du,
\label{semi.arrange.1}
\end{align}
with
\begin{align}
f(s,y)&=\Pc^{1}_s\log^2(y)-2\log(y)\Pc^{1}_s\log(y)+\log^2(y)\Pc^{1}_s1,\quad y=X^\phi_{u}.
\label{semi.arrange.2}
\end{align}
\end{proposition}
\begin{proof} From Theorem~\ref{mkvian.sub.biv} we observe that  $\E[(1-D^\phi_t)f(X^\phi_t)]=(1-d)\Pc_t^{\phi,1}f(x)$. Therefore, the first term of Eq.~\eqref{semi.arrange.1} follows immediately. From Proposition 32.5$(iii)$ in~\cite{sato}, p.215, we know that if $\|\Pc_t f(x)-f(x)\|= O(t)$ as $t\downarrow0$, then $\int_{\mathbb R} f(y)\pi^\phi(x,y)dy=\int_{(0,\infty)}(\Pc_sf(x)-f(x))\nu(ds)$. Since $f(y)=\log^2(y/x)=O(|x-y|^2)$ as $y\rightarrow x$, then to prove $\|\Pc^1_t f(x)-f(x)\|= O(t)$ it suffices to show that $\int_E(y-x)^2p^1(t,x,y)dy=O(t)$ as $t\downarrow0$. Indeed, the latter holds true since  for an arbitrary $\epsilon>0$, we have $\int_E \I_{\{|x-y|<\epsilon\}} (y-x)^2p^1(t,x,y)dy\leq C t$ as $t\downarrow0$ (cf., \cite{mckean-1}, Theorem 4.5). Therefore, 
\begin{align}
&\int_{\mathbb R}\log^2\Big(1+\frac{y}{X^\phi_{u-}}\Big)\pi^{\phi,1}(X_{u-}^\phi,y)dy \\
		&=\int_{(0,\infty)}\Big(\int_{E\setminus\{x\}}\log^2\Big(\frac{y}{X^\phi_{u-}}\Big)p^1(s,X_{u-}^\phi,y)dy\Big)\nu(ds)\\
		&=\int_{(0,\infty)}\big(\Pc^1_s\log^2(X_{u-})-2\log(X_{u-}) \Pc^1_s\log(X_{u-})+\log^2(X_{u-}) \Pc^1_s1\big)\nu(ds).
\end{align}
The rest follows from observing that $\E[(1-D^\phi_u)f(s,X^\phi_u)]=(1-d)\Pc^{\phi,1}_uf(s,x)$.
\end{proof}

%
%

\section{Spectral Expansions}
\label{sec:spectral0}
In order for the results of Sections \ref{sem.gens.ito} and \ref{sec:VS} to be useful, we need a practical way to construct the FK semigroups $\{\Pc^1_t,t\geq 0\}$ and $\{\Pc_t^{\phi,1},t\geq 0\}$ as well as the associated transition densities $p^1(t,x,y)$ and $p^{\phi,1}(t,x,y)$ (recall that we had dropped the super-index $\alpha$, since we only need the case of $\alpha=1$).  Spectral theory, or more specifically, the theory of eigenfunction expansions, provides a straightforward method of constructing these operators and functions.  Below, we review some useful results relating to eigenfunction expansions.  A detailed description of the spectral theorem for self-adjoint operators in a Hilbert space is given in Appendix \ref{sec:spectral}.
\par
Recall that the FK semigroup $\{\Pc^1_t, t \geq 0\}$ has infinitesimal generator $\A^1$ \eqref{infg.1}.  With $\A^1$ we associate a \emph{scale density} $\s$ and \emph{speed density} $\m$
\begin{align}
\s(x)
		&:=		\exp \( - \int_{x_0}^x \frac{2 \, b(y)}{a^2(y)} dy \) , &
\m(x)
		&:=		\frac{2}{a^2(x)}\exp \( \int_{x_0}^x \frac{2 \, b(y)}{a^2(y)} dy \) ,	\label{eq:m}
\end{align}
where the point $x_0$ is arbitrarily chosen in $E=(0,\infty)$.  The generator $\A^1$ \eqref{infg.1} with domain
\begin{align}
{\rm Dom}(\A^1)
		&=		\{ f \in L^2(E,\m) : \A^1 f \in L^2(E,\m) \} , 
\end{align}
is a self-adjoint operator in the Hilbert space $\H = L^2(E,\m)$. 
\footnote{$\A^1$ is dense in $\H$ implies that $\A^1$ has a unique self-adjoint extension $\overline{\A^1}$ with ${\rm Dom}(\overline{\A^1})=\H$.  We will not distinguish between $\A^1$ and its extension $\overline{\A^1}$.}  Therefore, we have spectral representations for the operators $\A^1$ and $g(\A^1)$, where $g$ is any Borel-measurable function.
\par
Let $\psi_\lam$ and $\lam$ be the generalized eigenfunctions/values of  $-\A^1$.  Note that, since $\A^1$ is the generator of a contraction semigroup $\{\Pc^1_t, t \geq 0\}$, the eigenvalues of $-\A^1$ are non-negative.  The operator $\Pc^1_t$ can be written as $\Pc^1_t = e^{t\A^1}$ (for a general Banach space, when the generator $\A^1$ is unbounded the latter is understood as a strong limit via the Yosida approximation (see \cite{pazy-1}, Corollary 3.5)).  Thus, using \eqref{eq:g.A}, for any $f \in \H$ we have
\begin{align}
\Pc^1_t f(x)
		&=		\sum_\lam e^{-\lam t} c_\lam \psi_\lam(x) , &
c_\lam
		&=		\(\psi_\lam,f\)	, &
\(\psi_\lam,f\)
		&=		\int_E \overline{\psi_\lam}(x) f(x) \m(x) dx , \label{eq:Pt.eigen}
\end{align}
where $\overline{\psi_\lam}$ indicates the complex conjugate of $\psi_\lam$.  The notation $\sum_\lam (\cdots)$ is shorthand for
\begin{align}
\sum_\lam \, (\cdots)
	&=		\sum_{\lam_n \in \sig_d(-\A^1)} (\cdots) + \int_{\lam_\om \in \sig_c(-\A^1)} (\cdots) \, d\om ,
\end{align}
where $\sig_d(-\A^1)$ and $\sig_c(-\A^1)$ are the discrete and continuous portions of the spectrum of $-\A^1$ respectively and $d\om$ is the Lebesgue measure.

Similarly, using the functional calculus of Theorem~\ref{thm:spectral}, the subordinated semigroup $\Pc_t^{\phi,1}$ defined in Eq.~\eqref{subsem.1} can be obtained as,
\begin{align}
\Pc_t^{\phi,1} f(x)
		&=		\sum_\lam e^{-\phi(\lam) t} c_\lam \psi_\lam(x) , &
c_\lam
		&=		\(\psi_\lam,f\)	, &
\(\psi_\lam,f\)
		&=		\int_E \overline{\psi_\lam}(x) f(x) \m(x) dx , \label{eq:Pt.eigen.sub}
\end{align}
where $\phi(\lam)$ is the Laplace exponent of $T$, defined in Eq.~\eqref{eq:LevyKintchine}. One should mention that the recent book of \cite{schilling-song-vondraceck-1} is an excellent reference for Bochner subordination of semigroups (for example, the last result above is obtained from their Remark 12.4, p.133).

\subsection{Uniform convergence of the discrete spectrum}\label{sec.trace.class}  
In general, when the spectrum is discrete, the spectral expansion~\eqref{eq:Pt.eigen} of the semigroup $\Pc^1$ (and hence, $\Pc^{\phi,1}$)  leads to an infinite series. When the semigroup $\Pc^1$ is of trace class, then it is possible to establish uniform convergence for the expansions as follows.  Assume that the eigenvalues of $-\A^1$ satisfy the condition
\begin{align}
\sum_\lam e^{-\lam t}
		&<		\infty , &
\forall \, t
		&>		0 , \label{eq:trace}
\end{align}
so that the FK semigroup $\{\Pc^1_t,t\geq0\}$ is \emph{trace class} (see Section 7.2 of \cite{davies-3}).  According to Theorem 7.2.5 of \cite{davies-3}, if $\{\Pc^1_t,t\geq0\}$ is trace class, then the eigenfunctions $\psi_\lam(x)$ are continuous functions with the global estimate $|\psi_\lam(x)| \leq e^{\lam t/2} \sqrt{p^1(t,x,x)/\m(x)}$ for all $t\geq0$. Setting $f(x) = \del_y(x)$ in \eqref{eq:Pt.eigen} yields the transition density $p^1(t,x,y)$ of the FK semigroup.
\begin{align}
p^1(t,x,y)
		&=	\m(y) \sum_\lam e^{- \lam t} \psi_\lam(x) \psib_\lam(y) . \label{eq:ptxy.eigen}
\end{align}
The sum on the right-hand side of \eqref{eq:ptxy.eigen} converges uniformly in $x$ and $y$ on compacts.  This ensures that, in addition to the $L^2$ convergence, the eigenfunction expansion \eqref{eq:Pt.eigen} converges uniformly in $x$ on compacts for all $f \in L^2(E,\m)$ and $t>0$.  This follows from the Cauchy-Schwarz bound for the expansion coefficients $|c_\lam|\leq \sqrt{(f,f)}$, the eigenfunction estimate, and the trace class condition \eqref{eq:trace}.
\par
In this case, the spectral expansion for the subordinated FK semigroup $\{\Pc_t^{\phi,1},t \geq 0\}$ can be obtained by conditioning on the subordinator $T_t$.  For any $f \in \H$ we have
\begin{align}
\Pc_t^{\phi,1} f(x)
		&=		\E_x \[ \Pc^1_{T_t} f(x) \]
		=			\E \[ \E_x \[\Pc^1_{T_t} f(x) | T_t \] \]
		=		\sum_\lam \E \[ e^{-\lam T_t} \] c_\lam \psi_\lam(x)
		=			\sum_\lam e^{-\phi(\lam) t} c_\lam \psi_\lam(x) , \label{eq:Pphi.eigen}
\end{align}
where $c_\lam = \(\psi_\lam,f\)$ and $\phi(\lam)$ is the L\'{e}vy exponent of the subordinator $T$.  If we assume that the Laplace exponent $\phi$ is such that
\begin{align}
\sum_\lam e^{-\phi(\lam) t}
		&<		\infty , &
\forall \, t
		&>		0 , \label{eq:trace2}
\end{align}
then the subordinated FK semigroup $\{\Pc_t^{\phi,1},t\geq0\}$ is trace class.  If we further assume that the eigenfunctions $\psi_\lam$ of the FK semigroup $\{\Pc^1_t,t\geq0\}$ have a bound independent of $\lam$ on each compact interval $K=[a,b]\subset E$ (that is, if there exist constants $C_K$, which depend on the compact interval $K$ but are {\em independent} of $\lam$, such that $|\psi_\lam(x)|\leq C_K$ for all $\lam \in \sig_d(-\A^1)$) then, in addition to the $L^2$ convergence, the eigenfunction expansion of the subordinated FK semigroup \eqref{eq:Pphi.eigen} converges uniformly in $x$ on compacts for all $f \in L^2(E,\m)$ and $t>0$.  As above, setting $f(x) = \del_y(x)$ in \eqref{eq:Pphi.eigen} yields the transition density $p^{\phi,1}(t,x,y)$ of the subordinated FK semigroup
\begin{align}
p^{\phi,1}(t,x,y)
		&=	\m(y) \sum_\lam e^{- \phi(\lam) t} \psi_\lam(x) \psib_\lam(y) . \label{eq:pphitxy}
\end{align}
The sum in \eqref{eq:pphitxy} is uniformly convergent on compacts in $x$ and $y$. Note that the semigroup corresponding to the JDCEV process of Section \ref{sec:JDCEV} is of trace class.

%
%

\section{Examples}
\label{sec:examples}
In this Section, we compute $K_{var}$ \eqref{eq:Kvar2} explicitly (up to an integral with respect to the L\'evy measure $\nu$ of the subordinator $T$), when the background Feller diffusion $X$ is modeled as (i) a geometric Brownian motion with constant killing rate and (ii) a Jump-to-default Constant Elasticity of Variance process.


\subsection{Example: geometric Brownian motion with default}
\label{sec:GBM}
Perhaps the most widely recognized non-negative diffusion in finance is the geometric Brownian motion process (GBM).  Here, we consider GBM with a constant killing rate, which is a diffusion of the form \eqref{eq:dX}-\eqref{eq:drift,vol} with constant parameters $k(x)=k\geq0$ and $\sig(x)=\sig>0$ (excuse the abuse of notation).  The generator $\A^1$ of the FK semigroup $\{\Pc^1_t,t \geq 0\}$ and the corresponding speed density $\m(x)$ are given by
\begin{align}
\A^1
		&=		\tfrac{1}{2} \sig^2 x^2 \, \d_{xx}^2 + (\mu + k)x \, \d_x - k, \label{eq:A.GBM} \\
\m(x)
		&=		\frac{2}{x\sig^2} \exp \( 2 \, \xi \frac{1}{\sig} \log x \) , &
\xi
		&=		\frac{\mu+k}{\sig}-\frac{\sig}{2} . \label{eq:m.GBM}
\end{align}
Most commonly, the FK transition density $p^1(t,x,y)$ of GBM with default is written as
\begin{align}
p^1(t,x,y)
		&=		\frac{e^{-kt}}{y\sig \sqrt{2\pi t}}\exp\(\frac{ -\( \log y - \log x -( \mu + k - \sig^2/2)t \)^2}{2 \sig^2 t}\). \label{eq:p.GBM}
\end{align}
Due to the fact that $\A^1$ \eqref{eq:A.GBM} is a self-adjoint operator on the Hilbert space $\H=L^2(E,\m)$, with $\m$ given by \eqref{eq:m.GBM}, the FK transition density $p^1(t,x,y)$ also has an (generalized) eigenfunction expansion of the form \eqref{eq:ptxy.eigen}.  The eigenfunctions and eigenvalues of $-\A^1$ are given in the following Theorem.
\begin{theorem}[GBM Eigenvalues and Eigenfunctions]
Let the operator $\A^1$ be given by \eqref{eq:A.GBM}.  The spectrum of $\A^1$  is purely continuous: $\sig(\A^1)=\sig_c(\A^1)$.  The improper eigenfunctions of $-\A^1$ and the corresponding improper eigenvalues are
\begin{align}
\psi_\om(x)
		&=		\sqrt{\frac{\sig}{4 \pi}} \exp \( ( i \om -\xi ) \frac{1}{\sig }\log x \) , &
\lam_\om
		&=		\frac{1}{2} \( \om^2 + \xi^2 \) + k, & 
\om
		&\in	(-\infty,\infty) , \label{eq:eigen.GBM}
\end{align}
where $\xi$ is given in \eqref{eq:m.GBM}.
\end{theorem}
\begin{proof}
One can check directly that the eigenfunctions and eigenvalues of \eqref{eq:eigen.GBM} satisfy the improper eigenvalue equation $-\A^1 \psi_\om = \lam_\om \psi_\om$ and the boundedness condition \eqref{eq:bound}.  The orthogonality relation $\(\psi_\om,\psi_\nu\)=\del(\om-\nu)$ follows by noting that $\tfrac{1}{2\pi}\int e^{-i(\om-\nu)x}dx=\del(\om-\nu)$.
\end{proof}
\begin{remark}
Transition density \eqref{eq:p.GBM} can be obtained by writing eigenfunction expansion \eqref{eq:ptxy.eigen} with the eigenfunctions and eigenvalues of \eqref{eq:eigen.GBM}, making a change of variables $z=\tfrac{1}{\sig}\log y$, and using the Fourier transform of a Gaussian density
\begin{align}
\frac{1}{\sqrt{2\pi a^2}}\exp\( \frac{-(z-b)^2}{2 a^2}\)
		&=		\int_{(-\infty,\infty)} \frac{1}{2 \pi} \exp\(i \( b - z \)  \om - \frac{a^2 \om^2}{2} \) d\om .
\end{align}
\end{remark}
\begin{remark}
When the underlying $S$ is given by \eqref{eq:S} and the background diffusion $X$ is modeled as a GBM with default, the at-the-money skew of the model-induced implied volatility surface is controlled by $\mu$.  For $\mu<0$ jumps in $S$ will be preferentially downward, causing a negative at-the-money skew.  For $\mu>0$, jumps in $S$ will be preferentially upward, causing a positive at-the-money skew.  As the skew for equity options is typically negative, it makes sense to choose $\mu<0$.
\end{remark}
\noindent
We are now in position to compute \eqref{eq:Kvar2} when $X$ is modeled as a GBM with default.
\begin{proposition}
\label{GBM.QV}
Let $X$ be a GBM process with default as described above.  Then we have $(i)$
\begin{align}
\int_0^t \E_x \[ \I_{\{\zeta^\phi>u\}} \gamma \sig^2(X_u^\phi) \] du 
		&=			\gamma \, \sig^2 \( \frac{1-e^{-k t}}{k} \) ,
\end{align}
and $(ii)$
\begin{align}
\int_0^t  \E_x \[ \I_{\{\zeta^\phi>u\}}\int_{\mathbb R}\log^2\Big(1+\frac{y}{X^\phi_{u-}}\Big)\pi^{\phi,1}(X_{u-}^\phi,y)dy \] du
		&=\sig^2 \( \frac{1 - e^{-\phi(k) t }}{\phi(k)} \) \int_{(0,\infty)}e^{-k s} \(  s + s^2 \xi^2 \) \nu(ds) .
\end{align}
\end{proposition}
\begin{proof}
See Appendix \ref{GBM.QV.proof}.
\end{proof}


\subsection{Example: Jump-to-Default constant elasticity of variance}
\label{sec:JDCEV}
The Constant Elasticity of Variance (CEV) model of \cite{cox-1} is a non-negative diffusion of the form \eqref{eq:dX}-\eqref{eq:drift,vol}, where $k(x)=0$ and 
\begin{align}
\sig(x) 
		&=		a x^\beta .
\end{align}
Here, $\beta < 0$ is the volatility elasticity parameter and $a>0$ is the volatility scale parameter.  The specification $\beta < 0$ is consistent with the leverage effect (volatility increases when the stock price falls).  For $\beta < 0$ the CEV process hits zero with positive probability.  In particular, for $\beta \in [-1/2,0)$, the origin is an exit boundary.  For $\beta < -1/2$ the origin is a regular boundary specified as a killing boundary.
\par
\cite{carr-linetsky-1} extend the CEV model to include a possible jump-to-default.  Their model is refereed to as jump-to-default CEV or, more succinctly, JDCEV.  In the JDCEV framework, the jump to default has a killing rate which is an affine functions of the local variance
\begin{align}
k(x)
		&=		b + c \sig^2(x)
		=			b + c a^2 x^{2 \beta} ,
\end{align}
where $b \geq 0$ and $c \geq 0$.  Although for all $c>0$ default may only occur through a jump from a positive value.  When $c\geq1/2$ the zero boundary is {\em entrance} for the JDCEV diffusion, and thus, the diffusion cannot reach zero from the interior of $E$. The majority of the expressions developed in this Section hold for all $c>0$. However, {\em the credit-equity modeling framework developed in Sections~\ref{sec:model}--\ref{sec:VS} works exclusively for the case in which $c\geq1/2$} (i.e., the case in which zero is an entrance boundary). Therefore, one should keep in mind this restriction when applying the following more general results. 

For a JDCEV diffusion, the generator $\A^1$ of the FK semigroup $\Pc^1$ and the corresponding speed density are given by
\begin{align}
\A^1
		&=		\tfrac{1}{2}a^2 x^{2\beta+2} \, \d_{xx}^2 + \( \mu + b + c a^2 x^{2 \beta} \) x \, \d_x - \( b + c a^2 x^{2 \beta} \) .
					\label{eq:A.JDCEV} \\
\m(x)
		&=		\frac{2}{a^2} x^{2c-2-2\beta} \exp \( \eps \, A \, x^{-2\beta}\) , \qquad \qquad
A
		=		\frac{|\mu+b|}{a^2|\beta|} , \qquad \qquad
\eps
		=		\text{sign}(\mu + b) . \label{eq:JDCEV.m}
\end{align}
The FK transition density $p^1(t,x,y)$ for the JDCEV diffusion was obtained by \cite{carr-linetsky-1}
\begin{align}
p^1(t,x,y)
		&=		\frac{\m(x) |\mu+b|(xy)^{\tfrac{1}{2}-c}e^{\om \nu t /2}}{1-e^{-\om t}}
					\exp \( -\eps A \frac{x^{-2\beta}+y^{-2\beta}}{1-e^{-\eps \om t}} - \lam_1 t\)
					I_\nu\(\frac{A(xy)^{-\beta}}{\sinh (\om t/2)}\) , \label{eq:p.JDCEV}
\end{align}
where $I_\nu$ is the modified Bessel function of order $\nu$, the constants $A$ and $\eps$ are given in \eqref{eq:JDCEV.m} and
\begin{align}
\nu
		&=		\frac{1 + 2c}{2 |\beta|} , &
\om
		&=		2 |\beta(\mu+b)| , &
\lam_1
		&=		\begin{cases}
					2(\mu+b)(|\beta|+c)+b, &\mu+b>0 \\
					|\mu|, &\mu+b<0
					\end{cases} . \label{eq:nu.om.lam1}
\end{align}
Due to the fact that the unique extension of $\A^1$ is a self-adjoint operator in the Hilbert space $\H = L^2(E,\m)$ with $\m(x)$ given by \eqref{eq:JDCEV.m}, the FK transition density \eqref{eq:p.JDCEV} has an eigenfunction expansion of the form \eqref{eq:ptxy.eigen}.  The eigenfunctions and eigenvalues of $-\A^1$ are given in the following theorem, which is due to \cite{linetsky-mendoza-arriaga-1}.
\begin{theorem}[JDCEV Eigenvalues and Eigenfunctions]
\label{prop.jdcev.2}
Let $\A^1$ be given by \eqref{eq:A.JDCEV}.  When $|\mu+b| \neq 0$, the spectrum of $\A^1$ is purely discrete: $\sig(\A^1)=\sig_d(\A^1)$.  The eigenfunctions of $-\A^1$ and the corresponding eigenvalues are:
\begin{align}
\psi_n(x)
		&=		A^{\nu/2} \sqrt{\frac{(n+1)!|\mu+b|}{\Gam(\nu+n)}} 
					x \exp \( -\tfrac{1}{2}(1+\eps) A x^{-2\beta}\)
					L_{n-1}^\nu (A x^{-2\beta}) , & \label{eq:eigenfunc.JDCEV} \\
\lam_n
		&=		\om (n-1) + \lam_1 , \qquad n =	1, 2, 3, \cdots , \label{eq:eigenval.JDCEV}
\end{align}
where $L_n^\nu$ are the generalized Laguerre polynomials, $A$ and $\eps$ are given in \eqref{eq:JDCEV.m} and $\nu$, $\om$ and $\lam_1$ are given in \eqref{eq:nu.om.lam1}.
\end{theorem}
\begin{proof}
One can verify directly that the eigenfunctions and eigenvalues \eqref{eq:eigenfunc.JDCEV}-\eqref{eq:eigenval.JDCEV} satisfy the (proper) eigenvalues equation $-\A^1 \psi_n = \lam_n \psi_n$.  Orthogonality of the eigenfunctions $(\psi_n,\psi_m)=\del_{n,m}$ follows from the orthogonality relations of the generalized Laguerre polynomials (see \cite{abramowitz-stegun}, pp. 775)
\begin{align}
\int_0^\infty x^\alpha e^{-x} L_n^\alpha(x)L_m^\alpha(x) dx
		&=		\frac{\Gam(n+\alpha+1)}{n!}\del_{n,m} .
\end{align}
\end{proof}
\begin{remark}
The transition density \eqref{eq:p.JDCEV} can be recovered from the eigenfunction expansion \eqref{eq:ptxy.eigen} with the eigenfunctions and eigenvalues defined in \eqref{eq:eigenfunc.JDCEV}-\eqref{eq:eigenval.JDCEV} by means of the Hille-Hardy formula (see
\cite{erdelyi-2}, p.189):
\begin{align}
\sum_{n=0}^\infty \frac{t^n n!}{\Gam(n+\nu+1)}L_n^\nu(a)L_n^\nu(b)
		&=		\frac{(abt)^{-\nu/2}}{1-t}\exp\(-\frac{(a+b)t}{1-t}\)I_\nu\(\frac{2\sqrt{abt}}{1-t}\) ,
\end{align}
which is valid for all $|t|<1$, $\nu>-1$ and $a,b>0$.
\end{remark}
\noindent
We are now equipped to compute $K_{var}$ \eqref{eq:Kvar2} when the background process $X$ is a JDCEV diffusion.  We shall focus specifically on the case $\mu+b<0$, since in this case, all of the relevant functions $f$ are in $L^2(E,\m)$ with $\m(x)$ given by \eqref{eq:JDCEV.m}.  Thus, we can compute all the necessary expectations explicitly using the eigenfunction expansion techniques of Section \ref{sec:spectral0}.
\par
Before computing the expectations in \eqref{eq:Kvar2} it will be useful to give an analytical solution for the $p$-th moment of the stock price
\begin{align}
\E_x\Big[(S_t)^p\Big]
=e^{p \rho t}\,\E_x \Big[(X_t^\phi)^p \I_{\{\zeta^\phi>t\}}\Big]
=e^{p \rho t}\,\Pc_t^{\phi,1} x^p.
\label{mom.0}
\end{align}
\begin{proposition}[$p$-th Moment]
\label{mom.p}
Let the diffusion $\{X_t,t\geq0\}$ be a JDCEV process with parameters $\beta <0$, $a >0$, $b \geq 0$, and $c > 0$.  Assume $\mu +b <0$.  Then, $(i)$ for $p>2(\beta -c )$, the expected value of the function $f(x)=x ^p$ is given by the eigenfunction expansion:
\begin{align}
\E_x \Big[(S_t )^p\Big]
= e^{p \rho t}\sum_{n=1}^{\infty}e^{-\phi (\lambda_n )t} {\tilde c}_n  \psi_n (x ) , \label{mom.1}
\end{align}
where  $\phi(\lam)$ is the Laplace exponent of the subordinator $T$. The JDCEV eigenvalues $\lambda_n$ and eigenfunctions $\psi_n$, are given in theorem \ref{prop.jdcev.2}, and the expansion coefficients are given by:
\begin{align} 
{\tilde c}_n =(x ^p,\psi_n )=\frac{A ^{\frac{\nu }{2}-\frac{p+2c }{2|\beta |}}\left(\frac{1-p}{2|\beta |}\right)_{n-1} }{\sqrt{(n-1)! |\mu +b |\Gamma(\nu +n)}}\Gamma\left(\frac{p+2c }{2|\beta |}+1\right),\quad n=1,2,\cdots,
\label{mom.2}
\end{align}
where $(z)_n=z(z-1)\cdots(z-n-1)$ is the Pochhammer symbol. Also, $(ii)$ the spectral expansion is uniformly convergent for all $t>0$, absolutely convergent at $t=0$, and uniformly convergent at $t=0$ if $p>(\beta+1)/2 -c$.
\end{proposition}
\begin{proof}
The proof of part $(i)$ is obtained from Lemma 3.1 and Proposition 3.1 in~\cite{linetsky-mendoza-arriaga-1}. For part $(ii)$ we note that the semigroup is of trace class. In addition, we note that for $x\in[a,b]\subset E$, the eigenfunctions $\psi_n(x)$ satisfy the bound $|\psi_n(x)|< C/n^{1/4}<C$ for some $C<\infty$ independent of $n$ although it may depend on the range $[a,b]$ (see inequality (27a) on p.54 of \cite{nikiforov-uvarov-1}).  Moreover, since the expansion coefficients satisfy the Cauchy-Schwartz bound, $|c_n|\leq\sqrt{(f,f)}$, then for any $f\in L^2((0,\infty),\m)$ the spectral expansion of $\Pc_t^{\phi,1} f$ converges uniformly for all $t>0$. That is, $\sum_{n=1}^\infty e^{-\lambda_n t}|c_n\psi_n(x)|\leq C\sqrt{(f,f)}\sum_{n=1}^\infty e^{-\lambda_n t}$ converges uniformly for all $t>0$. In Appendix \ref{sec:coeffs} we show that the spectral expansion of $\Pc_t^{\phi,1} x^p$ also converges absolutely at $t=0$.  In addition, if $p>(\beta+1)/2 -c$, then {the spectral expansion} converges uniformly at $t=0$.
\end{proof}

\begin{proposition}
\label{shortInt}
Let $X$ be a JDCEV process with parameters $\beta <0$, $a >0$, $b \geq 0$, and $c > 0$.   Assume $\mu +b <0$. Then we have 
\begin{align}
&\gamma\int_0^t\E_x \Big[ \I_{\{\zeta^\phi>u\}}\sigma^2(X^\phi_{u})\Big]du \\
&\qquad=\gamma a^2A ^{\frac{\nu }{2}-\frac{c }{|\beta |}+1}\Gamma\left(c /|\beta |\right)\sum_{n=1}^{\infty} \frac{\left(1/(2|\beta |)+1\right)_{n-1}}{\sqrt{(n-1)! |\mu +b |\Gamma(\nu +n)}} \frac{ \big(1-e^{-\phi (\lambda_n )t}\big)\psi_n (x )}{\phi (\lambda_n )} .
\label{volssig.1}
\end{align}
\end{proposition}
\begin{proof}
The expectation in \eqref{volssig.1} can be written explicitly as
\begin{align}
&\gamma a^2\int_0^t\E_x\Big[ \I_{\{\zeta^\phi>u\}}(X^\phi_{u})^{2\beta}\Big]du\\
&\qquad=\gamma a^2A ^{\frac{\nu }{2}-\frac{c }{|\beta |}+1}\Gamma\left(c /|\beta |\right)\int_0^t\Big(\sum_{n=1}^{\infty}e^{-\phi (\lambda_n )u} \frac{\left(1/(2|\beta |)+1\right)_{n-1} }{\sqrt{(n-1)! |\mu +b |\Gamma(\nu +n)}} \psi_n (x )\Big)du ,
\end{align}
where the last equality is due to Proposition~\ref{mom.p}. One should note that if $3\beta+2c>1$ then the sum converges uniformly at $t=0$ and the integral can be done term by term. Otherwise, observe that: (a) the series inside the integral is absolutely convergent for all $t\geq0$ due to Proposition \ref{mom.p} (and continuous for all $t\geq0$), and (b) the Laplace exponent $\phi(\lam)$ is increasing. Then, we can conclude that the resulting series~\eqref{volssig.1} is also absolutely convergent, and hence the exchange of sum and integral is justified (i.e., we integrate term by term with $\int_\epsilon^t\cdot ds$ for some $\epsilon>0$ and then take the limit as $\epsilon\downarrow0$).
\end{proof}

\begin{proposition}
\label{longInt}
Let $X$ be a JDCEV process with parameters $\beta <0$, $a >0$, $b \geq 0$, and $c > 0$.  Assume $\mu +b <0$ and $2c-2\beta>1$.  Then we have 
\begin{align}
&\int_0^t  \E_x \[ \I_{\{\zeta^\phi>u\}}\int_{\mathbb R}\log^2\Big(1+\frac{y}{X^\phi_{u-}}\Big)\pi^{\phi,1}(X_{u-}^\phi,y)dy \] du \\
&\qquad=\frac{A^{\frac{1-2c}{4|\beta|}}}{4|\beta|^2}\int_{\R_+\setminus\{0\}}\sum_{n=1}^\infty\sum_{m=1}^\infty \frac{e^{-\lambda_n s}(1-e^{-\phi(\lambda_m)t})}{\phi(\lambda_m)}\psi_{m}(x)\,\times\\
&\qquad\qquad
\times\bigg\{\sqrt{\frac{(n-1)! }{ |\mu+b|\Gamma(\nu+n)}}\Theta^2_{n-1}\Big(\frac{c+|\beta|}{|\beta|}\Big) \\
&\qquad\qquad\qquad
-2\sqrt{\frac{(m-1)!  }{|\mu+b|\Gamma(\nu+m)}}\sum_{k=0}^{n-1}\frac{(1-n)_{k}\Theta^1_{m-1}(\nu+k+1)\Theta^1_{n-1}\Big(\frac{c+|\beta|}{|\beta|}\Big)}{\Gamma(\nu+k+1)k!}\\
&\qquad\qquad\qquad+\frac{\Big(\frac{1}{2|\beta |}\Big)_{n-1}\,\Gamma\Big(\frac{c }{|\beta |}+1\Big) }{(n-1)!}\sqrt{\frac{(m-1)!  }{|\mu +b |\Gamma(\nu+m)}}\sum_{k=0}^{n-1}\frac{(1-n)_{k}\Theta^2_{m-1}(\nu+k+1)}{\Gamma(\nu+k+1)k!}\bigg\}\nu(ds) ,
\label{longint.3}
\end{align}
where
\begin{align}
\Theta^1_{n}(\alpha)
		&:=		\frac{(1-\alpha+\nu)_{n}}{n!}\Gamma(\alpha)\Big[\psi(\alpha)
					-\sum_{p=1}^{n}\frac{1}{(p-\alpha+\nu)}\Big] ,	\label{prud.ln.1a}\\
\Theta^2_{n}(\alpha)
		&:=		\frac{(1-\alpha+\nu)_{n}}{n!}\Gamma(\alpha)\Big[\Big(\psi(\alpha)-\sum_{p=1}^{n}\frac{1}{p-\alpha+\nu}\Big)^2
					+\psi'(\alpha)-\sum_{p=1}^{n}\frac{1}{(p-\alpha+\nu)^2}\Big] , \label{prud.ln.2a}
\end{align}
and $\psi(\alpha)=\Gamma'(\alpha)/\Gamma(\alpha)$ (with no subscript on $\psi$) is the Polygamma function.
\end{proposition}
\begin{proof}
We start by mentioning that the functions $\Theta^\delta_{n}$ result from the integrals
\begin{align}
\Theta^\delta_{n}(\alpha)
	&=	\int_{0}^{\infty}z^{\alpha -1} e^{-z}\log^\delta(z)L_{n}^{\nu}(z)dz,  &
\delta
	&\in\{1,2\} , &
\text{Re}(\alpha)
	&>0 , \label{eq:Theta}
\end{align}
which are available in \cite{prudnikov-2},  Eq. 2.19.6.1 and 2.19.6.3, pp.469.  The restriction $2c-2\beta>1$ is imposed such that the sums converge absolutely at $t=0$. Indeed, it is easy to show that for all $x>0$ we have $|\log(x)|\leq (1/x+x)$ and $\log^2(x)\leq (1/x+x)$. Moreover, since the expansion of $(1/x+x)$ converges absolutely at $t=0$ due to Proposition~\ref{mom.p}, then each of the series also converge absolutely at $t=0$. The rest of the proof consists of integrating term by term.  Details are found in Appendix \ref{sec:long.proof}.
\end{proof}

%
%

\section{Market Implied L\'{e}vy Subordinators}
\label{sec:impliedtime}
Note that the value of $K_{var}$ \eqref{eq:Kvar2} depends on the drift $\gamma$ and L\'{e}vy measure $\nu$ of the L\'{e}vy subordinator $T$.  One could, of course, compute the value of $K_{var}$ by choosing a specific drift $\gamma$ and L\'evy measure $\nu$.  However, this parametric approach would lead to a considerable amount of model misspecification risk, as there is no guarantee that the chosen subordinator would induce European option prices consistent with those observed on the market.  An alternative approach would be to use knowledge of (liquidly traded and efficiently priced) European call and put options to constrain one's choice of L\'evy subordinator.  In this section we will show that, when the background diffusion $X$ is fixed, the drift $\gamma$ and L\'evy measure $\nu$ of the subordinator $T$ can be obtained \emph{non-parametrically} from the implied volatility smile of $t$-expiry European options.  This approach greatly reduces model misspecification risk, as the obtained subordinator induces $t$-expiry option prices that are consistent with those observed on the market.
\par
Let $S$ be described by \eqref{eq:S}.  We assume that background diffusion $X$ is given, but that the drift $\gamma$ and L\'{e}vy measure $\nu$ of the L\'{e}vy subordinator $T$ are unknown.  Denote by $C(t,x;K)$ the price of a European call option with time to maturity $t$ and strike price $K$.  Note that the price of a call with strike price $K$ can be obtained from the price of a put with the same strike through put-call parity.  We assume the existence of European call options at all strikes $K \in (0,\infty)$.  While calls at all strikes $K \in (0,\infty)$ do not trade in practice, \cite{bondarenko-1} shows how to estimate the value of call at any strike, given the value of calls at a discrete set of strikes.
\par
Let $p_S(t,x,y)$ be the transition density of $S$ under the risk-neutral pricing measure
\begin{align}
p_S(t,x,y) dy
		&=		\P_x \[ S_t \in dy \] .
\end{align}
Note that $S_t \in dy$ if and only if $(1-D_t^\phi)e^{\rho t}X_t^\phi \in dy$.  Thus, 
\begin{align}
p_S(t,x,y)
		&=	e^{-\rho t} p^{\phi,1}(t,x,y') 
		=	e^{-\rho t} \m(y') \sum_\lam e^{- \phi(\lam) t } \psi_\lam(x) \psib_\lam(y') , &
y'
		&:=		y\,e^{-\rho t} .\label{eq:q=qphi}
\end{align}
As \cite{breeden-litzenberger-1} show, the transition density $p_S(t,x,y)$ can be implied from a semi-infinite strip of call prices.  We have
\begin{align}
C(t,x;K)
		&=		\E_x \[ (S_t - K)^+ \] 
		=			\int_E (y - K)^+ p_S(t,x,y) dy . \label{eq:C=p}
\end{align}
Differentiating both sides of \eqref{eq:C=p} twice with respect to $K$, and noting that $\d_{KK}^2 (y-K)^+ = \delta(y-K)$, one obtains
\begin{align}
\d_{KK}^2 C(t,x;K)
		&=		\int_E \d_{KK}^2 (y-K)^+ p_S(t,x,y) dy
		=			p_S(t,x,K) . \label{eq:q=dKKC}
\end{align}
Setting our the two expressions \eqref{eq:q=qphi} and \eqref{eq:q=dKKC} for $p_S$ equal to each other yields
\begin{align}
\d_{KK}^2 C(t,x;K)
		&=		e^{-\rho t} \m(K') \sum_\lam e^{-\phi(\lam) t } \psi_\lam(x) \psib_\lam(K') , &
K'
		&:=		K\,e^{-\rho t} . \label{eq:dKKC=pphi,2}
\end{align}
Multiplying both sides of \eqref{eq:dKKC=pphi,2} by $\psi_{\lam'}(K')$ and integrating with respect to $K$, we obtain
\begin{align}
\int_E \d_{KK}^2 C(t,x;K) \psi_{\lam'}(K') dK
		&=	e^{-\rho t} \sum_\lam e^{-\phi(\lam) t } \psi_\lam(x) \int_E \psib_\lam(K') \psi_{\lam'}(K') \m(K') dK \\
		&=	\sum_\lam e^{-\phi(\lam) t } \psi_\lam(x) \( \psi_{\lam} , \psi_{\lam'} \) 
		=		e^{-\phi(\lam') t } \psi_{\lam'}(x) . \label{eq:temp}
\end{align}
Note that we have used $dK' = \,e^{-\rho t} dK$ and $\( \psi_{\lam} , \psi_{\lam'} \) = \del_{\lam,\lam'}$.  We solve \eqref{eq:temp} for $\phi(\lam)$
\begin{align}
\phi(\lam)
		&=		\frac{-1}{t} \log \( \frac{\int \d_{KK}^2 C(t,x;K) \psi_\lam(K') dK}{\psi_\lam(x)} \) . \label{eq:ImpLevSub}
\end{align}
If there exists a L\'{e}vy subordinator $T$ independent of $X$, which is capable of generating the prices of call options on the market, then its Laplace exponent, evaluated at $\lam$ is given by \eqref{eq:ImpLevSub}.
\footnote{In a working paper, \cite{carr-lee-1} obtain $\E\[e^{-\lam T_t}\]$ using similar methods.  The authors do not deal with L\'{e}vy subordinators specifically.}
Thus, we refer to $T$ with Laplace exponent \eqref{eq:ImpLevSub} as the \emph{market implied L\'{e}vy subordinator}.
\begin{remark}
It is worth noting that, although one can theoretically compute $\d_{KK}^2 C(t,x;K)$ with call prices available at all strikes $K \in (0,\infty)$, a more convenient expression for the integral in \eqref{eq:ImpLevSub} can be obtained by integrating by parts twice 
\footnote{Our thanks to Marco Avellaneda for pointing this out.}
\begin{align}
\int_0^\infty \d_{KK}^2 C(t,x;K) \psi_\lam(K') dK
		&=		\d_K C(t,x;K) \psi_\lam(K') \Big|_0^\infty - C(t,x;K) \d_K \psi_\lam(K') \Big|_0^\infty \\ & \qquad
					+ \int_0^\infty C(t,x;K) \d_{KK}^2 \psi_\lam(K') dK .
\end{align}
We have the following limits
\begin{align}
\lim_{K \to \infty} \d_K C(t,x;K)
		&=	0 , &
\lim_{K \to \infty} C(t,x;K)
		&=	0 , \\
\lim_{K \to 0} \d_K C(t,x;K) 
		&=	-1 , &
\lim_{K \to 0} C(t,x;K)
		&=	x .
\end{align}
Hence, we find
\begin{align}
\int_0^\infty \d_{KK}^2 C(t,x;K) \psi_\lambda(K') dK
		&=		\psi_\lam(0) + x \, e^{-\rho t} \d_x \psi_\lam(0) 
					+ \int_0^\infty C(t,x;K) \d_{KK}^2 \psi_\lam(K') dK . \label{eq:IntByParts}
\end{align}
Note that $\d_x$ is a derivative with respect to the argument of $\psi_\lam$ whereas $\d_{KK}^2$ is a derivative with respect to $K$.  The advantage of using expression \eqref{eq:IntByParts} rather than the integral in \eqref{eq:ImpLevSub} is that the differential operators in \eqref{eq:IntByParts} act on the eigenfunction $\psi_\lam$ rather than the call price $C(t,x;K)$.  Derivatives of $\psi_\lam$ can be computed analytically, whereas derivatives of call prices $C(t,x;K)$ must be computed numerically from market data.
\end{remark}
Using \eqref{eq:ImpLevSub} one can obtain the value of $\phi(\lam)$ for all $\lam \in \sig(-\A^1)$.  This information is sufficient for constructing the FK transition density $p^{\phi,1}(t,x,y)$ and the transition density $p_S(t,x,y)$ of $S$.  However, to compute the value of $K_{var}$ we need the drift $\gamma$ and the L\'{e}vy measure $\nu$ of the subordinator $T$.  As we show in the next two subsections, $\gamma$ and $\nu(ds)$ can be obtained from limited knowledge of the map $\phi$.
\subsection{Example: the Geometric Brownian motion with default}
\label{s:Phi_GBM}
We now develop in detail how to imply $\phi(\lambda)$ from options data, when the underlying diffusion is given by (\ref{eq:m.GBM}). In order to use (\ref{eq:ImpLevSub}) we need to obtain $\phi(-\mu)$ since it appears on both sides of the equations, and the eigenfunctions of the killed diffusion, which are given by (\ref{eq:eigen.GBM}). The Laplace exponent (\ref{eq:ImpLevSub}) is given in terms of the ratio,
\begin{align}
\frac{\psi_\lambda\left(Ke^{-\phi(-\mu)t}\right)}{\psi_\lambda(x)}=\exp\left(-\frac{(i\lambda-\xi)}{\sigma}\phi(-\mu)t\right)x^{-\frac{(i\lambda-\xi)}{\sigma}}K^{\frac{(i\lambda-\xi)}{\sigma}} , \label{eq:0}
\end{align}
Note that (\ref{eq:ImpLevSub}) simplifies to,
\begin{align}
\phi(\lambda)&=&\left(\frac{i\lambda-\xi}{\sigma}\right)\phi(-\mu)-t^{-1}\left[\log\left(\int_E\partial^2_{KK}C(t,x;K)K^{\frac{i\lambda-\xi}{\sigma}}\mathrm{d}K\right)-\left(\frac{i\lambda-\xi}{\sigma}\right)\log(x)\right] , \label{eq:1}
\end{align}
setting $\lambda=-\mu$ we find,
\begin{align}
\phi(-\mu)=-\left[t\left(1+\frac{\xi+i\mu}{\sigma}\right)\right]^{-1}\left[\log\left(\int_E\partial^2_{KK}C(t,x;K)K^{-\frac{(i\mu+\xi)}{\sigma}}\mathrm{d}K\right)+\frac{i\mu+\xi}{\sigma}\log(x)\right] . \label{eq:2}
\end{align}
Now that we have obtained $\phi(-\mu)$, we can compute the whole function $\lambda\mapsto\phi(\lambda)$ with numerical integration. Moreover, we can use the same approach to derive the Laplace exponent of the L\'evy subordinator for the JDCEV dynamics. However, in this case we need to solve numerically for $\phi(-\mu)$.


\subsection{Obtaining $\gamma$ and $\nu$ from $\phi(\lam)$: the compound Poisson case}
\label{sec:poisson}
As noted in Section \ref{sec:model}, when the subordinator $T$ is of the compound Poisson type, its L\'{e}vy measure $\nu$ can be written as the product of the net jump intensity $\alpha:=\nu((0,\infty))$ times the jump distribution $F$
\begin{align}
\nu(ds)
		&=		\alpha \, F(ds) . \label{eq:nu=aF}
\end{align}
In this scenario, the L\'{e}vy-Kintchine formula \eqref{eq:LevyKintchine} can be written
\begin{align}
\phi(\lam)
		&=		\gamma \, \lam + \alpha \int_0^\infty (1 - e^{-\lam s}) F(ds) 
		=			\gamma \, \lam + \alpha \( 1 - \Fh(\lam) \) , \label{eq:phi,2}
\end{align}
where we have defined $\Fh(\lam)$, the Laplace transform of the measure $F(ds)$ 
\begin{align}
\Fh(\lam)
		&:=		\int_0^\infty e^{-\lam s} F(ds) . 
\end{align}
The drift of the subordinator $\gamma$ and the net jump intensity $\alpha$ can now be obtained from $\phi(\lam)$ by taking the following limits
\begin{align}
\lim_{\lam \to \infty} \frac{\phi(\lam)}{\lam} 
		&= 	\lim_{\lam \to \infty} \( \gamma + \alpha \int_0^\infty \frac{1 - e^{-\lam s}}{\lam} F(ds) \)
		=		\gamma , \label{eq:gamma} \\
\lim_{\lam \to \infty} \( \phi(\lam) - \gamma \, \lam \)
		&=		\lim_{\lam \to \infty} \alpha \int_0^\infty \(1 - e^{-\lam s}\) F(ds)
		=			\alpha , \label{eq:alpha}
\end{align}
where we have used $\int_0^\infty F(ds) = 1$.  After obtaining $\gamma$ and $\alpha$, we use \eqref{eq:phi,2} to solve for $\Fh$
\begin{align}
\Fh(\lam)
		&=		1 + \frac{\gamma \lam - \phi(\lam)}{\alpha} . \label{eq:Fhat}
\end{align}
Note, because $\sig(-\A^1) \subseteq \R^+$, equation \eqref{eq:Fhat} will not give us the value of $\Fh(\lam)$ for any $\lam < 0$ (we know $\Fh(0)=1$).  However, knowledge of $\Fh(\lam)$ for $\lam<0$ is not needed in order to uniquely determine $\Fh$.  To see this, we need the following theorem.
\begin{theorem}[Analyticity of the Laplace Transform]
\label{thm:laplace}
Let $\lam \in \R$.  Let $\alpha(s)$ be a real, non-negative, non-decreasing function which satisfies $\alpha(0)=0$ and is of bounded variation on $[0,R]$ for every $R>0$.  If the integral
\begin{align}
\widehat{\alpha}(\lam)
		&:=		\int_0^\infty e^{-\lam s} \alpha(ds) ,
\end{align}
converges for all $\lam > \lam_c$, $\lam_c<\infty$ then $\widehat{\alpha}(\lam)$ is analytic for all $\lam>\lam_c$, and
\begin{align}
\frac{d^n}{d\lam^n}\widehat{\alpha}(\lam)
		&:=		\int_0^\infty (-s )^n e^{-\lam s} \alpha(ds) . \label{eq:laplace.derivative}
\end{align}
\end{theorem}
\begin{proof}
See \cite{widder-1}, page 57, Theorem 5a.
\end{proof}
\begin{remark}
A function that is analytic in a domain $D$ is uniquely determined over $D$ by its values along a line segment in $D$.
\end{remark}
\begin{remark}
From \eqref{eq:laplace.derivative}, it is clear that $\widehat{\alpha}(\lam)$ is decreasing and convex.
\end{remark}
\noindent
Let ${\rm Dom}(\Fh)=(\lam_c,\infty)$ where $\lam_c \in \R$.  If the continuous spectrum of $-\A^1$ is non-empty $\sig_c(-\A^1) \neq \{ \emptyset \}$, equation \eqref{eq:Fhat} gives us a map of $\Fh(\lam)$ on some interval $I \subset {\rm Dom}(\Fh)$.  The analytic extension of that map is unique and well-defined in throughout ${\rm Dom}(\Fh)$.  If the continuous spectrum of $-\A^1$ is empty $\sig_c(-\A^1) = \{ \emptyset \}$, then \eqref{eq:Fhat} gives us a map of $\Fh(\lam)$ at a countably infinite number of points in ${\rm Dom}(\Fh)$ (i.e., the proper eigenvalues $\lam_n$ of $-\A^1$).  In this case, the analytic extension of $\Fh(\lam)$ is still uniquely determined if $F$ is Lipschitz (see \cite{baumer-neubrander-1}, Corollary 1.3)
\begin{align}
F(0)
		&=		0 &
		&\text{and}  &
\sup_{t,s\geq0} \frac{|F(t)-F(s)|}{|t-s|}
		&<		\infty .
\end{align}
If $F$ is not Lipschitz, the analytic continuation of $\Fh(\lam)$ is uniquely determined if we know the value of $\Fh(\lam)$ at equally spaced intervals, i.e., if, for $n=1,2,3,\cdots$, we know $\Fh(\lam_n)$ where $\lam_n = a + b \, n$ for some $a>\lam_c$ and $b > 0$ (see \citet{widder-1}, Theorem 6.2).  Note that the eigenvalues \eqref{eq:eigenval.JDCEV} of the JDCEV process are equally spaced.
\par
From a practical standpoint, one cannot evaluate \eqref{eq:Fhat} at an infinite number of $\lam$.  Thus, in order to obtain $F(ds)$ from \eqref{eq:Fhat}, one should seek to fit a positive, analytic, decreasing, convex function to a finite number of points of $\Fh(\lam)$.  Upon doing this, one can use the inverse Laplace transform (Bromwich integral) to obtain $F((0,s))$
\begin{align}
F((0,s))
		&=		\frac{1}{2 \pi i} \int_0^s \( \int_{C-i\infty}^{C+i\infty} e^{\lam u} \Fh(\lam) d\lam \) du  .
\end{align}
Here the constant $C\in\R$ is chosen so that the contour of integration lies to the right of all singularities of $\Fh(\lam)$.
Another option for obtaining $F(ds)$ is to numerically invert the Laplace transform $\Fh(\lam)$.  For a survey of numerical techniques for Laplace inversion we refer the reader to \cite{davies-martin-1} and the references therein.  Once $F(ds)$ is obtained, the L\'{e}vy measure $\nu$ is given by \eqref{eq:nu=aF}.  In figure \ref{fig:phi}, we graphically illustrate how to obtain $\gamma$, $\alpha$ and $\Fh(\lam)$ from knowledge of $\phi(\lam)$ at a discrete set of points.


\subsection{Obtaining $\gamma$ and $\nu$ from $\phi(\lam)$: the general case}
\label{sec:general}
When the subordinator $T$ is \emph{not} of the compound Poisson type, its drift $\gamma$ can still be found using
\begin{align}
\lim_{\lam \to \infty} \frac{\phi(\lam)}{\lam} 
		&= 	\lim_{\lam \to \infty} \( \gamma + \int_0^\infty \frac{1 - e^{-\lam s}}{\lam} \nu(ds) \)
		=		\gamma .
\end{align}
To obtain $\nu$ we must introduce $\om(s)$ the \emph{tail of the L\'{e}vy measure}  
\begin{align}
\om(s)
		&:= 	\nu((s,\infty))		
		=			\int_{(s,\infty)} \nu(dz) .
\end{align}
Following \cite{bertoin-4}, pp. 7, we note that
\begin{align}
\frac{\phi(\lam)}{\lam}-\gamma
		&=		\int_{(0,\infty)} \frac{1-e^{-\lam s}}{\lam}\nu(ds)
		=			\int_{(0,\infty)} e^{-\lam s} \om(s) ds 
		=:			\omh(\lam) , \label{eq:omegahat}
\end{align}
where $\omh(\lam)$ is the Laplace transform of the function $\om(s)$.  Depending on the nature of $\sig(-\A^1)$, equation \eqref{eq:omegahat} either gives us a map of $\omh(\lam)$ along a line segment $I \subset {\rm Dom}(\omh)$, in which case the analytic extension of $\omh(\lam)$ is unique, or \eqref{eq:omegahat} gives us the value of $\omh(\lam)$ at a countably infinite number of points.  In this case, the analytic continuation of $\omh(\lam)$ is uniquely determined if we know the value of $\omh(\lam)$ at equally spaced intervals (the Lipschitz condition would not be satisfied for an infinite activity L\'{e}vy process).  From a practical standpoint, one may seek to fit an analytic, decreasing, convex function to a finite number of points of $\omh(\lam)$.  Upon doing this, one can obtain $\om(s)$ from the Bromwich integral
\begin{align}
\om(s)
		&=		\frac{1}{2 \pi i} \int_{C-i\infty}^{C+i\infty} e^{\lam s} \omh(\lam) d\lam .
\end{align}
Finally, one obtains $\nu(ds)$ from
\begin{align}
\nu(ds) 
		&=		-d \om(s).
\end{align}

\section{Empirical results}
\label{sec:calibration}
In this Section, we use the framework developed in Section \ref{sec:impliedtime} to imply $\phi$ from the market using call options on Apple, Google, Microsoft, and Facebook.  We then use this information to compute Variance swap rates.  
Call option quotes were obtained on April 12, 2013 from Google Finance and have a maturity of 19 days. None of the firms payed dividends over the tenor of the option. For each individual symbol we implement the procedure described below.
\par
In order to imply $\phi$ from the market, we must first fix a background diffusion $X$.  For simplicity, as in Section \ref{sec:GBM}, we assume killed geometric Brownian motion dynamics for the background diffusion $X$.  The generator $\A_1$ of $X$ is given in \eqref{eq:A.GBM}.  In order to find the parameters ($\sig$, $\mu$, and $k$) of $\A_1$, we choose a range of parametric subordinators, and we fit the subordinated diffusion model to observed implied volatilities. Our calibration results reveal that $\mu$ is an order of magnitude smaller than $k$.  Therefore we assume $\mu=0$.   Next, to obtain the remaining parameters $(k,\sigma)$, we fit the killed GBM model (with no subordination) to observed implied volatilities.  The results of this calibration fix $k$ and $\sigma$.
\par
Having fixed a background process $X$, we use equation \eqref{eq:1} from Section \ref{s:Phi_GBM} to derive the Laplace characteristic exponent $\phi$ of the unknown market subordinator.  In order to obtain a continuum $K \mapsto C(t,K)$ of call option prices (which are needed to compute the integral in \eqref{eq:1}), we use the \emph{Stochastic Volatility Inspired} (SVI) parametric family of \cite{Gatheral-2}  and \cite{Gatheral-Jacquier-1} to interpolate the volatility smile.  Note that the Laplace exponent $\phi$ is mathematically expected to be purely real.  However, because we are projecting the density $p_S(t,x,y)$ onto an imposed basis of eigenfunctions corresponding to the infinitesimal generator $\A_1$, we find that the market implied $\phi$  has an imaginary component.   However, the imaginary component is significantly smaller than the real part of $\phi$.  Furthermore, we find that the real part of $\phi$ verifies the properties of characteristic functions.
\par
In Figures \ref{fig:MarketImpliedPhiLambda} and \ref{fig:MarketImpliedPhiLambda2}, we plot the obtained (real part) of $\phi$ for the four stocks in our data set (Apple, Google, Microsoft, and Facebook).  In the same Figures, we also plot $\phi(\lambda)/\lambda$ and $\phi(\lambda)-\gamma\lambda$, which, using equations \eqref{eq:gamma} and \eqref{eq:alpha}, enable us to identify the drift $\gamma$ and net jump-intensity $\alpha$ of the subordinator, and $\widehat{F}(\lam)$, which we obtain using \eqref{eq:Fhat}.  Note that the jump intensity appears to be finite for all four stocks (though, we emphasize that we did not assume this in our analysis).  Next, we use \eqref{eq:Fhat} to derive $\widehat{F}$.  Finally, we obtain the L\'evy measure of the subordinator $\nu(ds)=\alpha \cdot F(ds)$ using the Laplace transform inversion techniques of \cite{davies-martin-1}.  Having fully characterized the subordinator for each stock, we use Proposition \ref{GBM.QV} to compute variance swap rates.  The results are given below.
\newline
\begin{center}
\begin{tabular}{|c|| c | c | c | c |}
\hline
& AAPL & FB & GOOG & MSFT \\
\hline \hline
1M & 0.007821752 & 0.009443126 & 0.00649349 & 0.002771267\\ \hline
3M & 0.023352561 & 0.028183341 & 0.019121125 & 0.008281498\\ \hline
6M & 0.046369764 & 0.055932341 & 0.037197644 & 0.016466711\\ \hline
1Y & 0.091417487 & 0.11015411 & 0.070445294 & 0.032552754\\ \hline
\end{tabular}
\end{center}

%
%

\section{Conclusion}
\label{sec:conclusion}
In this paper we model the price process of an underlying $S$ as a Feller diffusion time changed by a  L\'{e}vy subordinator.  This class of models, first developed in \cite{mendoza-carr-linetsky-1}, allows for the underlying to experience jumps with a state-dependent L\'{e}vy measure, local stochastic volatility and a local stochastic default intensity.
\par
The contribution of this paper is two-fold.  First, we show how to compute the price of a VS contract in the general L\'{e}vy subordinated diffusion setting.  Using our general formula, we perform specific VS computations when the background diffusion is modeled as (i) a GBM with default and $(ii)$ a JDCEV process.  Second, we show that the drift and L\'{e}vy measure of the L\'{e}vy subordinator that drives the price process can be obtained directly by observation of the $t$-expiry volatility smile.  By using call and put prices to uniquely determine the L\'{e}vy subordinator that drives the underlying price process, we reduce the risk of model misspecification.

\subsection*{Thanks}
The authors would like to thank Stephan Sturm and Vadim Linetsky for their helpful comments on this work.

%
%

\appendix

%
%

\section{Proofs}



\subsection{Proof of Proposition \ref{GBM.QV}}
\label{GBM.QV.proof}
Part $(i)$ is a straightforward computation
\begin{align}
\int_0^t \E_x \[ \I_{\{\zeta^\phi>u\}} \gamma \sig^2(X_u^\phi) \] du 
		&=		\gamma \sig^2 \int_0^t \E_x \[ \I_{\{\zeta^\phi>u\}} \] du 
		=			\gamma \sig^2 \int_0^t e^{-ku} du
		=			\gamma \sig^2 \( \frac{1-e^{-k t}}{k} \) .
\end{align}
Part $(ii)$. Using Proposition~\ref{prop.semi.g.1} we have
\begin{align}
&\int_0^t  \E_x \[ \I_{\{\zeta^\phi>u\}}\int_{\mathbb R}\log^2\Big(1+\frac{y}{X^\phi_{u-}}\Big)\pi^{\phi,1}(X_{u-}^\phi,y)dy \] du \\
&\qquad=  \int_{(0,\infty)} \int_0^t \int_\R e^{-\lam_\om s} \bigg\{
		\Pc_u^{\phi,1} \( \psi_\om(x) \) \cdot \( \int_{E\setminus\{x\}} \log^2(y) \psib_\om(y) \m(y)  dy \) \\ & \qquad \qquad
		- 2 \, \Pc_u^{\phi,1} \( \psi_\om(x) \log x \) \cdot\( \int_{E\setminus\{x\}} \log (y) \psib_\om(y) \m(y)  dy \) \\ & \qquad \qquad
		+ \Pc_u^{\phi,1} \( \psi_\om(x) \log^2 x \) \cdot\( \int_{E\setminus\{x\}} \psib_\om(y) \m(y)  dy \)
		\bigg\} \, d\om \, du \, \nu(ds)	,	\label{eq:GBM.temp}
\end{align}
where
\begin{align}
\Pc_u^{\phi,1} \( \psi_\om(x) \log^n x \)
		&=		(- i \sig )^n	\d_{\om}^n \( e^{-\phi(\lam_{\om}) u } \psi_{\om}(x) \) , \label{eq:GBM.semigroup} \\
 \int_E \log^n(y) \psib_\om(y) \m(y)  dy
		&=		2 \, \sqrt{\pi} \, \sig^{n-1/2} \, i^n \, \d_\om^n \del(\om + i \xi) . \label{eq:GBM.integral}
\end{align}
Inserting \eqref{eq:GBM.semigroup} and \eqref{eq:GBM.integral} into \eqref{eq:GBM.temp} yields
\begin{align}
\eqref{eq:GBM.temp}
&= \int_{(0,\infty)} \int_0^t \int_\R e^{-\lam_\om s} \bigg\{
		\( e^{-\phi(\lam_{\om}) u } \psi_{\om}(x) \)  
			\cdot 2 \, \sqrt{\pi} \, \sig^{2-1/2} \, i^2 \, \d_\om^2 \del(\om + i \xi) \\ & \qquad \qquad
		- 2 \, (- i \sig )	\d_{\om} \( e^{-\phi(\lam_{\om}) u } \psi_{\om}(x) \) 
			\cdot 2 \, \sqrt{\pi} \, \sig^{1-1/2} \, i \, \d_\om \del(\om + i \xi) \\ & \qquad \qquad
		+ (- i \sig )^2	\d_{\om}^2 \( e^{-\phi(\lam_{\om}) u } \psi_{\om}(x) \)  
			\cdot 2 \, \sqrt{\pi} \, \sig^{-1/2} \, \del(\om + i \xi)
		\bigg\} \, d\om \, du \, \nu(ds)	\\
&= 2 \, \sqrt{\pi} \sig^{2-1/2} i^2 \int_{(0,\infty)} \int_\R \bigg\{ 
		\d_\om^2 \[ e^{-\lam_\om s} \( \frac{1 - e^{-\phi(\lam_{\om}) t }}{\phi(\lam_\om)} \psi_{\om}(x) \) \] \\&\qquad\qquad
		- 2 \,\d_\om \[ e^{-\lam_\om s} \d_\om \( \frac{1 - e^{-\phi(\lam_{\om}) t }}{\phi(\lam_\om)} \psi_{\om}(x) \) \] \\&\qquad\qquad
		+ \[ e^{-\lam_\om s} \d_\om^2 \( \frac{1 - e^{-\phi(\lam_{\om}) t }}{\phi(\lam_\om)} \psi_{\om}(x) \) \]
		\bigg\} \del(\om + i \xi) d\om \, \nu(ds) . \label{eq:GBM.temp2}
\end{align}
In the second equality we have integrated with respect to $u$ and used integration by parts to move the $\d_\om^n$ off of the Dirac delta functions.  Noting that
\begin{align}
(f g)'' - 2 (f g')' + f g''
		&=		f'' g , &
\psi_{-i\xi}(x)
		&= \sqrt{\sig/4\pi}, &
\lam_{-i\xi}
		&=	k ,
\end{align}
we find
\begin{align}
\eqref{eq:GBM.temp2}
&=  2 \, \sqrt{\pi} \sig^{2-1/2} i^2 \int_{(0,\infty)} \bigg\{ 
		\d_\om^2 \( e^{-\lam_\om s} \) \cdot \( \frac{1 - e^{-\phi(\lam_{\om}) t }}{\phi(\lam_\om)} \psi_{\om}(x) \) 
		\bigg\}_{\om = - i \nu} \nu(ds) \\
&=	\sig^2 \( \frac{1 - e^{-\phi(k) t }}{\phi(k)} \) \int_{(0,\infty)}e^{-k s} \(  s + s^2 \nu^2 \) \nu(ds) . \label{eq:nu.eps}
\end{align}


\subsection{Computation of ${\tilde c}_n =(x ^p,\psi_n )$ from Proposition \ref{mom.p}}
\label{sec:coeffs}
Part $(ii)$. Consider the series $U(x)=\sum_{k=1}^\infty u_k(x)$ for $x\in D\subset (0,\infty)$. If for all $x\in D$ the function $u_k(x)$ satisfies the inequality $|u_k(x)|\leq d_k$ for $k=1,2,\cdots,$ where the series $\sum_{k=1}^\infty d_k <\infty$, then the series $U(x)$ converges uniformly (see \cite{prudnikov-3}, Section I.3.4.3, p.751). From the inequality (27a) on p.54 of \cite{nikiforov-uvarov-1}, we find that $|\psi^1_n(x)|< C/(n-1)^{1/4}$ for some $C<\infty$ independent of $n$. Therefore, we have
\begin{align}
|c_n\psi_n(x)|\leq  \frac{C\Big|\left(\frac{1-p}{2|\beta |}\right)_{n-1}\Big| }{\sqrt{(n-1)! \Gamma(\nu +n)}(n-1)^{1/4}}=d_n,\end{align}
To show that the series $\sum_{n}^\infty d_n $ converges, it is enough to show that for a large $n$ we have $\log(d_n)/\log(n)<-1$ 
(see \cite{prudnikov-3} Section I.3.2.19, p.751). 
Therefore, observe that
\begin{align}
d_n&\approx\frac{(n-1)^{\frac{1-p}{2|\beta |}+n/2-3/2}}{ (\nu+n)^{(\nu+n)/2-1/4} } ,& n &\gg 1 .
\end{align}
Thus
\begin{align}
\frac{\log(d_n)}{\log(n)}
&\approx \frac{\left(\left(\frac{1-p}{2|\beta |}\right)+n/2-3/2\right)\log(n-1) - \left((\nu +n)/2 -1/4\right)\log(\nu+n) }{\log(n)}\\
&<\Big[\left(\left(\frac{1-p}{2|\beta |}\right)+\frac{n}{2}-\frac{3}{2}\right) - \left(\frac{\nu +n}{2} -\frac{1}{4}\right)\Big]\frac{\log(n-1) }{\log(n)} .
\end{align}
Note that $\log(n-1)/\log(n)\uparrow1$. Moreover, it can be verified that the term inside the bracket is less than $-1$ for all $ p>(\beta+1)/2 -c$ (recall $\nu=(2c+1)/(2|\beta|))$. This shows uniform convergence at $t=0$ for $ p>(\beta+1)/2 -c$. To show that the series is absolutely convergent at $t=0$ it suffices to show that  $\lim_{n\rightarrow\infty} d_{n+1}/d_{n}<1$ with $d_n=|c_n\psi_n(x)|$ (i.e., d'Alambert's test for convergence). Equivalently, $\sum_{n}^\infty d_n $ converges if $\log(d_{n+1}/d_{n})<0$ as $n\rightarrow\infty$.
Hence, analyzing the asymptotic behavior and noticing that $L_{n-1}^{\nu}(z)\approx  e^{\frac{z}{2}}z^{-(2\nu+1)/4} (n-1)^{\nu/2-1/4} \cos\{2\sqrt{(n-1)z}-\pi(2\nu+1)/4\}/\sqrt{\pi}$ for $n\gg1$, we find that
\begin{align}
d_n&=\Big|\frac{A ^{\frac{\nu }{2}-\frac{p+2c }{2|\beta |}}\left(\frac{1-p}{2|\beta |}\right)_{n-1} \Gamma\left(\frac{p+2c }{2|\beta |}+1\right) }{\sqrt{(n-1)! |\mu +b |\Gamma(\nu +n)}} \psi_n (x )\Big|<C\frac{(n-1)^{\frac{1-p}{2|\beta |}+\frac{\nu}{2}+n-\frac{7}{4}}}{(\nu+n)^{(\nu+n)-1/2}},&n&\gg1.
\end{align}
Thus, we would like to test $\lim_{n\rightarrow\infty}\log(d_{n+1}/d_n)<0$, where $n\gg1$.  First observe that
\begin{align}
\log\Big(\frac{d_{n+1}}{d_n}\Big)&=\Big(\frac{1-p}{2|\beta |}+\frac{\nu}{2}+n-\frac{3}{4}\Big)\log(n)-\Big(\frac{1-p}{2|\beta |}
+\frac{\nu}{2}+n-\frac{7}{4}\Big)\log(n-1)\\
&\qquad+\Big(\nu+n-\frac{1}{2}\Big)\log(\nu+n)-\Big(\nu+n+\frac{1}{2}\Big)\log(\nu+n+1).
\end{align}
Then, making use of the approximation $\log(n+a)\approx\log(n)+a/n-a^2/(2n^2)+\cdots$, we find
\begin{align}
\lim_{n\rightarrow\infty}\log\Big(\frac{d_{n+1}}{d_n}\Big)=\lim_{n\rightarrow\infty}-\frac{1}{n}\Big(\frac{1-p}{2|\beta |}+\frac{\nu}{2}+n-\frac{7}{4}\Big)+\Big(\nu+n-\frac{1}{2}\Big)\frac{\nu}{n}-\Big(\nu+n+\frac{1}{2}\Big)\frac{\nu+1}{n}=-2,
\end{align}
which concludes the proof.


\subsection{Proof of Proposition~\ref{longInt}}
\label{sec:long.proof}
First observe that from Proposition~\ref{prop.semi.g.1} we obtain,
\begin{align}
&\int_0^t  \E_x \[ \I_{\{\zeta^\phi>u\}}\int_{\mathbb R}\log^2\Big(1+\frac{y}{X^\phi_{u-}}\Big)\pi^{\phi,1}(X_{u-}^\phi,y)dy \] du\\
&=\int_{(0,\infty)}\int_0^t\sum_{n=1}^\infty e^{-\lambda_n s}\Big\{\Big(\int_{E\setminus\{x\}} \log^2(y)\psi_n(y){\mathfrak m}(y)dy\Big)\big[\Pc_u^{\phi,1} \psi_n(x)\big]\\
&\qquad-2\Big(\int_{E\setminus\{x\}} \log(y)\psi_n(y){\mathfrak m}(y)dy\Big)\big[\Pc_u^{\phi,1} \big(\log(x)\psi_n(x)\big)\big]\\
&\qquad+\Big(\int_{E\setminus\{x\}} \psi_n(y){\mathfrak m}(y)dy\Big)\big[\Pc_u^{\phi,1} \big(\log^2(x)\psi_n(x)\big)\big]\Big\}du\,\nu(ds),\label{longint.2}
\end{align} 
where
\begin{align}
\Pc_u^{\phi,1} \psi_n(x)
		&=\sum_{m=1}^\infty e^{-\phi(\lambda_m)u}(\psi_n,\psi_m)\psi_m(x)=e^{-\phi(\lambda_n)u}\psi_n(x),\\
\Pc_u^{\phi,1} \big(\log(x)\psi_n(x)\big)
		&=\sum_{m=1}^\infty e^{-\phi(\lambda_m)u} d^1_{n,m}\psi_m(x),\quad d^1_{n,m}=\int_{E}\log(y)\psi_n(y)\psi_m(y){\mathfrak m}(y)dy,\\
\Pc_u^{\phi,1} \big(\log^2(x)\psi_n(x)\big)
		&=\sum_{m=1}^\infty e^{-\phi(\lambda_m)u} d^2_{n,m}\psi_m(x),\quad d^2_{n,m}=\int_{E}\log^2(y)\psi_n(y)\psi_m(y){\mathfrak m}(y)dy .
\end{align}
Explicit representation of $d^1_{n,m}$ and $ d^2_{n,m}$ are found using
\begin{align}
d^1_{n,m}=2|\beta|A^{\nu+1}\sqrt{\frac{(n-1)!(m-1)!  }{\Gamma(\nu+n)\Gamma(\nu+m)}}\int_{E}y^{2c-2\beta}\log(y)e^{- Ay^{-2\beta}}L_{n-1}^{\nu}(Ay^{-2\beta})L_{m-1}^{\nu}(Ay^{-2\beta})dy,
\end{align}
and
\begin{align}
d^2_{n,m}=2|\beta|A^{\nu+1}\sqrt{\frac{(n-1)!(m-1)!  }{\Gamma(\nu+n)\Gamma(\nu+m)}}\int_{E}y^{2c-2\beta}\log^2(y)e^{- Ay^{-2\beta}}L_{n-1}^{\nu}(Ay^{-2\beta})L_{m-1}^{\nu}(Ay^{-2\beta})dy.
\end{align}
Making the change of variable $z=Ay^{-2\beta}$ we obtain
\begin{align}
d^1_{n,m}
		&=\frac{1}{2|\beta|}\sqrt{\frac{(n-1)!(m-1)!  }{\Gamma(\nu+n)\Gamma(\nu+m)}}\int_{E}z^{\nu}(\log(z)-\log(A))e^{- z}L_{n-1}^{\nu}(z)L_{m-1}^{\nu}(z)dz\\
		&=\frac{1}{2|\beta|}\sqrt{\frac{(n-1)!(m-1)!  }{\Gamma(\nu+n)\Gamma(\nu+m)}}\int_{E}z^{\nu}\log(z)e^{- z}L_{n-1}^{\nu}(z)L_{m-1}^{\nu}(z)dz\\ &\qquad
-\frac{\log(A)}{2|\beta|}\sqrt{\frac{(n-1)!(m-1)!  }{\Gamma(\nu+n)\Gamma(\nu+m)}}\int_{E}z^{\nu}e^{- z}L_{n-1}^{\nu}(z)L_{m-1}^{\nu}(z)dz,
\end{align}
and
\begin{align}
d^2_{n,m}
		&=\frac{1}{(2|\beta|)^2}\sqrt{\frac{(n-1)!(m-1)!  }{\Gamma(\nu+n)\Gamma(\nu+m)}}\int_{E}z^{\nu}(\log(z)-\log(A))^2e^{- z}L_{n-1}^{\nu}(z)L_{m-1}^{\nu}(z)dz\\
		&=\frac{1}{4|\beta|^2}\sqrt{\frac{(n-1)!(m-1)!  }{\Gamma(\nu+n)\Gamma(\nu+m)}}\int_{E}z^{\nu}\log^2(z)e^{- z}L_{n-1}^{\nu}(z)L_{m-1}^{\nu}(z)dz\\ &\qquad
-\frac{\log(A)}{2|\beta|^2}\sqrt{\frac{(n-1)!(m-1)!  }{\Gamma(\nu+n)\Gamma(\nu+m)}}\int_{E}z^{\nu}\log(z)e^{- z}L_{n-1}^{\nu}(z)L_{m-1}^{\nu}(z)dz\\ &\qquad
+\frac{\log^2(A)}{4|\beta|^2}\sqrt{\frac{(n-1)!(m-1)!  }{\Gamma(\nu+n)\Gamma(\nu+m)}}\int_{E}z^{\nu}e^{- z}L_{n-1}^{\nu}(z)L_{m-1}^{\nu}(z)dz.
\end{align}
Now, we use the identity $\int_{E}z^{\nu}e^{- z}L_{n-1}^{\nu}(z)L_{m-1}^{\nu}(z)dz=\delta_{n,m}\Gamma(\nu+n)/(n-1)!$ (where $\delta_{n,m}$ is the Kronecker delta) and the series expansion for the Generalized Laguerre polynomials ($L_{n}^\nu(z)=\Gamma(\nu+n+1)/n! \sum_{k=0}^n(-n)_kz^k /(\Gamma(\nu+k+1)k!)$) to obtain
\begin{align}
d^1_{n,m}=\frac{1}{2|\beta|}\left\{\sqrt{\frac{(m-1)!  \Gamma(\nu+n)}{(n-1)!\Gamma(\nu+m)}}\sum_{k=0}^{n-1}\frac{(1-n)_{k}}{\Gamma(\nu+k+1)k!}\int_{E}z^{(\nu+k+1)-1}\log(z)e^{- z}L_{m-1}^{\nu}(z)dz-\log(A)\right\},
\end{align}
and
\begin{align}
d^2_{n,m}
		&=\frac{1}{4|\beta|^2}\Bigg\{\sqrt{\frac{(m-1)!  \Gamma(\nu+n)}{(n-1)!\Gamma(\nu+m)}}\sum_{k=0}^{n-1}\frac{(1-n)_{k}}{\Gamma(\nu+k+1)k!}\int_{E}z^{(\nu+k+1)-1}\log^2(z)e^{- z}L_{m-1}^{\nu}(z)dz\\ &\qquad
-2\log(A)\sqrt{\frac{(m-1)!  \Gamma(\nu+n)}{(n-1)!\Gamma(\nu+m)}}\sum_{k=0}^{n-1}\frac{(1-n)_{k}}{\Gamma(\nu+k+1)k!}\int_{E}z^{(\nu+k+1)-1}\log(z)e^{- z}L_{m-1}^{\nu}(z)dz+\log^2(A)\Bigg\}.
\end{align}
Similarly, we obtain
\begin{align}
c_n=\int_E\psi_n(y){\mathfrak m}(y)dy=
\frac{A ^{\frac{1-2c }{4|\beta |}}(1/(2|\beta |))_{n-1}\,\Gamma(c /|\beta |+1) }{\sqrt{(n-1)! |\mu +b |\Gamma(\nu +n)}},
\label{mom.0.1}
\end{align}
where $c_n$ is found by setting $p=0$ in \eqref{mom.2}.  Also,
\begin{align}
\int_E \log(y)\psi_n(y){\mathfrak m}(y)dy
		&=\frac{A^{\frac{1-2c}{4|\beta|}}}{2|\beta|}\sqrt{\frac{(n-1)! }{ |\mu+b|\Gamma(\nu+n)}}\int_E z^{\frac{c+|\beta|}{|\beta|}-1}(\log(z)-\log(A))e^{- z} L_{n-1}^{\nu}(z)dz\\
		&=\frac{A^{\frac{1-2c}{4|\beta|}}}{2|\beta|}\sqrt{\frac{(n-1)! }{ |\mu+b|\Gamma(\nu+n)}}\int_E z^{\frac{c+|\beta|}{|\beta|}-1}\log(z)e^{- z} L_{n-1}^{\nu}(z)dz-\frac{\log(A)}{2|\beta|}c_n ,
\end{align}
and
\begin{align}
\int_E \log^2(y)\psi_n(y){\mathfrak m}(y)dy
		&=\frac{A^{\frac{1-2c}{4|\beta|}}}{4|\beta|^2}\sqrt{\frac{(n-1)! }{ |\mu+b|\Gamma(\nu+n)}}\int_E z^{\frac{c+|\beta|}{|\beta|}-1}(\log(z)-\log(A))^2e^{-z} L_{n-1}^{\nu}(z)dz\\
		&=\frac{1}{4|\beta|^2}\Bigg\{A^{\frac{1-2c}{4|\beta|}}\sqrt{\frac{(n-1)! }{ |\mu+b|\Gamma(\nu+n)}}\int_E z^{\frac{c+|\beta|}{|\beta|}-1}\log^2(z)e^{-z} L_{n-1}^{\nu}(z)dz\\ &\qquad
-2\log(A)A^{\frac{1-2c}{4|\beta|}}\sqrt{\frac{(n-1)! }{ |\mu+b|\Gamma(\nu+n)}}\int_E z^{\frac{c+|\beta|}{|\beta|}-1}\log(z)e^{-z} L_{n-1}^{\nu}(z)dz+\log^2(A)c_n\Bigg\} .
\end{align}
Substituting the above expressions into Eq. \eqref{longint.2} and using the relation \eqref{eq:Theta} for $\Theta^1_{n,m}$, and $\Theta^2_{n,m}$ as well as equations \eqref{prud.ln.1a} and \eqref{prud.ln.2a}, we arrive to the final expression \eqref{longint.3}.

%
%

\section{Spectral Theorem}
\label{sec:spectral}
In this Appendix we summarize the theory of self-adjoint operators acting on a Hilbert space.  A detailed exposition on this topic (including proofs) can be found in \cite{reed-simon-1} and \cite{rudin-3}.
\par
Let $\H$ be a Hilbert space with inner product $(\cdot,\cdot)$.  A \emph{linear operator}  is a pair $({\rm Dom}(\A),\A)$ where ${\rm Dom}(\A) $ is a linear subset of $\H$ and $\A$ is a linear map $\A:{\rm Dom}(\A) \to \H$.  
The \emph{adjoint} of an operator $\A$ is an operator $\A^{*}$ such that $(\A f,g)	=	(f, \A^{*} g), \forall \, f \in {\rm Dom}(\A), g \in {\rm Dom}(\A^{*})$, where
\begin{align}
{\rm Dom}(\A^{*}):=\{ g \in \H : \exists \, h \in \H \text{ such that } (\A f, g) = (f,h) \,\, \forall \, f \in {\rm Dom}(\A) \} .
\end{align}
An operator $({\rm Dom}(\A),\A)$ is said to be \emph{self-adjoint} in $\H$ if
\begin{align}
{\rm Dom}{(\A)}	&=	{\rm Dom}(\A^*) , &
(\A f,g)					&=	(f,\A g) & \forall \, f,g \in {\rm Dom}(\A).
\end{align}
Throughout this Appendix, for any self-adjoint operator $\A$, we will assume that ${\rm Dom}(\A)$ is a dense subset of $\H$.
A densely defined self-adjoint operator is closed (see \cite{rudin-3}, Theorem 13.9).
\par
Given a linear operator $\A$, the \emph{resolvent set} $\rho(\A)$ is defined as the set of $\lam \in \mathbb{C}$ such that the mapping $(\A - \text{Id} \, \lam)$ is one-to-one and $R_\lam:=(\A - \text{Id} \, \lam)^{-1}$ is continuous with ${\rm Dom}( R_\lam ) = \H$.  The operator $R_\lam:\H \to \H$ is called the \emph{resolvent}.  The \emph{spectrum} $\sig(\A)$ of an operator $\A$ is defined as $\sig(\A):= \mathbb{C}\setminus\rho(\A)$. We say that $\lam \in \sig(\A)$ is an \emph{eigenvalue} of $\A$ if there exists $\psi \in {\rm Dom}(\A)$ such that the \emph{eigenvalue equation} is satisfied
\begin{align}
\A \, \psi &= \lam \, \psi . \label{eq:EigenvalueEquation}
\end{align}
A function $\psi$ that solves \eqref{eq:EigenvalueEquation} is called an \emph{eigenfunction} of $\A$ corresponding to $\lam$.  The \emph{multiplicity} of an eigenvalue $\lam$ is the number of linearly independent eigenfunctions for which equation \eqref{eq:EigenvalueEquation} is satisfied.  The spectrum of an operator $\A$ can be decomposed into two disjoint sets called the \emph{discrete} and \emph{essential}
\footnote{  
The essential spectrum may be further decomposed into the \emph{continuous} spectrum and the \emph{residual} spectrum.  It can be shown that the residual spectrum of an ordinary differential operator is empty (see \cite{roach-1}, page 184).}
spectra:  $\sig(\A)=\sig_d(\A) \cup \sig_e(\A)$.  For a normal operator $\A$, a number $\lam \in \R$ belongs to $\sig_d(\A)$ if and only if $\lam$ is an isolated point of $\sig(\A)$ and $\lam$ is an eigenvalue of finite multiplicity (see \cite{rudin-3}, Theorem 12.29).
\par
A \emph{projection-valued measure} on the measure space $(\R,\B(\R))$ is a family of bounded linear operators $\{ E(B),B \in \B(\R) \}$ in $\H$ that satisfies:
\begin{enumerate}
\item $E(\emptyset)=0$ and $E(\R)=\text{Id}$.
\item $E(B)$ is an orthogonal projection.  That is, $E^2(B)=E(B)$ and $E(B)$ is self-adjoint: $E^*(B)=E(B)$.
\item $E(A \cap B) =  E(A)E(B)$.
\item If $B=\bigcup_{i=1}^\infty B_i$ and $B_i \cap B_j = \emptyset$ for $i \neq j$ then $E(B) = \lim_{n\to\infty}\sum_{i=1}^n E(B_j)$, where the limit is in the strong operator topology.
\item For every in $f,g \in \H$ the set function $\mu_{f,g}(B):=(f,E(B)g)$ is a complex measure on $\B(\R)$.
\end{enumerate}
\begin{theorem}[Spectral Representation Theorem]
\label{thm:spectral}
There is a one-to-one correspondence between self-adjoint operators $\A$ and projection-valued measures $E$ on $\H$, the correspondence being given by
\begin{align}
\A
		&=		\int_{\sig(\A)} \lam \, E( d\lam ) .
\end{align}
If $g(\cdot)$ is a Borel function on $\R$ then
\begin{align}
g(\A)
		&=		\int_{\sig(\A)} g(\lam) E( d\lam ) , &
{{\rm Dom}}(g(\A))
		&=		\{ f \in \H : \int_{\sig(\A)} |g(\lam)|^2 \mu_{f,f}(d\lam ) < \infty \}. \label{eq:g.A}
\end{align} 
\end{theorem}
\begin{proof}
See \cite{rudin-3} Theorems 12.21 and 13.33.
\end{proof}
\noindent
As a practical matter, if $\A$ is a differential operator acting on a Hilbert space $L^2(I,\m(x)dx)$, where $I$ is an interval with endpoints $l<r$, then the operators defined by \eqref{eq:g.A} can be constructed by solving the \emph{proper} and \emph{improper}
\footnote{The term ``improper'' is used because the improper eigenvalues $\lam \notin \sig_d(\A)$ and the improper eigenfunctions $\psi_\lam \notin \H$ since $\(\psi_\om,\psi_\om\)=\infty$.}
eigenvalue problems
\begin{align}
\text{proper:}&&
\A \, \psi_n
		&= 		\lam_n \, \psi_n , &
\lam_n 
		&\in 	\sig_d(\A) , &
\psi_n
		&\in	\H , \label{eq:proper} \\
\text{improper:}&&
\A \, \psi_\om
		&= 		\lam_\om \, \psi_\om , &
\lam_\om 
		&\in 	\sig_e(\A) , &
\psi_\om
		&\notin	\H . \label{eq:improper}
\end{align}
For the improper eigenvalue problem one extends the domain of $\A$ to include functions all functions $f$ for which the following boundedness conditions are satisfied
\begin{align} 
\lim_{x \searrow l} |f(x)|^2 \m(x)
		&< \infty , &
\lim_{x \nearrow r} |f(x)|^2 \m(x)
		&< \infty , \label{eq:bound}
\end{align}
We will use Latin subscripts (e.g., $l,m,n$) to denote proper eigenfunctions/values and Greek subscripts (e.g., $\om,\nu,\mu$) to denote improper eigenfunctions/values.  When we do not wish to distinguish between proper and improper eigenfunctions/values we will write $\psi_\lam$ and $\lam$ with no subscript.  We refer to $\psi_\lam$ and $\lam$ as \emph{generalized} eigenfunctions/values.
\par
After normalizing, the proper and improper eigenfunctions $\A$ satisfy the following orthogonality relations
\begin{align}
\( \psi_n, \psi_m \)
		&=	\del_{n,m} , & 
\( \psi_\om, \psi_{\mu} \)
		&=	\del(\om - \mu) , &
\( \psi_n, \psi_\om \) 
		&=	0 .
\end{align}
The operator $g(\A)$ in \eqref{eq:g.A} is constructed as follows (see \cite{hanson-yakovlev-1}, Section 5.3.2)
\begin{align}
g(\A) f
		&=	\sum_\lam g(\lam) (\psi_\lam,f) \psi_\lam \\
		&:=	\sum_{\lam_n \in \sig_d(\A)}	g(\lam_n) \( \psi_n, f \) \psi_{n} +
				\int_{\lam_\om \in \sig_e(\A)} g(\lam_\om) \( \psi_\om, f \) \psi_\om d\om . \label{eq:operator}
\end{align}
It is not always easy to evaluate divergent integrals of the form $\( \psi_\lam, \psi_{\lam'} \)$ and verify that they are in fact delta functions $\del(\lam - \lam')$.  A method for directly obtaining properly normalised improper eigenfunctions can be found on page 238 of \cite{friedman-1}.

\clearpage

\begin{figure}
\centering
\begin{tabular}{|c|c|}
\hline
\includegraphics[width=.475\textwidth]{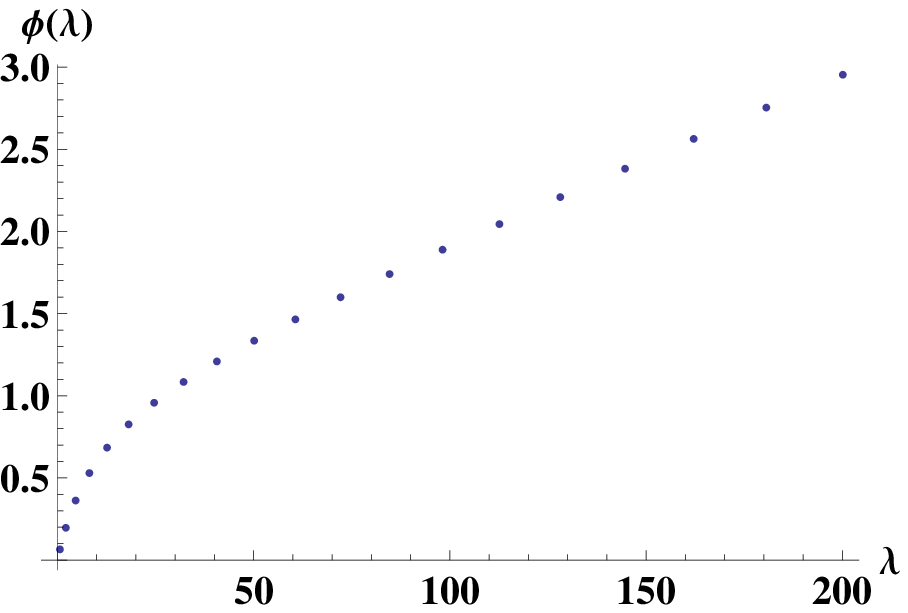} &
\includegraphics[width=.475\textwidth]{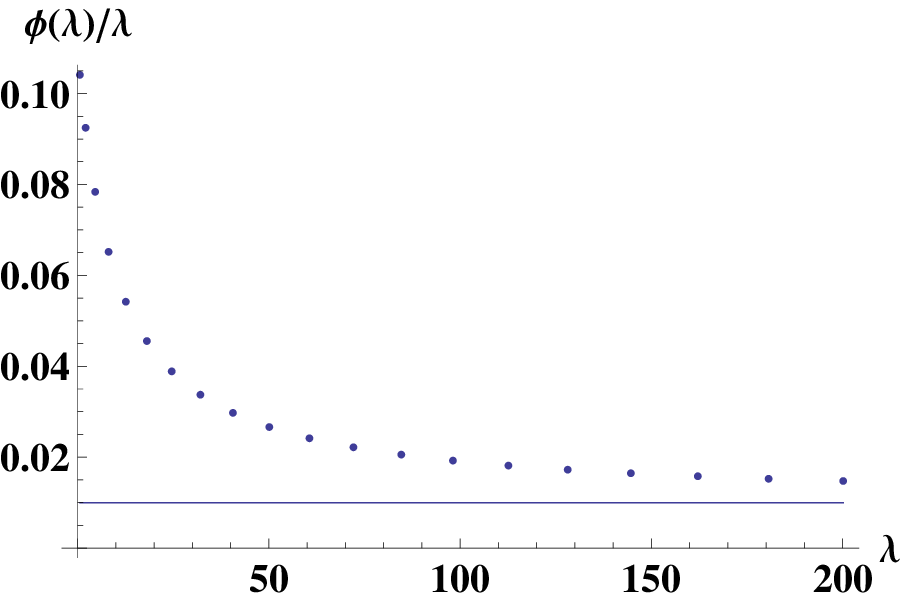} \\ \hline
\includegraphics[width=.475\textwidth]{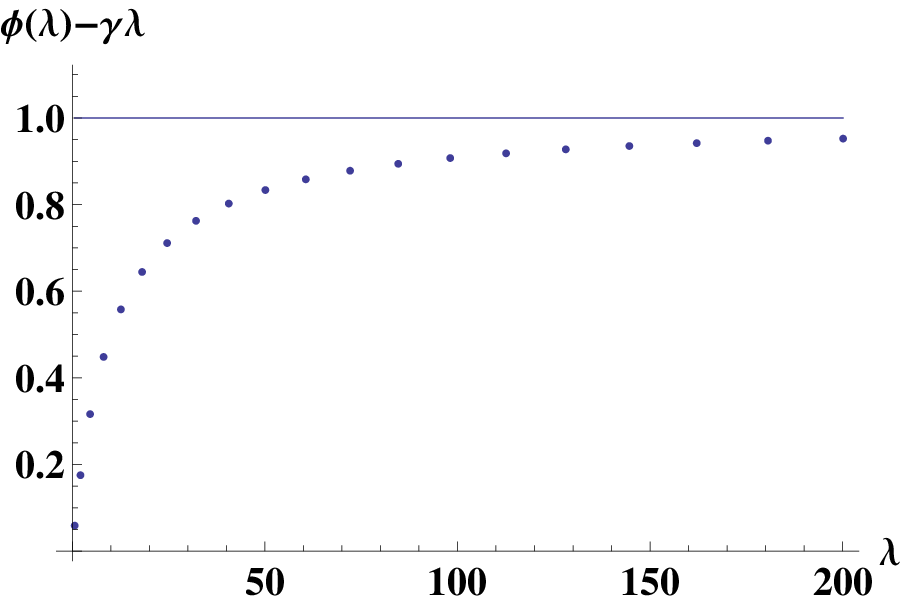} &
\includegraphics[width=.475\textwidth]{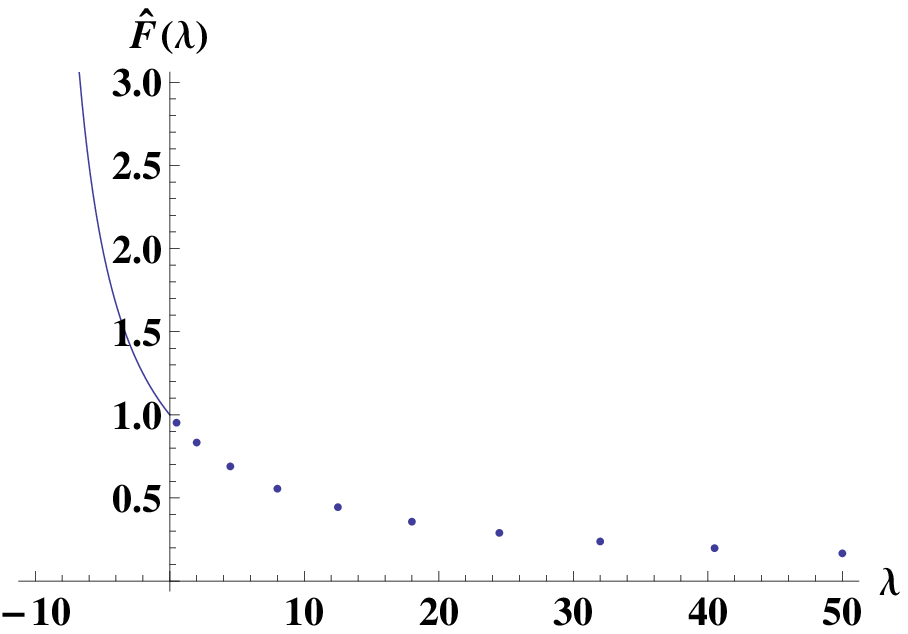} \\
\hline
\end{tabular}
\caption{A graphical illustration of how to obtain the drift $\gamma$, the jump intensity $\alpha$ and the Laplace transform $\Fh(\lam)$ of the jump distribution $F(ds)$ for a subordinator of the compound Poisson type from a discrete set of values of the subordinator's Laplace exponent $\phi(\lam)$.  Top left: a plot of $\phi(\lam)$ for a L\'{e}vy subordinator with exponentially distributed jumps $\nu(ds)=\alpha \eta e^{-\eta s} ds$.  Top right: a plot of $\phi(\lam)/\lam$ for the same subordinator.  The level of the solid line corresponds to drift of the subordinator $\gamma$.  Bottom Left: a plot of $\phi(\lam)-\gamma\lam$ for the same subordinator.  The level of the solid line corresponds to the net jump intensity $\alpha$.  Bottom right: a plot of $\Fh(\lam)$ and its analytic continuation for values of $\lam<0$.  In all four plots we use parameter values $\gamma=1/100$, $\alpha=1$ and $\eta=10$.}
\label{fig:phi}
\end{figure}

\begin{figure}
\centering
\begin{align}
\text{AAPL}
\end{align}
\begin{tabular}{ | c | c | }
\hline
$\phi(\lambda)$ & $\phi(\lambda)/\lambda$ \\
\includegraphics[width=0.37\textwidth,height=.185\textheight]{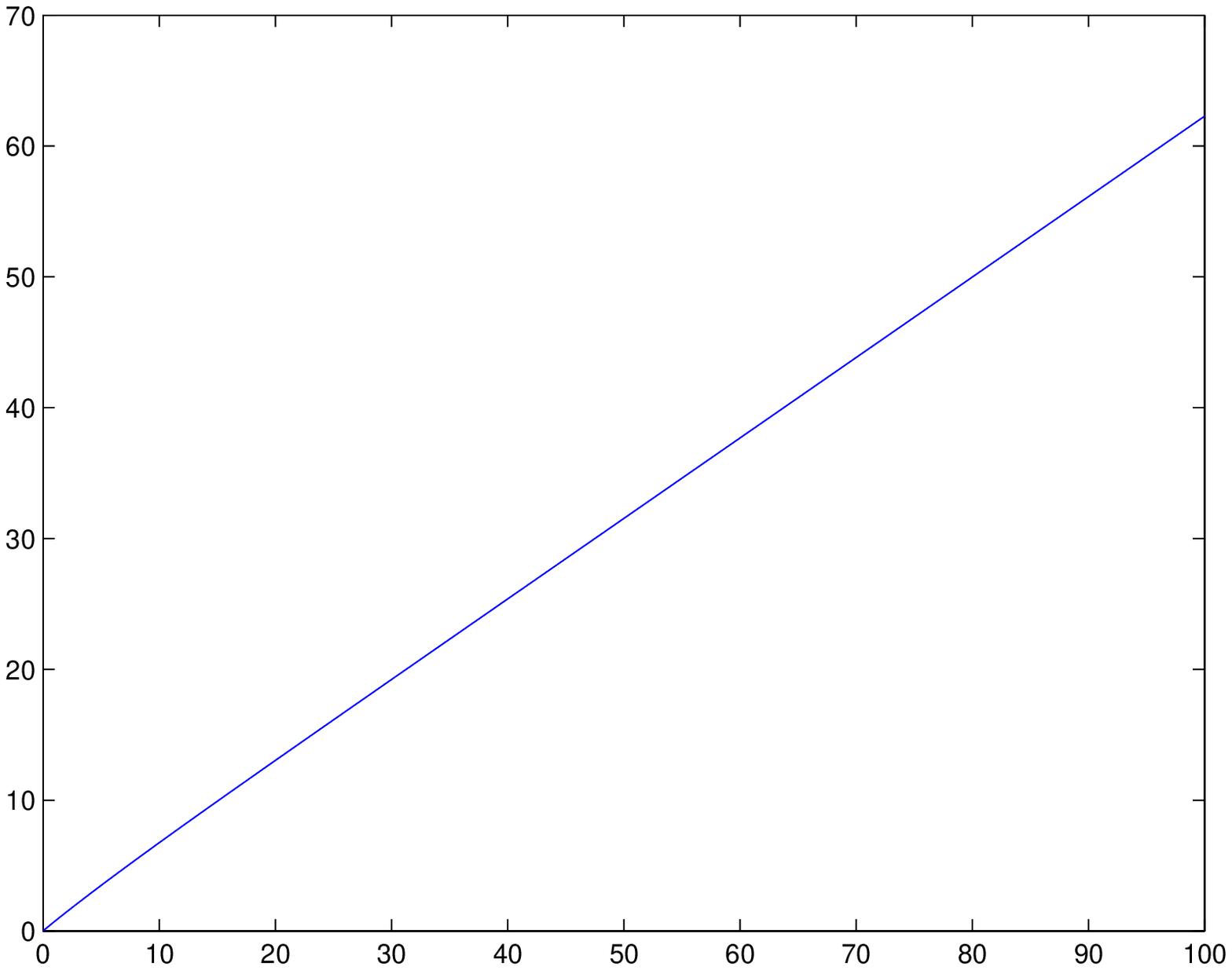} &
\includegraphics[width=0.37\textwidth,height=.185\textheight]{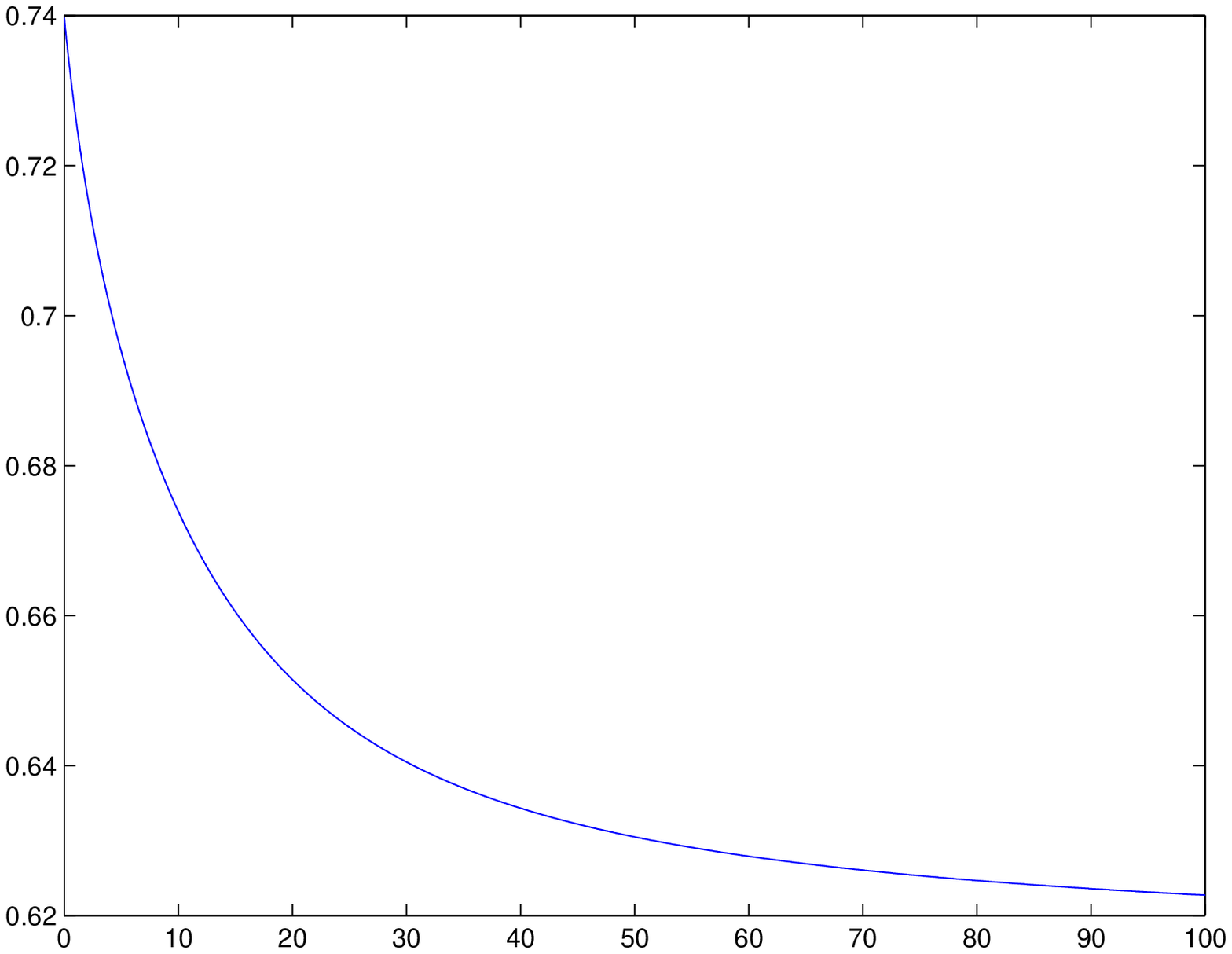} \\ \hline
$\phi(\lambda)-\gamma\lambda$ & $\widehat{F}(\lam)$ \\
\includegraphics[width=0.37\textwidth,height=.185\textheight]{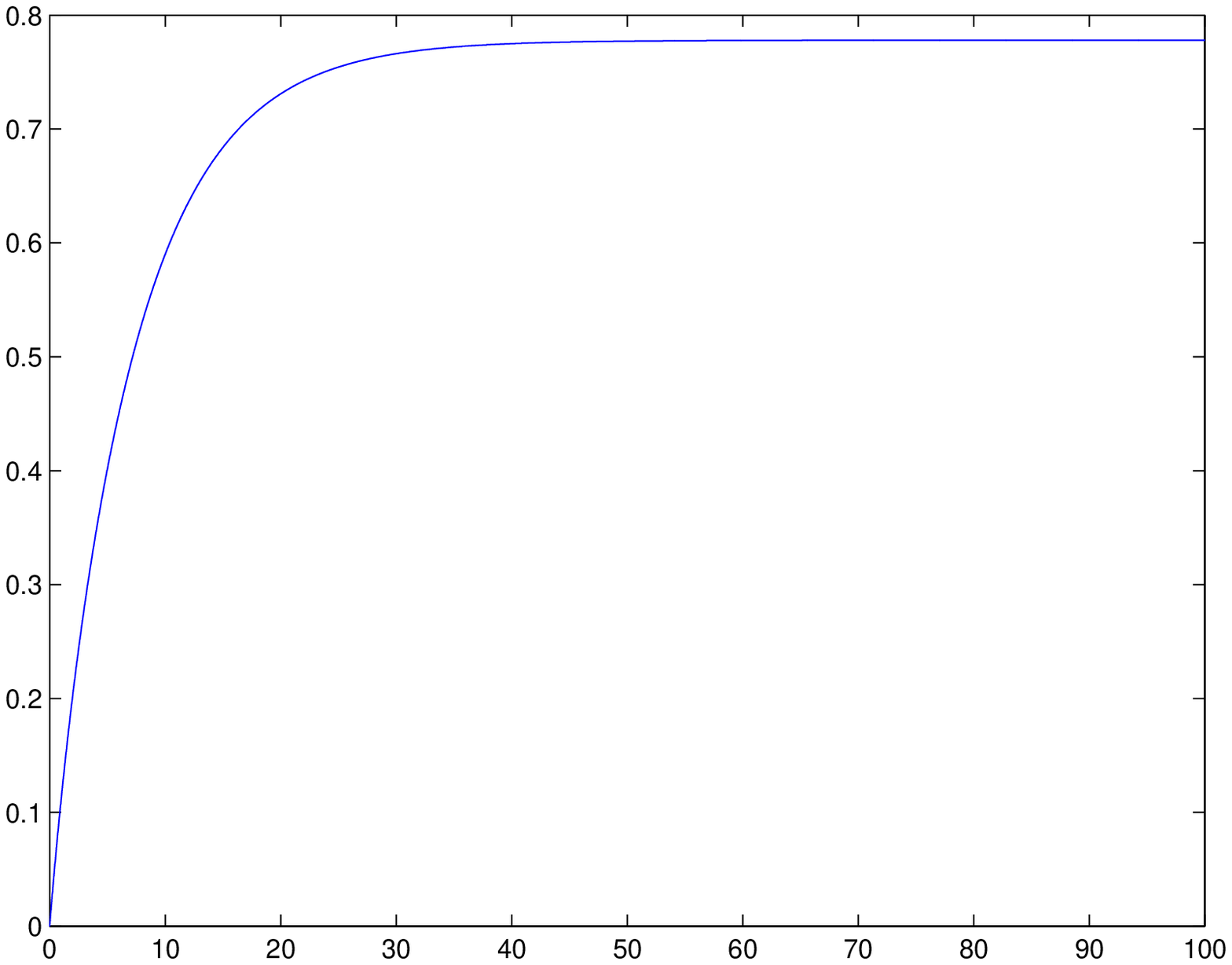} &
\includegraphics[width=0.37\textwidth,height=.185\textheight]{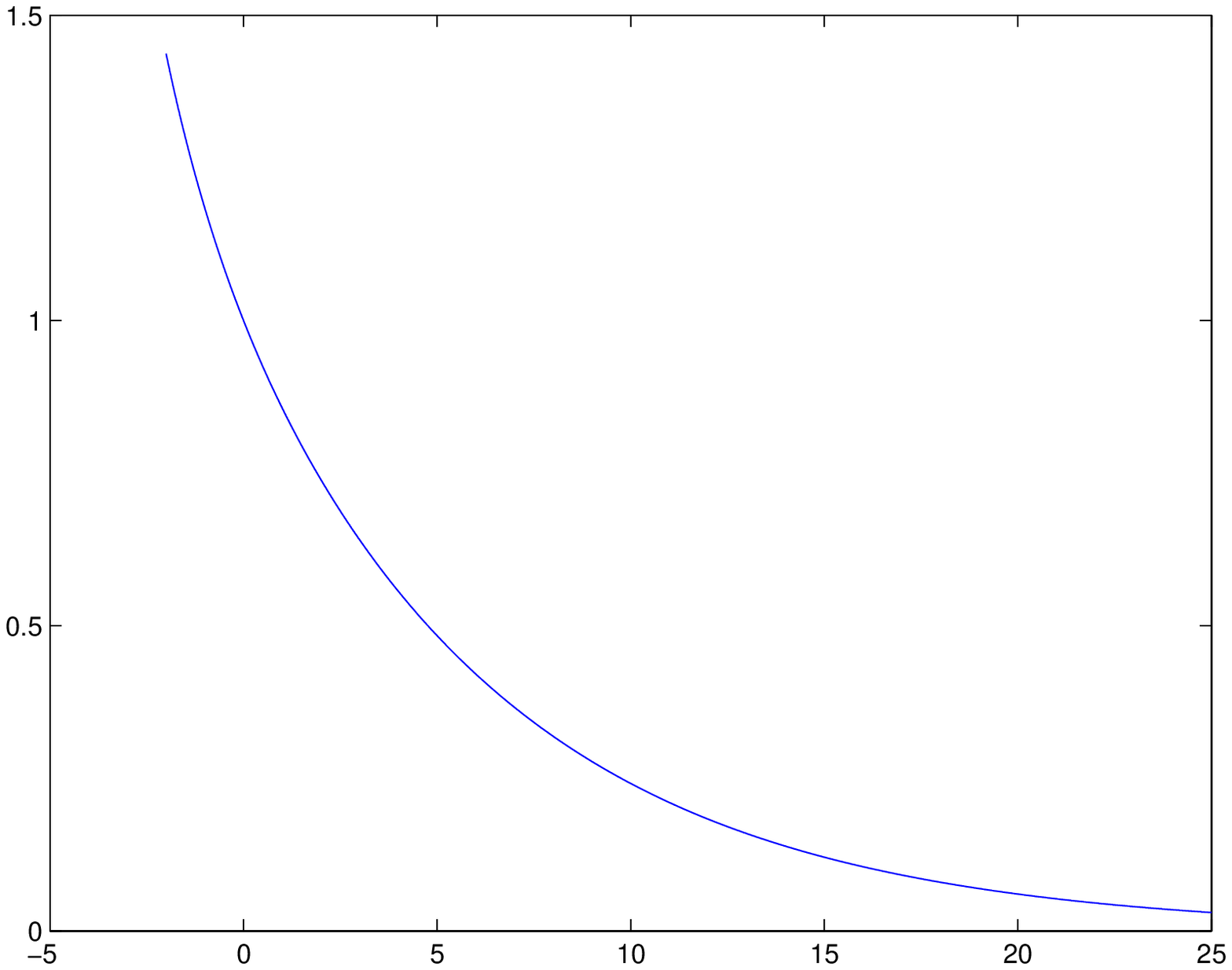} \\ \hline
\end{tabular}
\begin{align}
\text{FB}
\end{align}
\begin{tabular}{ | c | c | }
\hline
$\phi(\lambda)$ & $\phi(\lambda)/\lambda$ \\
\includegraphics[width=0.37\textwidth,height=.185\textheight]{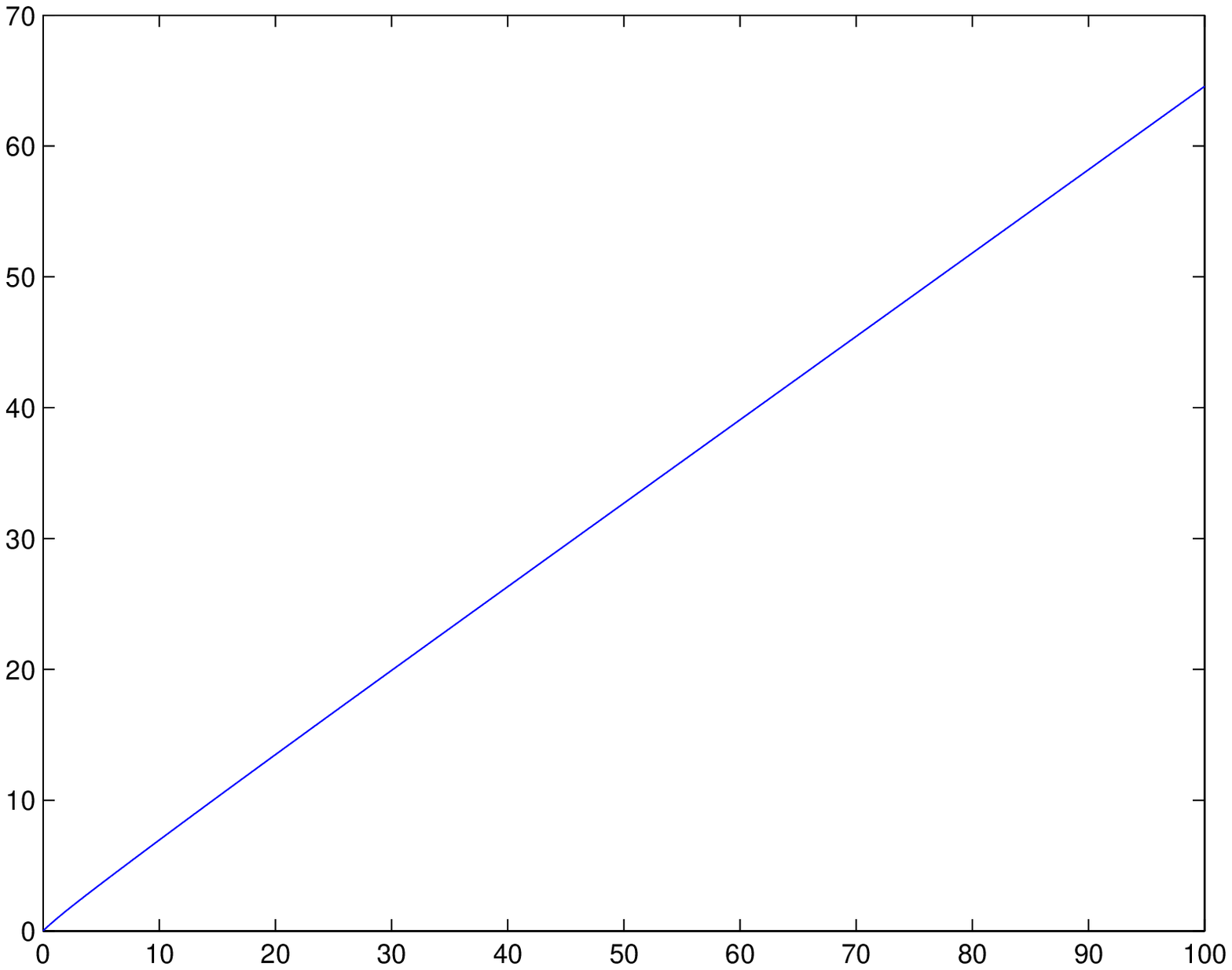} &
\includegraphics[width=0.37\textwidth,height=.185\textheight]{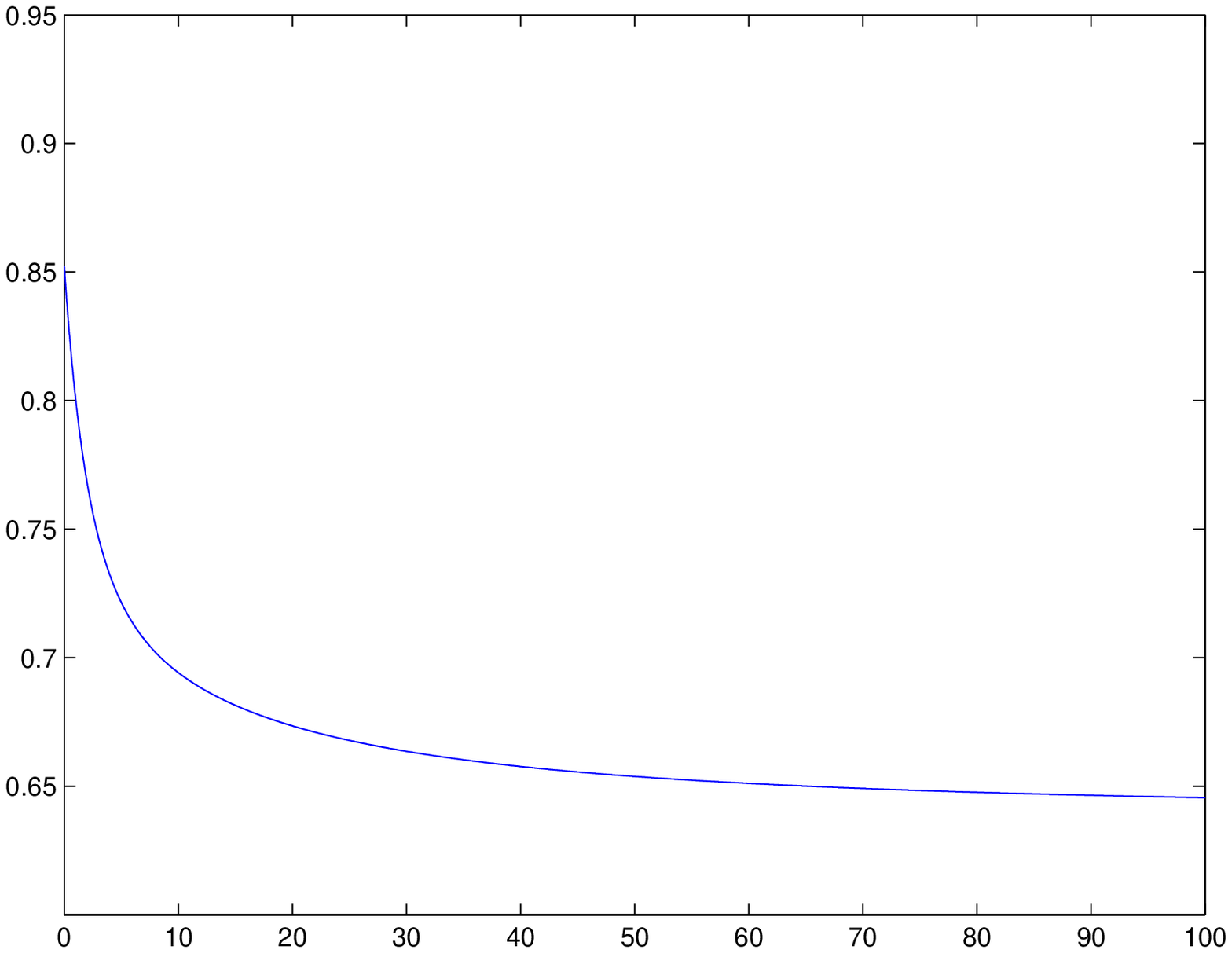} \\ \hline
$\phi(\lambda)-\gamma\lambda$ & $\widehat{F}(\lam)$ \\
\includegraphics[width=0.37\textwidth,height=.185\textheight]{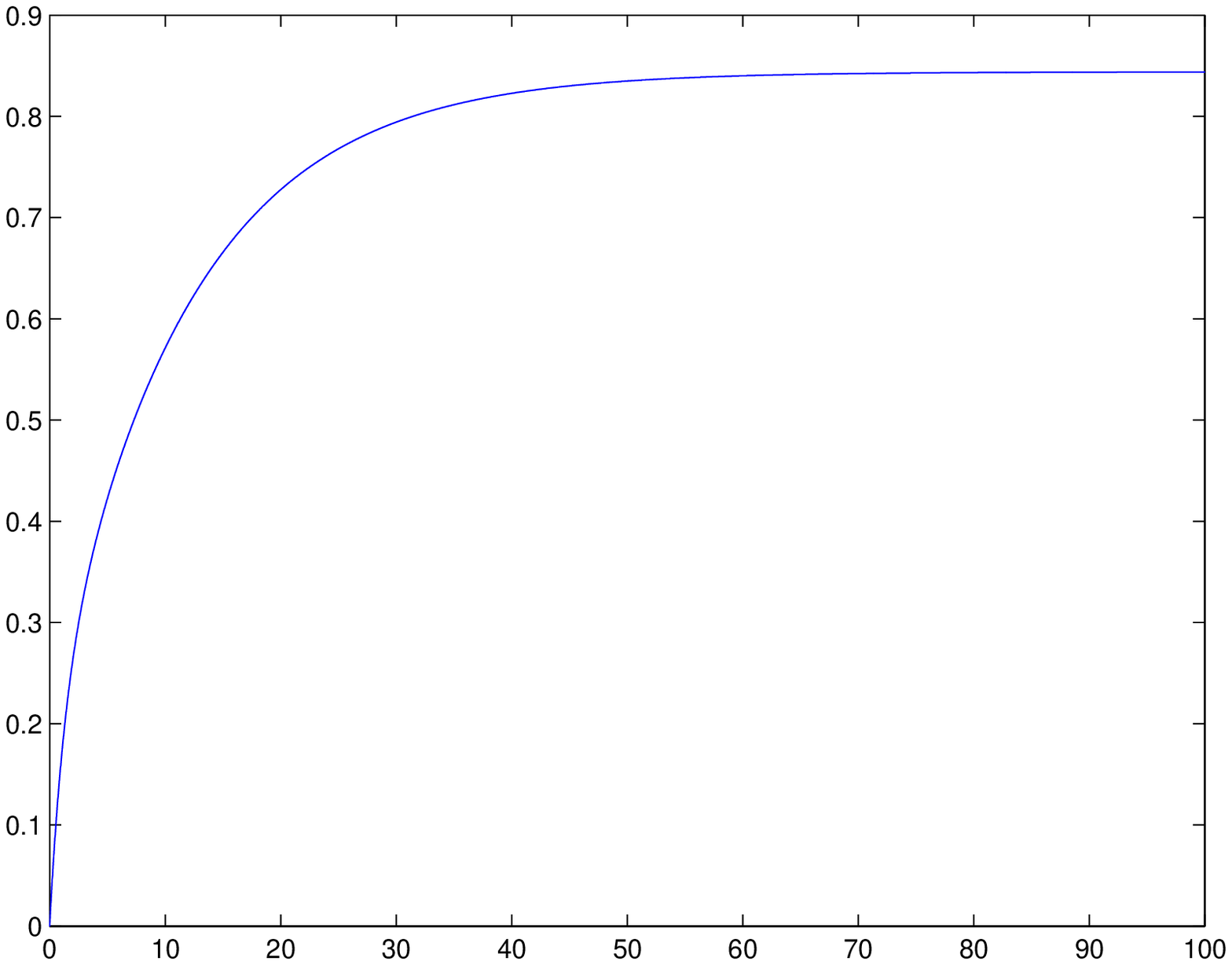} &
\includegraphics[width=0.37\textwidth,height=.185\textheight]{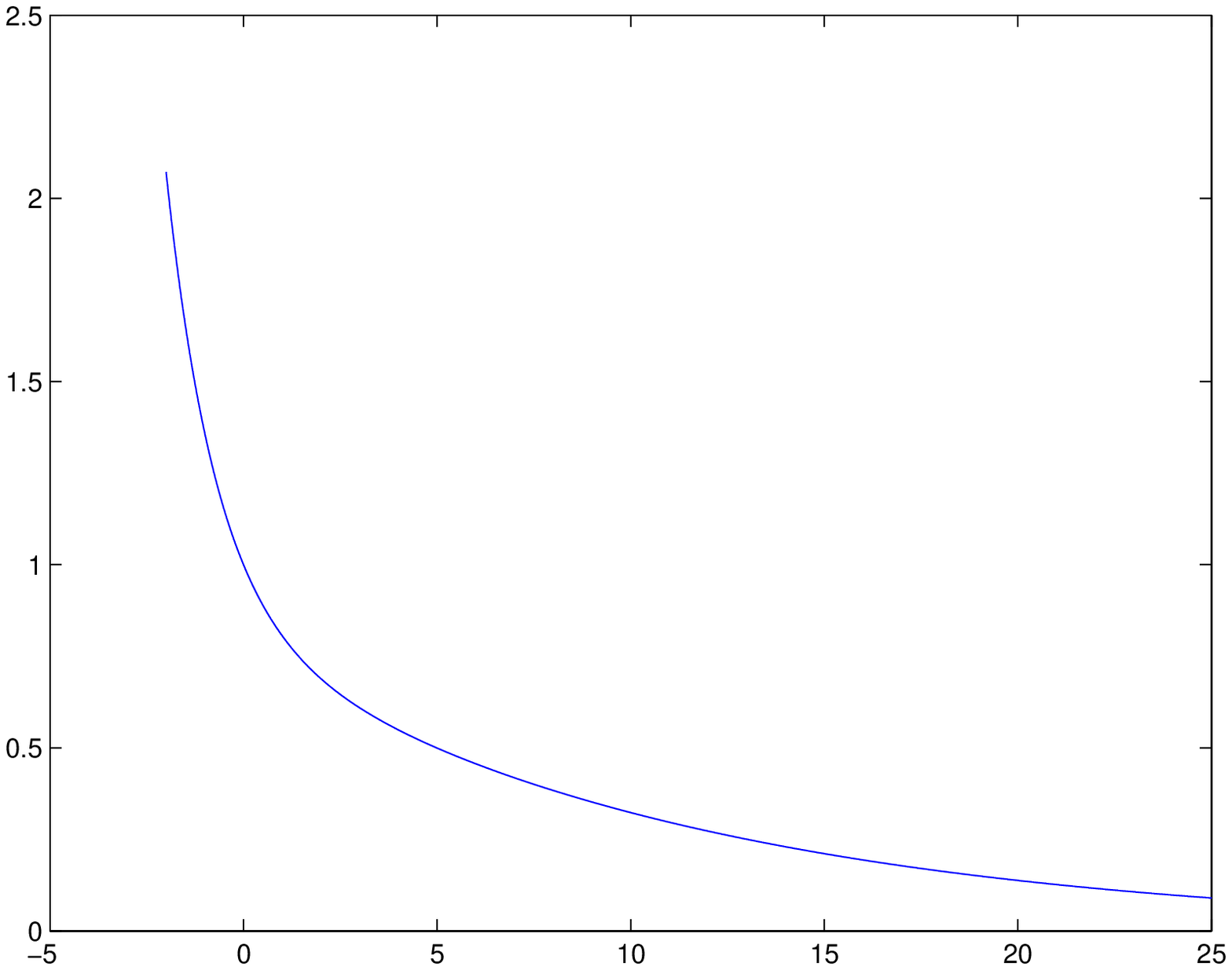} \\ \hline
\end{tabular}
\caption{Market implied $\phi(\lambda)$ from observed options in April 12, 2013 maturing in 19 days. Recall that (in the case of a finite activity subordinator) $\lim_{\lambda\rightarrow\infty}\phi(\lambda)/\lambda=\gamma$, $\lim_{\lambda\rightarrow\infty}\phi(\lambda)-\gamma\lambda=\nu([0,\infty])=\alpha$ and $\widehat{F}(\lam)=1+(\gamma \lambda - \phi(\lam))/\alpha$.}
\label{fig:MarketImpliedPhiLambda}
\end{figure}

\begin{figure}
\centering
\begin{align}
\text{GOOG}
\end{align}
\begin{tabular}{ | c | c | }
\hline
$\phi(\lambda)$ & $\phi(\lambda)/\lambda$ \\
\includegraphics[width=0.37\textwidth,height=.185\textheight]{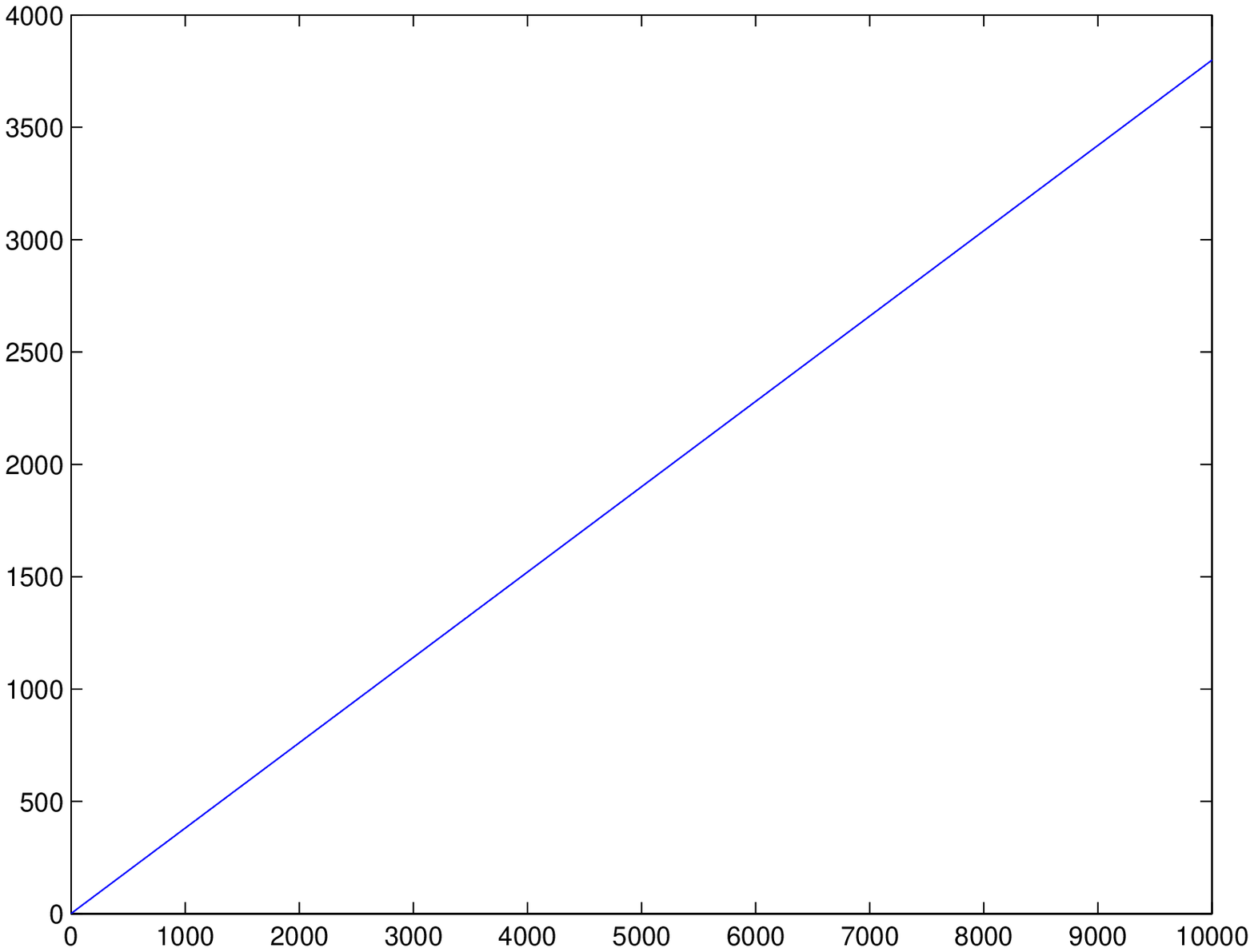} &
\includegraphics[width=0.37\textwidth,height=.185\textheight]{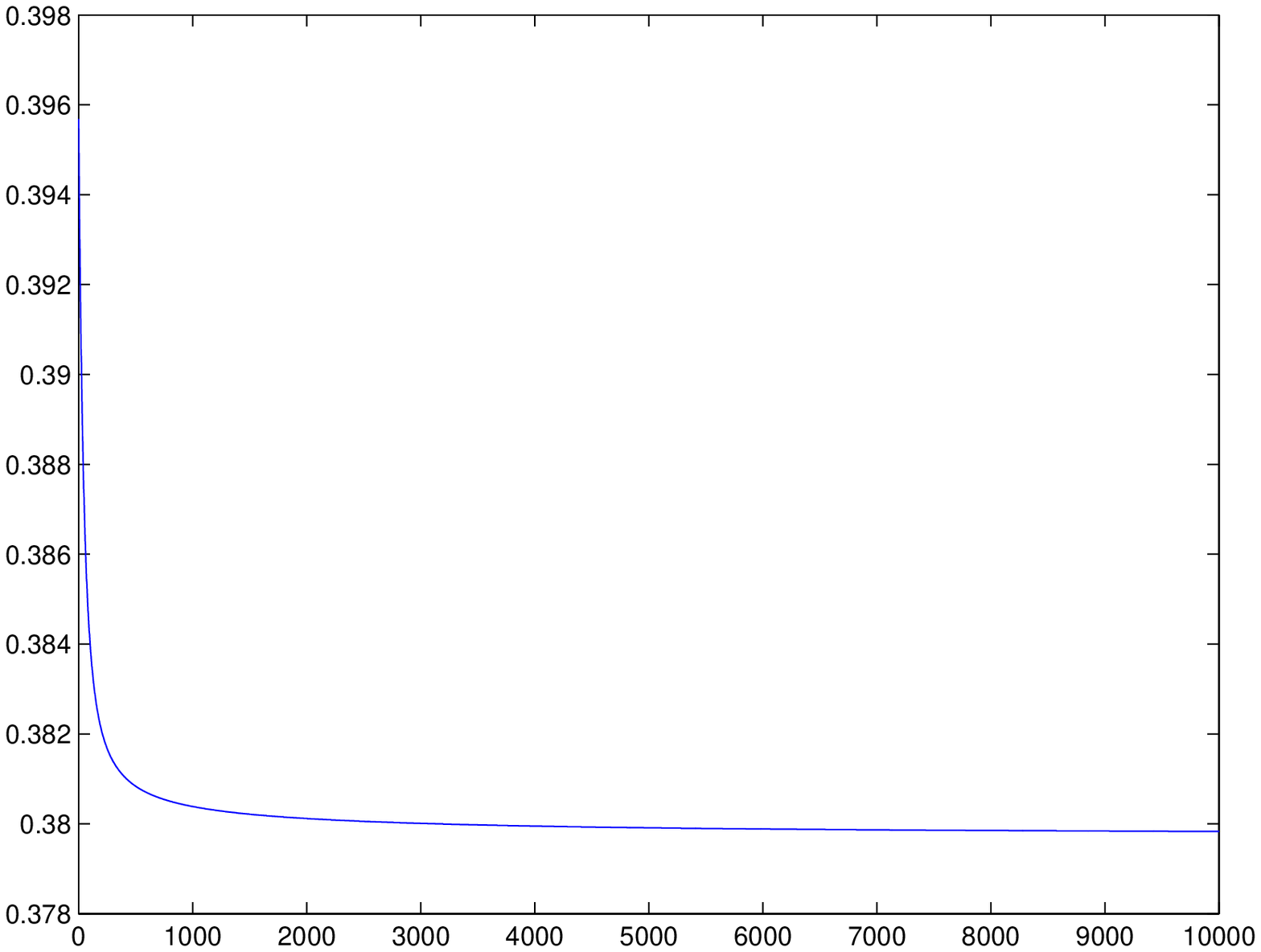} \\ \hline
$\phi(\lambda)-\gamma\lambda$ & $\widehat{F}(\lam)$ \\
\includegraphics[width=0.37\textwidth,height=.185\textheight]{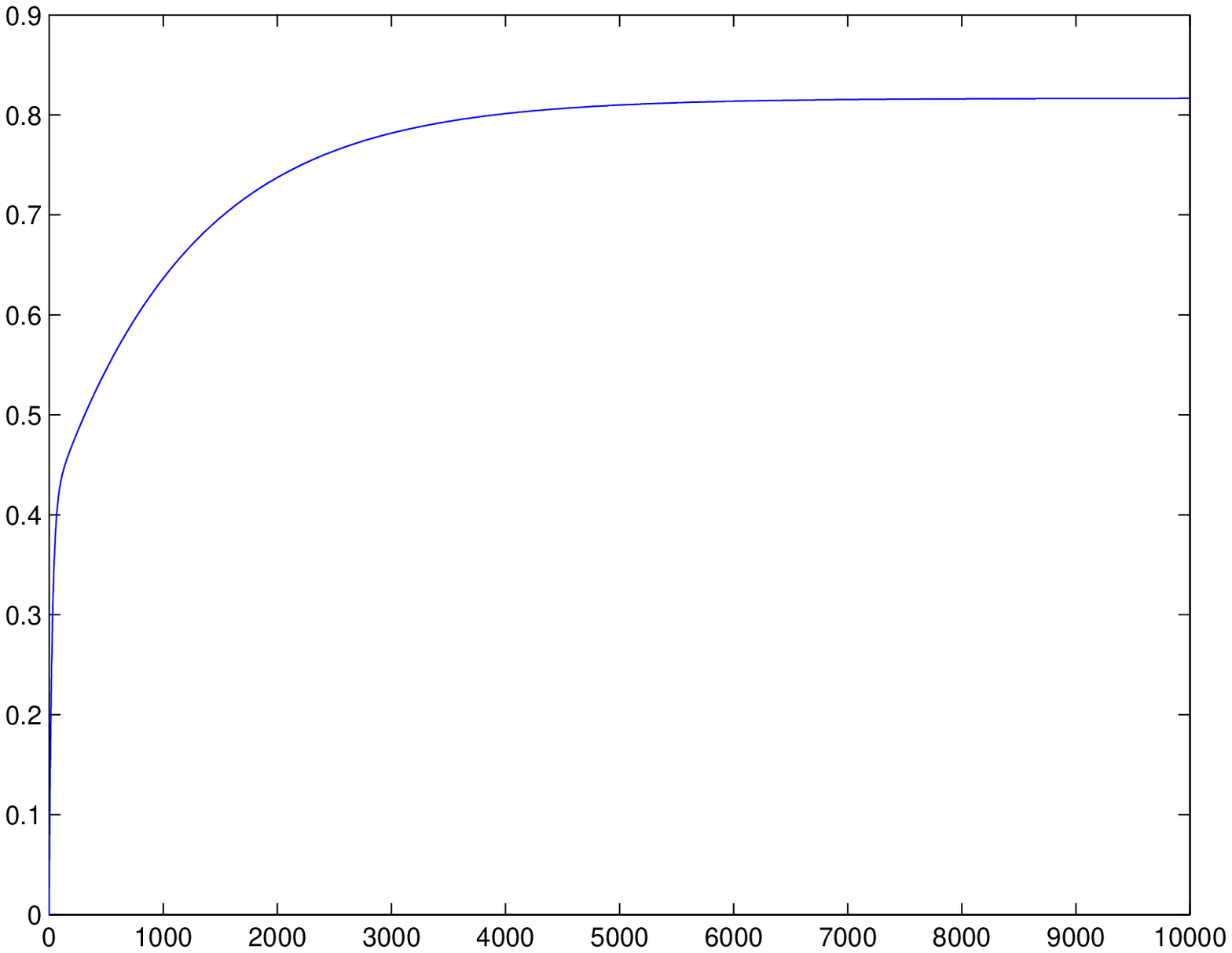} &
\includegraphics[width=0.37\textwidth,height=.185\textheight]{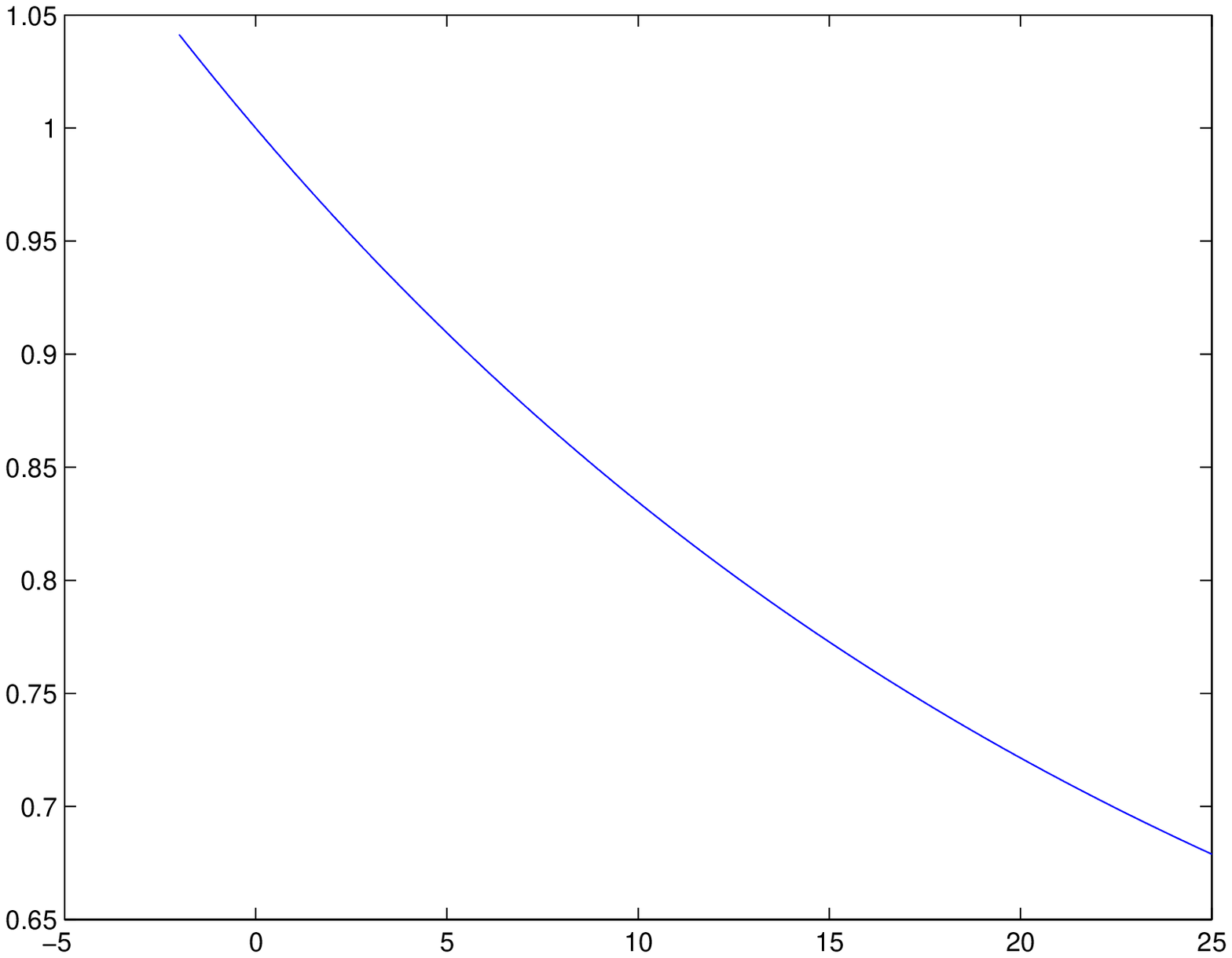} \\ \hline
\end{tabular}
\begin{align}
\text{MFST}
\end{align}
\begin{tabular}{ | c | c | }
\hline
$\phi(\lambda)$ & $\phi(\lambda)/\lambda$ \\
\includegraphics[width=0.37\textwidth,height=.185\textheight]{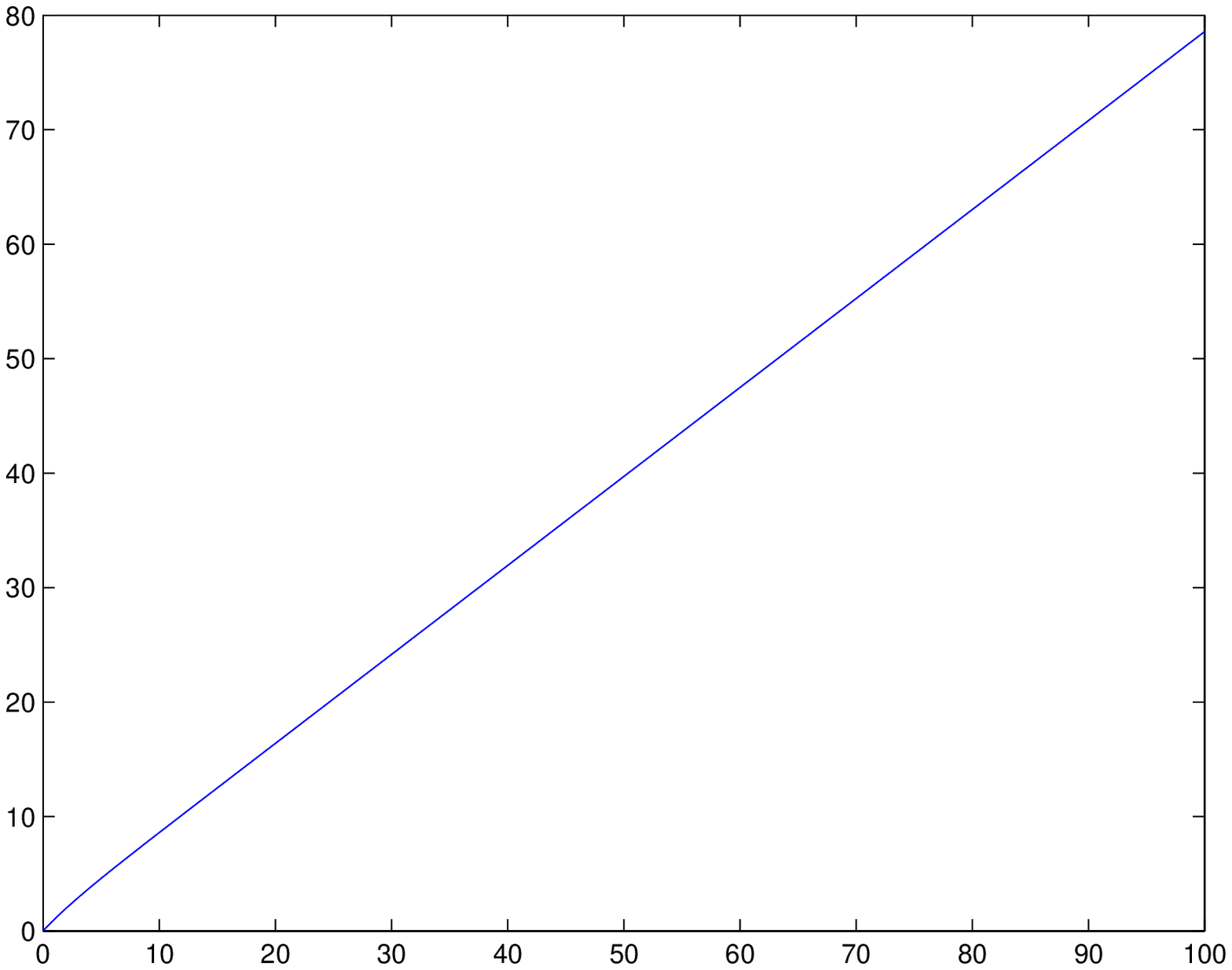} &
\includegraphics[width=0.37\textwidth,height=.185\textheight]{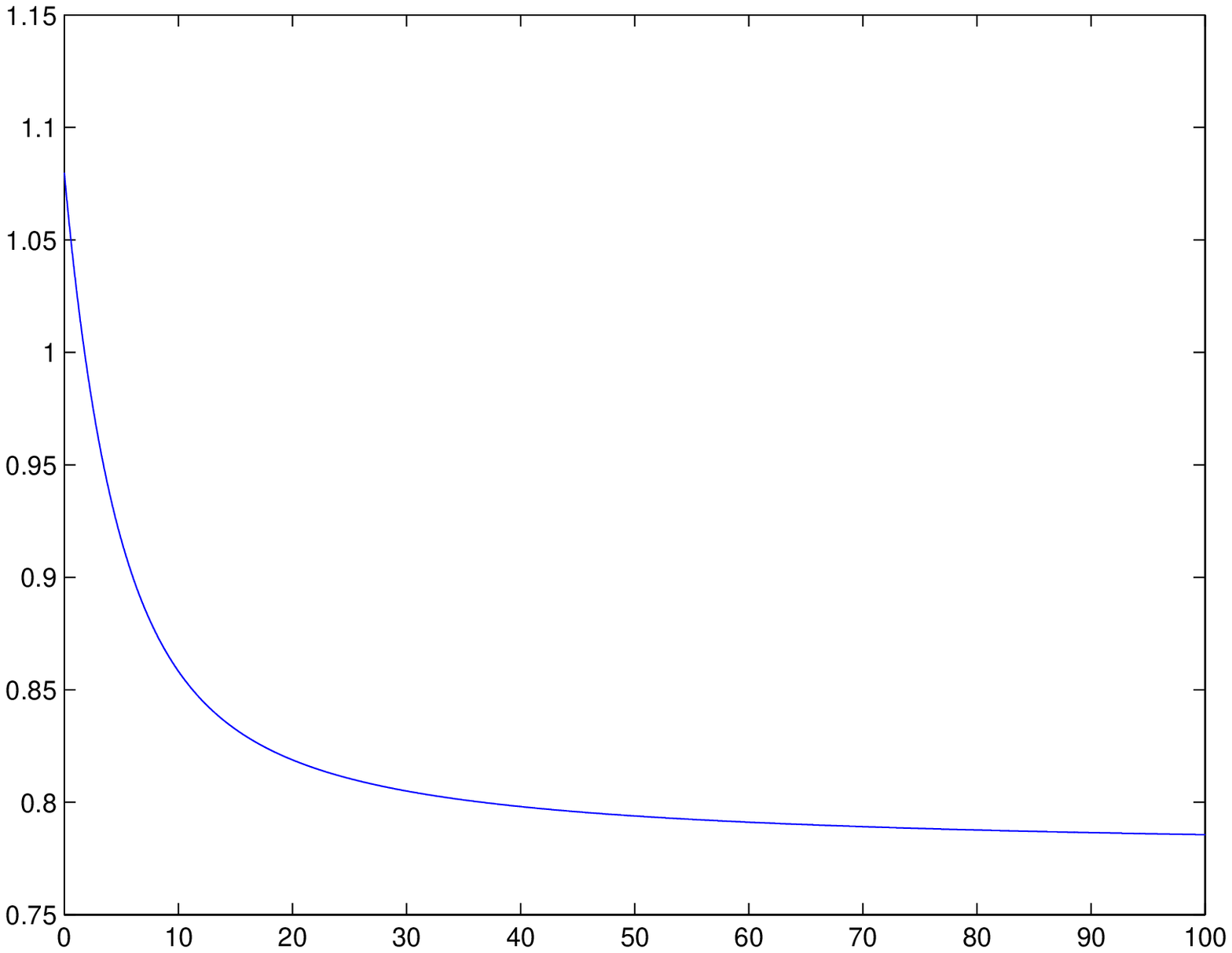} \\ \hline
$\phi(\lambda)-\gamma\lambda$ & $\widehat{F}(\lam)$ \\
\includegraphics[width=0.37\textwidth,height=.185\textheight]{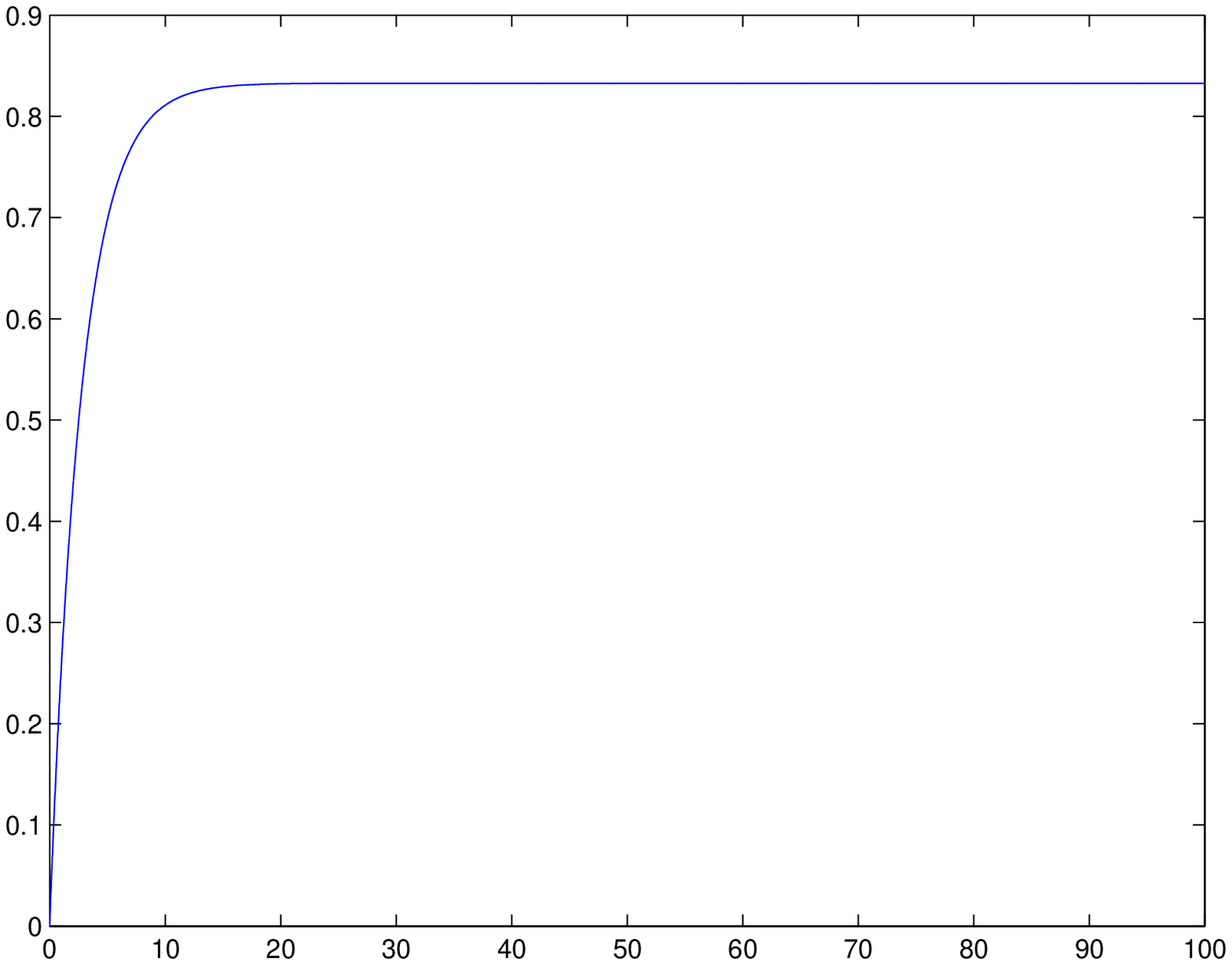} &
\includegraphics[width=0.37\textwidth,height=.185\textheight]{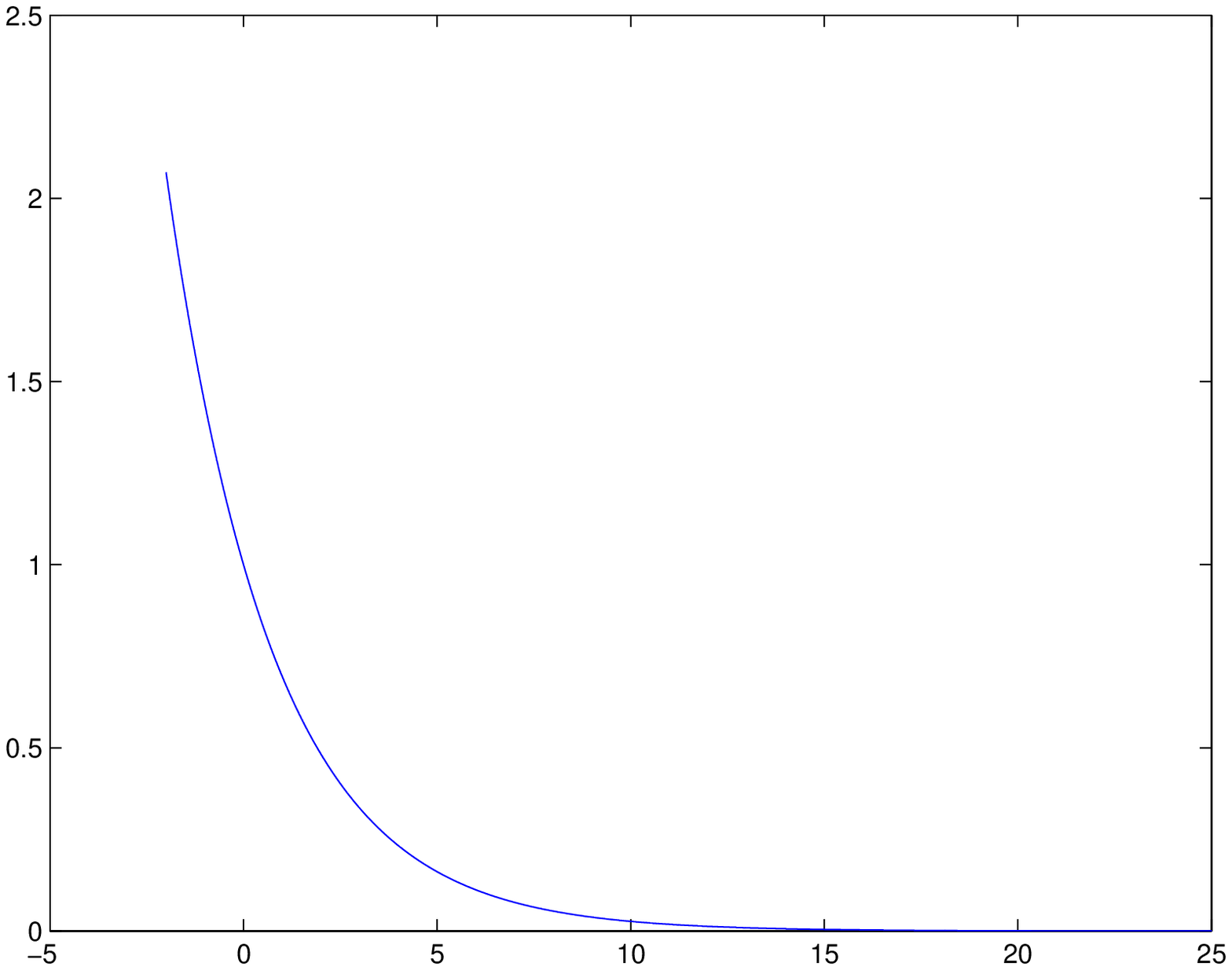} \\ \hline
\end{tabular}
\caption{Market implied $\phi(\lambda)$ from observed options in April 12, 2013 maturing in 19 days. Recall that (in the case of a finite activity subordinator) $\lim_{\lambda\rightarrow\infty}\phi(\lambda)/\lambda=\gamma$, $\lim_{\lambda\rightarrow\infty}\phi(\lambda)-\gamma\lambda=\nu([0,\infty])=\alpha$ and $\widehat{F}(\lam)=1+(\gamma \lambda - \phi(\lam))/\alpha$.}
\label{fig:MarketImpliedPhiLambda2}
\end{figure}

%
%

\clearpage
\bibliographystyle{chicago}
\bibliography{RafaBib}

\end{document}